\documentclass[letterpaper,twocolumn,10pt]{article}
\usepackage{usenix}

\usepackage{tikz}
\usepackage{amsmath}

\usepackage{algorithm,algorithmic}

\usepackage{microtype}
\usepackage{graphicx}
\usepackage{subcaption}
\usepackage{booktabs}

\providecommand{\citep}{\cite}
\providecommand{\citet}{\cite}

\usepackage{algorithm,algorithmic}

\usepackage{tikz}
\usetikzlibrary{arrows.meta, positioning}

\usepackage{xspace}
\usepackage{xcolor}
\definecolor{darkgreen}{RGB}{0,100,0}   
\definecolor{darkred}{RGB}{139,0,0}     

\newcommand{\dtrain}{\mathcal D_\mathrm{train}}

\newcommand{\E}{\mathbb{E}}

\newcommand{\cZ}{\mathcal{Z}}
\newcommand{\cI}{\mathcal{I}}

\newcommand{\cO}{\mathcal O}

\newcommand{\R}{\mathbb R}

\newcommand{\PLRV}{\mathrm{PLRV}}

\newcommand{\cE}{\mathcal{E}}

\newcommand{\eqdef}{\stackrel{\text{def}}{=}}

\def\<#1,#2>{\langle #1,#2\rangle}

\usepackage{dsfont}

\def\<{\langle}
\def\>{\rangle}
\def\|{\Vert}

\def\eps{\varepsilon}

\newcommand{\esp}[1]{\mathbb{E}\left[#1\right]}

\newcommand{\NRM}[1]{{{\left\| #1\right\|}}} 
\newcommand{\set}[1]{{\left\{ #1\right\}}} 
\newcommand{\proba}[1]{\mathbb{P}\left(#1\right)}

\newcommand{\Tot}{{\mathrm{Tot}}}
\newcommand{\DS}{{\mathrm{DS}}}

\renewcommand{\P}{\mathbb{P}}
\newcommand{\Q}{\mathbb{Q}}
\newcommand{\cA}{\mathcal{A}}
\newcommand{\cB}{\mathcal{B}}

\newcommand{\cX}{\mathcal{X}}

\newcommand{\cM}{\mathcal{M}}

\newcommand{\dd}{{\rm d}}

\newcommand{\cP}{\mathcal{P}}
\newcommand{\cH}{\mathcal{H}}

\newcommand{\N}{\mathbb{N}}
\newcommand{\cD}{\mathcal{D}}
\newcommand{\one}{\mathds{1}}

\usepackage{amsfonts}
\usepackage{amsthm}
\usepackage{amsmath}
\usepackage{amssymb}
\usepackage{mathrsfs}
\usepackage{amscd}
\usepackage{caption}
\usepackage{cleveref}

\newcommand{\pano}{\textsc{Panoramia}\xspace}

\usepackage{hyperref}

\usepackage{graphicx}
\usepackage[T1]{fontenc}
\usepackage[utf8]{inputenc}
\usepackage[english]{babel}
\usepackage{amsmath,amsthm,amssymb,amsfonts}
\usepackage{bm}

\usepackage{booktabs} 
\usepackage{threeparttable}
\usepackage{makecell}
\usepackage{multirow}

\usepackage[nice]{nicefrac}
\usepackage{pifont}
\usepackage[colorinlistoftodos,textwidth=2.1cm]{todonotes}

\usepackage{tcolorbox}
\usepackage{pifont}
\definecolor{mydarkgreen}{RGB}{39,130,67}
\definecolor{mydarkred}{RGB}{192,25,25}
\definecolor{mydarkblue}{RGB}{0,0,140}
\definecolor{lightorange}{RGB}{252,210,153}

\usepackage[table]{xcolor} 
\usepackage{booktabs}

\newcommand{\subtlerule}{%
    \arrayrulecolor{black!12}%
    \specialrule{0.25pt}{1.5pt}{1.5pt}%
    \arrayrulecolor{black}%
}

\usepackage{empheq}
\definecolor{darkgreen}{rgb}{0.00,0.5,0.00}

\usepackage{tcolorbox}
\usepackage{framed}

\usepackage{amsmath,amsthm,amssymb,amsfonts}
\usepackage{bm}
\usepackage{aliascnt}

\usepackage{hyperref}

\theoremstyle{plain}

\newtheorem{theorem}{Theorem}[section]

\newaliascnt{proposition}{theorem}
\newtheorem{proposition}[proposition]{Proposition}
\aliascntresetthe{proposition}

\newaliascnt{lemma}{theorem}
\newtheorem{lemma}[lemma]{Lemma}
\aliascntresetthe{lemma}

\newaliascnt{corollary}{theorem}
\newtheorem{corollary}[corollary]{Corollary}
\aliascntresetthe{corollary}

\theoremstyle{definition}

\newaliascnt{definition}{theorem}
\newtheorem{definition}[definition]{Definition}
\aliascntresetthe{definition}

\newaliascnt{assumption}{theorem}
\newtheorem{assumption}[assumption]{Assumption}
\aliascntresetthe{assumption}

\newaliascnt{example}{theorem}

\aliascntresetthe{example}

\theoremstyle{remark}

\newaliascnt{remark}{theorem}
\newtheorem{remark}[remark]{Remark}
\aliascntresetthe{remark}

\usepackage{booktabs}
\usepackage{tabularx}
\usepackage{array}

\usepackage{cleveref}

\crefname{theorem}{theorem}{theorems}
\Crefname{theorem}{Theorem}{Theorems}

\crefname{proposition}{proposition}{propositions}
\Crefname{proposition}{Proposition}{Propositions}

\crefname{lemma}{lemma}{lemmas}
\Crefname{lemma}{Lemma}{Lemmas}

\crefname{corollary}{corollary}{corollaries}
\Crefname{corollary}{Corollary}{Corollaries}

\crefname{definition}{definition}{definitions}
\Crefname{definition}{Definition}{Definitions}

\crefname{assumption}{assumption}{assumptions}
\Crefname{assumption}{Assumption}{Assumptions}

\crefname{example}{example}{examples}
\Crefname{example}{Example}{Examples}

\crefname{remark}{remark}{remarks}
\Crefname{remark}{Remark}{Remarks}
\begin{document}

\date{}

\title{\Large \bf Privacy Auditing with Zero (0) Training Run
}

\author{
{\rm Tudor Cebere\thanks{Equal contribution}}\\
PreMeDICaL team, Inria \\
  Idesp, Inserm, Univ. de Montpellier
\and
{\rm Mathieu Even\footnotemark[1]}\\
PreMeDICaL team, Inria \\
  Idesp, Inserm, Univ. de Montpellier
\and
{\rm Linus Bleistein}\\
School of Computer and Communication \\ Sciences
  and School of Life Sciences (EPFL)
\and
{\rm Aurélien Bellet}\\
PreMeDICaL team, Inria \\
  Idesp, Inserm, Univ. de Montpellier
} 

\maketitle

\begin{abstract}
Privacy auditing provides empirical lower bounds on the differential privacy parameters of learning algorithms. Existing methods, however, require interventional access to the training pipeline, either to retrain models multiple times or to randomize data inclusion. This is often infeasible for large and opaque deployed systems such as foundation or language models. We introduce Zero-Run privacy auditing, a framework for auditing models post-hoc using two fixed datasets: examples known to be training-set members and examples known to be non-members. In this observational regime, membership is no longer randomized; instead, member and non-member data often differ in distribution, so membership inference scores may reflect a distribution shift rather than algorithmic leakage. 
Drawing on ideas from causal inference, we formalize this confounding effect and propose two complementary corrections that yield valid privacy audits.
Our first approach models the combined effect of distribution shift and algorithmic leakage as an adaptive composition, producing conservative global corrections. 
Our second approach conditions on observed data and adjusts pointwise membership guesses, yielding sharper instance-dependent bounds. 
Experiments on synthetic data and large-scale image and text models show that Zero-Run auditing enables practical privacy evaluation when retraining or controlled data insertion is infeasible.
\end{abstract}

\section{Introduction}

As AI systems increasingly process sensitive data, mitigating privacy leakage (i.e., the inadvertent disclosure of training examples) has become a critical priority. This risk is exacerbated by high-capacity foundation models trained across various modalities \citep{touvron2023llama,jiang2023mistral7b, ramesh2021zero, defossez2024moshi, hollmann2023tabpfntransformersolvessmall,dingtabula}. Trained on massive, web-crawled corpora, they are inherently susceptible to memorizing personal or proprietary information \citep{carlini2021extracting,carlini2019secret,lukas2023analyzing}, with the extent of memorization observed to increase 
with model size \citet{carlini2023quantifyingmemorizationneurallanguage}. Moreover, these vulnerabilities introduce challenges in domain-specific contexts; for instance, clinical foundation models memorize rare patient profiles despite rigorous data de-identification \citep{tonekaboni2025an}.

A countermeasure to these risks is \emph{differential privacy} (DP) \citep{DworkMNS06}, the de facto standard for privacy-preserving algorithms. By bounding the impact of any individual's data on the result, DP protects against privacy attacks that infer sensitive information from released outputs. However, despite offering worst-case guarantees over the privacy leakage, its adoption in large-scale training pipelines is limited due to steep utility penalties \cite{jayaraman2019evaluating} and the prohibitive cost of overhauling existing data-processing infrastructure \cite{sinha2025vaultgemma}. 

Motivated by these practical barriers, \emph{privacy auditing} offers an empirical complement to DP by lower-bounding an algorithm's privacy parameters \citep{ding2018detecting,jagielski2020auditing,nasr2023tight,auditing_cebere}. Audits probe a model's behavior on its training data to expose implementation flaws \citep{tramer2022debugging}, identify misconfigured privacy mechanisms \cite{cebere2026privacy}, or demonstrate that theoretical guarantees are tight (i.e., the auditing lower bound closely matches the DP upper bound) \cite{nasr2021adversary, nasr2023tight}.
More broadly, privacy auditing provides an empirical means of assessing privacy even for models trained without formal privacy guarantees.
In contrast to worst-case theoretical upper bounds, auditing quantifies empirical leakage for a specific model architecture, dataset, and threat model, yielding deployment-relevant evidence for privacy claims, compliance, and risk assessment \citep{edpb}.

While traditional privacy auditing requires repeatedly retraining models on randomized datasets \citep{jagielski2020auditing,nasr2023tight,auditing_cebere}, recent advances reduce this requirement to a single \emph{randomized} training run, in which the inclusion of multiple data points is randomized jointly \citep{steinke2023privacy,mahloujifar2024auditing}. 
This substantially reduces the computational burden, but still requires access to the training process itself. For modern AI models, whose training consumes thousands of GPU days and relies on rigid, distributed data pipelines, such access is generally limited to the model provider. Large language models perfectly illustrate this tension. While trained on vast web corpora loaded with sensitive data, their immense scale, opaque data mixtures, and black-box deployments prevent third parties from independently performing such privacy audits.
Consequently, researchers and regulators lack tools to verify privacy guarantees, relying on trust and the voluntary cooperation of model providers. This creates an \textit{accountability gap}: the systems posing the greatest privacy risks are the most insulated from independent privacy evaluation.


\paragraph{Contributions.} To address this gap, we introduce a \emph{Zero-Run} auditing framework that enables auditors to assess the privacy leakage of an algorithm \emph{post-hoc from its release}, without modifying the algorithm or data pipeline.

We ground our approach in a fundamental analogy to causal inference. Existing privacy audits based on randomized data insertion are inherently interventional—analogous to \textit{randomized controlled trials} (RCTs)—because training-set membership is experimentally randomized, allowing these audits to isolate privacy leakage. In contrast, our Zero-Run framework is the \textit{observational} counterpart, estimating privacy leakage from fixed member and non-member datasets. Since membership is not randomized, systematic differences can arise between the distributions of members and non-members, acting as a confounding factor \citep{rosenbaum1983central,wager2024causal}. For the audit to yield a valid measure of privacy leakage, these confounding effects must therefore be accounted for.


Guided by this analogy, our main contributions are as follows:
\begin{enumerate}
    \item \textbf{Formalizing Zero-Run Auditing (\Cref{sec:zero-run-auditing-framework}).} We formalize the post-hoc setting in which an auditor relies on disjoint datasets of known members and non-members to assess a released model's privacy. We demonstrate that naive audits can substantially overestimate privacy leakage: membership may be inferred merely by exploiting distributional shifts between member and non-member data rather than algorithmic leakage, causing standard procedures to incorrectly attribute this spurious signal to the audited model.

    
    \item \textbf{Privacy Leakage Correction (\Cref{sec:worstcase_compo,sec:conditional}).} We extend One-Run analyses \citep{steinke2023privacy,mahloujifar2024auditing} to the post-hoc regime, correcting empirical privacy estimates for distribution shift through two complementary viewpoints:
    \begin{itemize}
        \item \emph{A Composition Viewpoint (\Cref{sec:worstcase_compo})}, modeling observed leakage as an adaptive composition of a distribution-shift mechanism and the audited algorithm. This yields a deconvolution-based estimate of privacy leakage, though it can be conservative under severe distribution shift.
        \item \emph{A Conditional Viewpoint (\Cref{sec:conditional})}, characterizing pointwise privacy leakage conditional on the audit data. We tamper the membership guesses to remove overconfidence induced by confounding, yielding a robust correction that adapts to the observed distribution shift.
    \end{itemize}
    \item \textbf{Propensity Scores and Uncertainty Quantification (\Cref{sec:propensity_estimation}).}
    Our composition and conditional correction approaches rely on \textit{propensity scores} (the probability that an example is a member given its features), which must be estimated. We explicitly quantify uncertainty in propensity score estimation in our privacy lower bounds via bootstrapping \citep{efron1987better}, and provide a principled criterion for rejecting invalid propensity score models.
    \item \textbf{Empirical Evaluation (\Cref{sec:experiments}).}
    We evaluate Zero-Run auditing in three settings: a controlled Gaussian mechanism noisy-sum, a  CIFAR-10 image classification with synthetic non-members, and GPT-2-small on WikiText-103 with generated text as non-members. We show that ignoring distribution shift can produce invalid privacy lower bounds, while our corrections remain valid and retain power when sufficient member--non-member overlap is available. On CIFAR-10 and WikiText-103, we obtain non-trivial propensity-corrected privacy lower bounds and show that the resulting conclusions are stable across multiple propensity-score estimators, training checkpoints, and training regimes. 

\end{enumerate}

\section{Background on Differential Privacy and Privacy Auditing}

\paragraph{Differential Privacy.}
Informally, Differential Privacy (DP) \citep{DworkMNS06,Dwork2014a} ensures that a randomized algorithm's output remains largely unchanged whether a single individual's data is included or not in the input. DP is widely considered the standard for privacy-preserving data analysis because it provides a mathematically rigorous guarantee against adversaries with arbitrary side information. DP is immune to post-processing (protections cannot be degraded by further computation) and composes gracefully, enabling accurate tracking of the cumulative privacy loss across multiple analyses.

\begin{definition}[$(\eps,\delta)$-DP]
\label{def:eps-delta-dp}
Let $\eps,\delta>0$.
A randomized algorithm $\cA$ that takes as input a dataset $\cD \subset \cX$ and outputs $\cA(\cD) \in \Theta$ is said to satisfy $(\eps,\delta)$-differential privacy if, for any pair of adjacent datasets $\cD$ and $\cD'$ (i.e., differing in a single data sample) and any measurable set $\cE \subset \Theta$:
\begin{equation*}
    \proba{\cA(\cD)\in\cE}\leq e^\eps\proba{\cA(\cD')\in\cE} + \delta\,.
\end{equation*}
\end{definition}
Pure $\eps$-DP corresponds to the special case where $\delta=0$.
In this work, we will also rely on the more expressive notion of $f$-DP \citep{Dong22GDP}, which is based on trade-off functions that lower bound the Type II error given a Type I error of $\alpha$ for testing $z\sim \P$ versus $z\sim \Q$.
We denote by $T_{Z,Z'}$ the trade-off function between the distributions of two random variables $Z,Z'$.

\begin{definition}[Trade-off function]
\label{def:tradeoff_function}
    For two distributions $\P$ and $\Q$ on some space $\cZ$, the trade-off function $T_{\P,\Q}$ is the function $[0,1]\to[0,1]$ defined as $T_{\P,\Q}(\alpha) = \inf\set{\beta_\phi\,|\,\alpha_\phi\leq \alpha}$, where the infimum is taken over all measurable functions $\cZ\to\set{0,1}$, $\alpha_\phi=\E_\P\phi$ and $\beta_\phi=1-\E_\Q\phi$.
\end{definition}
\begin{definition}[$f$-DP]
A randomized algorithm $\cA$
is said to satisfy $f$-differential privacy for  $f$ a decreasing and convex function satisfying $f(p) \le 1 - p$ for $p \in [0,1]$  if, for any pair of adjacent datasets $\cD$ and $\cD'$, we have, for all $\alpha\in[0,1] $, $T_{\cA(\cD),\cA(\cD')}(\alpha)\geq f(\alpha)$.
Equivalently, for any measurable set $\cE$, 
$\proba{\cA(\cD)\in\cE}\leq \bar f \big(\proba{\cA(\cD')\in\cE}\big)$, where $\bar f = 1-f.$
\end{definition}
For two mechanisms $\cA_1$ and $\cA_2$, respectively $f_1$-DP and $f_2$-DP, their composition $\cA_2\odot\cA_1$ satisfies $f$-DP, with $f$ given by the convolution $f=f_1\otimes f_2$, which characterizes the optimal testing error after observing both outputs. This remains valid even if $\cA_2$ depends on the outputs of $\cA_1$ \citep{Dong22GDP}.

Two special cases of $f$-DP will be used throughout: $(\varepsilon,\delta)$-DP (\Cref{def:eps-delta-dp}), parameterized by $(\varepsilon,\delta)$, and Gaussian DP ($\mu$-GDP, \citep{Dong22GDP}), parameterized by a scalar $\mu$. Their respective trade-off functions are $f_{\eps,\delta}(p)=\max(0,1-e^\eps p-\delta,e^{-\eps}(1-\delta-p))$ and $G_\mu(p)= \Phi(\Phi^{-1}(1-p)-\mu))$. 
For $\mu$-GDP, the function $f$ coincides with the false negative rate as a function of the false positive rate for the optimal test distinguishing two unit-variance Gaussians with means differing by $\mu$.

While $(\eps,\delta)$-DP summarizes privacy using two worst-case parameters, $f$-DP and GDP characterize the full trade-off between false positives and false negatives. This finer description provides a complete view of an adversary's inference capability and enables more expressive privacy reporting and accounting \citep{gomez2025varepsilondeltaconsideredharmful}.

\paragraph{Privacy Auditing.}
\label{sec:auditing}
\looseness=-1 An auditing algorithm takes as input an algorithm $\cA$ (e.g., DP-SGD) and a privacy hypothesis $\cH$ (e.g., ``$\cA$ satisfies $f$-DP''), outputting a bit indicating whether $\cA$ satisfies $\cH$. Standard ``interventional'' auditing \citep{steinke2023privacy,mahloujifar2024auditing,jagielski2020auditing} follows a three-step procedure:
\begin{enumerate}
    \looseness=-1
    \item \textbf{\textit{(Randomized Experiment)}} Evaluate $\cA$ on a dataset \(\cD\) with randomized membership assignments;
        \looseness=-1
    \item \textbf{\textit{(Membership Inference)}} A \textit{membership inference attack (MIA)} produces guesses $(T_i)_{i\in[m]} \in \set{-1,0,1}^m$ for the true membership $S_i\in\{-1,1\}$ of $X_i$ based on the output of the experiment, where a guess of $T_i = 1$ indicates membership, $T_i = -1$ indicates non-membership, and $T_i = 0$ corresponds to abstention (e.g., low confidence).
        \looseness=-1
    \item \textbf{\textit{(Audit Evaluation)}}
    Output $b\in\set{\textsc{True},\textsc{False}}$ based on $\cH$, the MIA output $T$, and the experiment's observations (including true memberships).
\end{enumerate}

In the \textit{(Randomized Experiment)} step, both \textit{Multi-Run} \citep{jagielski2020auditing} and \textit{One-Run} \citep{steinke2023privacy} auditing randomize data inclusions, executing $\cA$ repeatedly for individual points \(X_i\), or once for multiple points \((X_i)_{i \in [m]}\), respectively. Both rely fundamentally on membership indicators \(S_i \in \{-1,1\}\) being independent and centered Bernoulli random variables. This randomized assignment is the key statistical assumption underlying the validity of interventional audits---an assumption that our Zero-Run auditing framework explicitly eliminates.

\begin{remark}\label{remark:empirical_param_0}
    For an error probability $p\in[0,1]$, the audit is $p$-valid if $\proba{b = \textsc{False} |\cA\text{ satisfies }\cH}\leq p$.
From such procedures, one can derive \emph{empirical privacy parameters}---high-probability lower bounds on leakage---by identifying the largest privacy level (e.g., $\hat\mu$ for GDP) that the audit can confidently rule out, see \Cref{remark:empirical_privacy_parameters} in \Cref{sec:propensity_estimation} for details.
\end{remark}

\section{Zero-Run Privacy Auditing Framework}
\label{sec:zero-run-auditing-framework}

Traditional privacy auditing techniques rely on randomized experiments (see \Cref{sec:auditing}) that demand computationally prohibitive re-executions and randomized data insertions, making them infeasible for complex, large-scale systems. We address this bottleneck by replacing the interventional randomized experiment with an \textit{observational study}. Instead of randomizing input data and re-running the algorithm $\cA$, our \emph{Zero-Run} auditor passively evaluates its existing output. This approach requires zero control over or access to the original execution pipeline.\looseness=-1

\begin{figure*}
    \centering
    \resizebox{\textwidth}{!}{%
\begin{tikzpicture}[
    >=stealth,
    font=\small,
    box/.style={draw, rounded corners=4pt, thick, align=center, fill=white, inner sep=6pt},
    data/.style={box, draw=gray!80!black, fill=gray!8},
    algo/.style={box, draw=mydarkblue, fill=mydarkblue!8},
    proc/.style={box, draw=purple, fill=purple!8},
    prop/.style={box, draw=orange!90!black, fill=lightorange!30},
    audit/.style={box, draw=mydarkgreen, fill=mydarkgreen!8},
    bound/.style={font=\bfseries, text=mydarkgreen, align=center},
    arrow/.style={->, thick, draw=black!70}
]

\node[algo, text width=3.8cm] (algo) at (0, 1.35) {
    \textbf{Algorithm Output}\\[0.6ex] 
    $\theta \sim \mathcal{A}(\mathcal{D}_{\text{train}})$
};

\node[data, text width=3.8cm] (data) at (0, -1.35) {
    \textbf{Observational Data}\\[0.3ex]
    {\footnotesize $X_i \in \mathcal{D}_1 \cup \mathcal{D}_0$}\\[0.5ex]
    \begin{tabular}{@{}c@{}}
        \scriptsize \textbf{Members} ($S_i=1$) \\
        \scriptsize $X_i \in \mathcal{D}_1 \subseteq \mathcal{D}_{\text{train}} \sim \mathbb{P}_1$ \\[0.4ex]
        \scriptsize \textbf{Non-members} ($S_i=-1$) \\
        \scriptsize $X_i \in \mathcal{D}_0 \subseteq \mathcal{D}_{\text{train}}^c \sim \mathbb{P}_0$
    \end{tabular}\\[0.5ex]
    \textcolor{mydarkred}{\scriptsize \textbf{Confounding Shift: } $\mathbb{P}_1 \neq \mathbb{P}_0$}
};

\node[proc, text width=3.2cm] (mia) at (5.8, 1.35) {
    \textbf{Membership Inference}\\[0.5ex] 
    Guesses $T_i$
};

\node[prop, text width=3.2cm] (prop) at (5.8, -1.35) {
    \textbf{Propensity Estimator}\\[0.5ex] 
    Scores $\hat{\pi}(X_i)$
};

\node[audit, text width=4.0cm] (audit) at (10.8, 0) {
    \textbf{Zero-Run Correction}\\[0.6ex] 
    \textbf{Compositional} {\scriptsize (Global)}\\
    \textbf{Conditional} {\scriptsize (Pointwise)}
};

\node[bound, text width=2.4cm] (bound) at (14.8, 0) {
    Empirical\\[0.3ex]
    Lower Bounds
};

\draw[arrow] (algo) -- (mia) node[midway, above, font=\footnotesize] {$\theta$};
\draw[arrow] (data) -- (prop) node[midway, above, font=\footnotesize] {$X_i$};
\draw[arrow] (data) -- (mia) node[midway, above, sloped, font=\footnotesize] {$X_i$};

\draw[arrow] (mia) -- (audit) node[midway, above, sloped, font=\footnotesize] {$T_i$};
\draw[arrow] (prop) -- (audit) node[midway, below, sloped, font=\footnotesize] {$\hat{\pi}(X_i)$};

\draw[arrow, dashed, rounded corners=6pt] 
    ([yshift=10pt]data.south east) -- ([yshift=10pt]data.south east -| audit.south) node[midway, above, font=\footnotesize] {$S_i$} 
    -- (audit.south);

\draw[arrow, very thick, mydarkgreen] (audit) -- (bound);

\end{tikzpicture}
    }
    \vspace{-0.2cm}
   \caption{Overview of the Zero-Run Privacy Auditing framework. Observational data ($X_i$) from members and non-members introduce confounding due to distribution shift ($\mathbb{P}_1 \neq \mathbb{P}_0$). We estimate this shift through a propensity scores $\hat{\pi}(X_i)$ and use them to correct raw membership guesses ($T_i$) via global composition or pointwise conditioning, isolating true algorithmic leakage and yielding valid empirical DP lower bounds.}

    \label{fig:zero_run_setup}
\end{figure*}

\paragraph{Zero-Run Auditing.}
As illustrated in \Cref{fig:zero_run_setup}, our \emph{Zero-Run auditing framework} operates as follows:
\begin{enumerate}
    \looseness=-1
    \item \textit{\textbf{(Observational Study)}}
    The auditor observes $\theta = \cA(\cD_{\mathrm{train}})$, the result of the algorithm $\cA$ executed on a private dataset $\cD_{\mathrm{train}}$. This output might be accessible via direct inspection or indirectly through a query interface. Crucially, the auditor lacks access to the full dataset $\cD_{\mathrm{train}}$ and the algorithm's internal mechanics. Instead, the auditor relies on \textit{observational data}: known members $\cD_1 \subset \cD_{\mathrm{train}}$ and known non-members $\cD_0 \subset \cX \setminus \cD_{\mathrm{train}}$. Let $\set{X_1, \ldots, X_m} = \cD_1 \cup \cD_0$ represent the combined data, with true membership indicators defined by $S_i = 1$ if $X_i \in \cD_1$ and $S_i = -1$ if $X_i \in \cD_0$. Unlike traditional randomized experiments, the membership variables $S_i$ are not centered conditionally on $X_i$; rather, $S_i$ inherently depends on the data features $X_i$, as $\cD_0$ and $\cD_1$ may exhibit a distribution shift.
    \looseness=-1
    \item \textit{\textbf{(Membership Inference)}} 
    Leveraging the accessible algorithm output $\theta$ and the features $X_i$ from the observational data, the auditor runs a Membership Inference Attack (MIA). This step produces raw membership guesses $T_i$ for each $X_i$. Formally, we abstract the MIA as a mapping $(\theta, (X_i)_{i \in [m]}) \mapsto T \in \set{-1, 0, 1}^m$, where $T_i=0$ denotes abstention. Further details on the MIA meta-algorithm are provided in \Cref{alg:mia} (\Cref{app:mia}).
    \looseness=-1
    \item \textit{\textbf{(Propensity Estimator)}} To correct bias in raw MIA estimates caused by distribution shifts between members and non-members, we compute an estimated \textit{propensity score} $\hat{\pi}(X_i)$, which approximates the \textit{true propensity} $\pi$ defined as follows.
    \begin{definition}[Propensity score]\label{def:propensity}
        \begin{equation*}
            \pi: x\in\cX \rightarrow \P(S_i = 1 \mid X_i = x)\in[0,1]
        \end{equation*}
    \end{definition}
    By capturing the inherent probability that a data point from the audit data belongs to the training set based solely on its features, the propensity score isolates structural selection likelihood from the signal contained in the algorithm's output.
    \Cref{sec:propensity_estimation} details this estimation and its induced uncertainty propagation.
    \looseness=-1
    \item \textit{\textbf{(Zero-Run Correction)}}
    Finally, we combine the true labels $S_i$, the raw MIA guesses $T_i$, and the propensity scores $\hat{\pi}(X_i)$ to produce a statistically valid audit evaluation. This correction stage effectively disentangles the algorithm's actual privacy leakage from the underlying data distribution shift, and constitutes our main theoretical contribution (\Cref{sec:worstcase_compo,sec:conditional}).  
    The correction can be applied \textit{compositionally} (yielding global bounds) or \textit{conditionally} (yielding pointwise bounds), outputting valid empirical DP lower bounds.
\end{enumerate}

This four-step meta-algorithm forms the backbone of our Zero-Run auditing framework. 

\paragraph{Observational Datasets.}

To ensure a provably valid audit, we require a specific assumption on the observed datasets $\cD_1$ and $\cD_0$.

\begin{assumption}[Data Generation]\label{hyp:data_generation}
The pairs $(X_i,S_i)_{i\in[m]}$ are i.i.d.; the member and non-member datasets $\cD_1=\set{X_i:S_i=1}$ and $\cD_0=\set{X_i:S_i=-1}$, of respective sizes $n_1$ and $n_0$, consist of independent samples drawn from distributions $\P_1$ and $\P_0$. 
\end{assumption}
There are several strategies for sourcing members and non-members, depending on the use case.

\textbf{Member dataset $\cD_1$.}  This dataset can be any subset of the training set of the audited model, such as a random subsample or a curated subset selected to maximize potential privacy leakage. Even when the complete training dataset is unknown to the auditor, it is often possible to identify \emph{some} training examples. For instance, when auditing large-scale models, $\cD_1$ may consist of public records, historical archives, or other corpora known to have been included in the training data. For instance, Wikipedia pages published before an LLM's training date can serve as known members \citep{shi2023detecting,meeus2024did,meeus2024sok}.

\textbf{Non-member dataset $\cD_0$.} 
The most straightforward scenario occurs when a randomized test split is available, yielding a dataset $\cD_0 \subset \cX \setminus \cD_{\mathrm{train}}$ with zero distribution shift ( $\P_1=\P_0$). However, such randomized splits are often unavailable for foundation models, whose training data are typically fixed, massive, and only partially known. In this case, alternative construction strategies are necessary. Following empirical evaluations for LLMs \citep{meeus2024sok}, one approach is to use data generated \textit{after} the algorithm's execution timestamp, ensuring $\cD_0$ is strictly excluded; however, a distribution shift ($\P_1 \neq \P_0$) is naturally expected due to temporal changes. Alternatively, one can synthesize non-members from a continuous approximation of the member distribution by training a generative model on $\cD_1$ and sampling from it, yielding non-members with
a close distributional match ($\P_0 \approx \P_1$), as in \pano \citep{kazmi2024panoramia} (further discussions in related works, \Cref{sec:related}). Our framework seamlessly integrates these varied choices of $\cD_0$.

\begin{remark}\label{remark:not_restricted_to_DP}
We emphasize that, as in \citep{steinke2023privacy, mahloujifar2024auditing}, our meta-algorithm does not restrict the auditor to algorithms that formally satisfy DP. Because our framework treats $\cA$ as a black box and evaluates privacy leakage solely from its outputs and observational data, it can be applied to arbitrary algorithms, regardless of whether they were designed to provide formal privacy guarantees.
\end{remark}

\section{Global Correction via Adaptive Composition}
\label{sec:worstcase_compo}

If $\P_1=\P_0$, membership is independent of the features $X_i$, so existing One-Run audits \citep{steinke2023privacy,mahloujifar2024auditing} apply directly: the observed member/non-member labels behave like randomized inclusion bits. Under distribution shift ($\P_1\neq \P_0$), however, the features themselves carry information about the membership bits $S_i$, allowing membership to be inferred with lower error than the fundamental limit imposed by DP (see Appendix~\ref{app:role_propensity}). An auditor may therefore attribute membership information arising from distribution shift to the algorithm's output, leading to an inflated estimate of privacy leakage.

More concretely, consider a constant, data-independent mechanism, which satisfies $0$-DP and therefore leaks no privacy. Suppose the known members are books published before 2020, while the known non-members are books published after 2023. A classifier can infer membership from the temporal distribution shift (e.g., changes in vocabulary or topics), which formally corresponds to a non-constant propensity score function $\pi$ (\Cref{def:propensity}), even though the model output contains no information about membership. A naive post-hoc audit that ignores this distribution shift would therefore falsely reject the true $0$-DP hypothesis because membership can be predicted better than chance.

\paragraph{Distribution Shift as Adaptive Composition.}
We account for this baseline ``distribution-shift leakage'' by modeling the auditor's observations---the audit dataset $(X_i)_{i\in[m]}$ and the output $\theta = \cA(\dtrain)$---as the composition of two randomized mechanisms: a \textit{distribution-shift mechanism} $\mathcal M_{\mathrm{DS}}$, which generates the audit data $(X_i)_{i\in[m]}$ from the membership bits, and the algorithm $\mathcal A$, which outputs $\theta=\mathcal A(\mathcal D_{\mathrm{train}})$. This composition view lets us use DP composition to separate spurious leakage due to distribution shift from the algorithm's actual privacy leakage.

Formally, let $\cM_\mathrm{Tot}$ be a virtual global randomized mechanism acting on a private membership vector $s \in \set{-1, 1}^m$. The auditor's available observations are then the output of two composed mechanisms applied to the hidden bits: $\cM_\Tot(s) = \big( \cM_\DS(s), \, \cM_\cA(\cM_\DS(s), s) \big)$, where :
\begin{align*}
    \textstyle
    \cM_\DS&(s) \sim \prod_{i=1}^m \big( \one_\set{s_i=1} \cP_1 + \one_\set{s_i=-1} \cP_0 \big)\,,\\
    &\cM_\cA(x,s) = \cA \big( \cD \cup \set{x_i : s_i=1} \big) \,,
\end{align*}
with $\cD = \dtrain \setminus \cD_1$. $\cM_\DS$ samples the audit examples from the member or non-member distribution according to their membership bits, and $\cM_\cA$ runs the algorithm on the examples whose bits indicate membership. We have the following result.
\begin{theorem}\label{thm:composition}
    Assume that \Cref{hyp:data_generation} holds, that $\cA$ is $f$-DP and that $\cM_{\DS}$ is $f_{\DS}$-DP.
    Then $\cM_\Tot$ is $f_{\Tot}$-DP with  $f_{\Tot} = f \otimes f_{\DS}\,,$ where $\otimes$ is the tensor product between trade-off functions \citep[Sec. 3]{Dong22GDP}.
\end{theorem}

\begin{remark}
    In \Cref{thm:composition}, $\cA$ takes data points in $\cX$ as inputs, and its DP guarantee is defined under \emph{Add/Remove} adjacency. In contrast,
    $\cM_{\Tot}$ and $\cM_{\DS}$ take membership bits in $\set{-1,1}$ as input, so their DP guarantees are defined under \emph{Replace} adjacency.
    \emph{Replace} adjacency for membership bits is however equivalent to \emph{Add/Remove} adjacency for data points.
\end{remark}

\paragraph{The Overlap Assumption.}
\Cref{thm:composition} accounts for distribution shift in observed privacy leakage by formalizing how leakage from the algorithm $\cA$ combines with leakage from the audit data. However, the $f_{\DS}$-DP guarantee becomes vacuous under severe distribution shift: if the features $X_i$ completely reveal membership, the underlying trade-off function $T_{\P_0, \P_1}$ approaches zero, effectively masking the algorithmic leakage of $\cA$ and yielding uninformative lower bounds. We therefore require the likelihood ratio between $\P_0$ and $\P_1$ to be bounded almost surely, which motivates our next overlap assumption.
Recall that $\pi(x) = \P(S_i = 1 \mid X_i = x)$ denotes the \textit{propensity score} function (\Cref{def:propensity}).

\begin{assumption}[Approximate overlap]
    \label{hyp:approx_overlap}
    There exists $\eta\in(0,1/2)$ and $\delta_\DS\in(0,1)$ such that $\P_1(\pi(X_i)\in [\eta,1-\eta])\geq 1-\delta_\DS$ and $\P_0(\pi(X_i)\in [\eta,1-\eta])\geq 1-\delta_\DS$.
\end{assumption}
In the i.i.d. setting ($\P_1=\P_0$), the propensity score is a constant function (equal to the proportion of members), so that Assumption~\ref{hyp:approx_overlap} holds (with $\eta=0.5$ in the One-Run regime, where members and non-members are balanced).
Under arbitrary distribution shift however, \Cref{hyp:approx_overlap} may not hold, and \Cref{prop:impossibility} (\Cref{app:impossibility}) shows that privacy auditing is impossible.

Under the overlap assumption, valid post-hoc auditing becomes feasible: we bound the trade-off function of $\cM_\DS$  and apply standard One-Run auditing to $\cM_\Tot$. Any leakage beyond that attributable to $\cM_\DS$ is then attributed to $\cA$, yielding valid empirical lower bounds.

\begin{corollary}
\label{cor:worst-case_audit}
Under \Cref{hyp:data_generation,hyp:approx_overlap}, let $\eps,\delta,\mu>0$, $\bar\eps_\DS=\log(\frac{1-\eta}{\eta})$, and $\bar\mu_\DS=\Phi^{-1}(\frac{e^{\bar \eps_\DS}}{1+e^{\bar \eps_\DS}})-\Phi^{-1}(\frac{1}{1+e^{\bar \eps_\DS}})$.
Then,
\textbf{(1)}
if $\cM_\Tot$ is not $(\eps+\bar\eps_\DS,1-(1-\delta)(1-\delta_\DS))$-DP, then $\cA$ is not $(\eps,\delta)$-DP; \textbf{(2)}
if $\cM_\Tot$ is not $G_{\sqrt{\mu^2+\bar\mu_\DS^2}}\otimes f_{0,\delta_\DS}$-DP, then $\cA$ is not $\mu$-GDP.
\end{corollary}

\paragraph{Zero-Run Auditing via Adaptive Composition.}
According to \Cref{thm:composition}'s compositional view, a valid Zero-Run audit of hypothesis $\cH$ for $\cA$ can be obtained by reducing it to a One-Run audit of $\cM_\Tot$ under a deflated hypothesis $\cH'$. Specifically, if the One-Run audit rejects $\cH'$ with probability $1-p$, then the resulting Zero-Run audit rejects $\cH$ with some probability $1-p'$.
\Cref{cor:worst-case_audit} gives two practical instances of this reduction. If $\cM_\Tot$ violates the deflated hypothesis $\cH'$ below, then $\cA$ violates the target hypothesis $\cH$:
\begin{itemize}
    \item For \textit{$(\eps,\delta)$-DP Auditing:} $\cH'$ is ``$\cM_\Tot$ is $(\eps+\bar\eps_\DS,1-(1-\delta)(1-\delta_\DS))$-DP'', $p'=p$, and the audit alogorithm is provided by \citet{steinke2023privacy}; 
    \item For \textit{$\mu$-GDP Auditing:} $\cH'$ is ``$\cM_\Tot$ is $(\mu+\mu_\DS)$-GDP'', $p'=p+\delta_\DS$, and the audit algorithm is provided by \citet{mahloujifar2024auditing}. 
\end{itemize}
The resulting audit algorithms (see \Cref{app:worst-case_bounds_algo_eps_delta,app:worst-case_bounds_algo_GDP} for pseudo-code) are however conservative, as they assign worst-case distribution-shift leakage uniformly across all data points, even when propensity scores are highly heterogeneous. Specifically, modeling distribution shift globally through $\mathcal{M}_{\DS}$ imposes a pessimistic, worst-case DP leakage:
\begin{equation*}
    \bar\eps_\DS=\sup_{x \in \mathcal{X}} \eps_\DS(x)\,,\quad\text{where}\quad \eps_\DS(x) = \left| \log \frac{\pi(x)}{1-\pi(x)} \right|\,.
\end{equation*} This supremum reduces the audit's statistical power because the local leakage $\eps_\DS(X_i)$ for most observed data points $(X_i)_{i\in[m]}$ may be substantially smaller than $\bar\eps_\DS$. For example, if we can only guarantee $\pi(x)\in[0.2,0.8]$, the resulting global penalty is $\bar\eps_\DS\approx 1.39$, even if most data points have propensity scores close to $0.5$.
We next address this limitation by \emph{conditioning} on the observed examples and adjusting each membership guess using its corresponding propensity score.

\section{A Pointwise Correction via Conditional Auditing}
\label{sec:conditional}

\looseness=-1 
To overcome the above worst-case penalty, we leverage \emph{pointwise leakage}.
Conditioned on an observation $X_i=x_i$, the total membership leakage, quantified as the log-odds ratio of membership likelihoods given $X_i$ and the model $\theta$, can be decomposed as the sum of a distribution shift-only leakage term and a model-induced privacy leakage term:
\begin{equation*}
    \begin{aligned}
        &\overbrace{\log\left(\frac{\proba{S_i=1|X_i=x_i,\theta}}{\proba{S_i=-1|X_i=x_i,\theta}} \right)}^{\text{Total pointwise membership leakage}}\\
        &=\underbrace{\log\left( \frac{\pi(x_i)}{1-\pi(x_i)} \right)}_{\leq\eps_\DS(x_i)} + \underbrace{\log\left(\frac{p(\theta|X_i=x_i,S_i=1)}{p(\theta|X_i=x_i,S_i=-1)} \right)}_{\text{Model-induced membership leakage}}
    \end{aligned}
\end{equation*}
Thus, the distribution shift's contribution to inferring $S_i$ is exactly the local log-odds $\eps_{\DS}(x_i)$.
For an $\eps$-DP algorithm, the model-induced membership leakage is upper-bounded by $\eps$, so the combined evidence can increase the correct membership odds by at most $\eps+\eps_{\DS}(x_i)$. This in turn bounds the probability of a correct membership guess ($T_i=S_i$) by
\begin{equation*} 
    \textstyle
    p_i(\eps)=\frac{e^{\eps+\eps_\DS(X_i)}}{1+e^{\eps+\eps_\DS(X_i)}}\,.
\end{equation*}
Thus, on data points with high distribution-shift leakage (large $\eps_\DS(X_i)$), MIAs will thus be overconfident even when the model itself leaks little privacy. We next develop our conditional auditing approach based on this observation.

\paragraph{Tampered Membership Guess.}
To leverage pointwise leakage, our conditional audit bounds the success rate of an independently \emph{tampered} MIA. Intuitively, tampering acts as \textit{randomized abstention}: 
a correct guess on a record with a highly imbalanced propensity score (e.g., $\pi(x)=0.99$) provides little evidence of algorithmic leakage, since the record was already highly likely to be identified as a member from $x$  alone. The audit therefore retains this guess only with a small probability. In contrast, a correct guess with $\pi(x)\approx 0.5$ is more informative and is retained with high probability.

Formally, given initial MIA predictions $(T_i)_{i\in[m]} \in \set{-1,0,1}^m$, we define the post-processed, tampered predictions as $T_i' = B_i T_i$, where $(B_i)_{i\in[m]} \in \set{0,1}^m$ are independent Bernoulli variables with parameter $b_i$, conditioned on $X_i$. Since our probabilistic bounds condition on the observed features $(X_i)_{i\in[m]}$, this procedure yields a \emph{conditional} audit rather than a global worst-case one.

\begin{theorem}[Conditional Auditing]
\label{thm:conditional_audit_tampered}
    Assume \Cref{hyp:data_generation} holds and the attacker makes at most $ r\in[m]$ guesses almost surely ($\#\set{i:T_i\ne0} \leq r$). Let $(B_i)_{i\in[m]}$ be independent Bernoulli variables with parameters $(b_i)_{i\in[m]}$, independent of all other variables.
    \begin{enumerate}
        \item \label{thm:conditional_audit_tampered:eps_delta}\textbf{$(\eps,\delta)$-DP.}
        If $\cA$ is $(\eps,\delta)$-DP and $b_i\leq\frac{1+e^{-\eps-\eps_\DS(X_i)}}{1+e^{-\eps}}$, then for any $v>0$:
        \begin{equation}\label{eq:bound_thm_eps_delta}
            \textstyle
            \begin{aligned}
                &\proba{\sum_{i=1}^m B_i \one_\set{T_i=S_i} \geq v \Big| (X_i)_{i\in[m]}} \\
                &\leq \proba{Z \geq v} + \alpha_{r,\eps}(v) m\delta(1+e^{-\eps})\,,
            \end{aligned}
        \end{equation}
        where $Z \sim \mathrm{Binom}\big(r , \frac{e^\eps}{1+e^\eps}\big)$ and $\alpha_{r,\eps}(v) = \sum_{i\in[v]}\frac{1}{i} \proba{Z = v-i}$.

        \item \label{thm:conditional_audit_tampered:fDP}\textbf{$f$-DP.}
        If $\cA$ is $f$-DP and $b_i\leq e^{-\eps_\DS(X_i)}$, then for any subset $T\subset[r]$:
        \begin{equation}\label{eq:bound_thm_fdp}
        \textstyle
            \sum_{k\in T}\frac{kp_k}{m} \leq \bar f\bigg( \sum_{k\in T} \frac{r-k+1}{m}p_{k-1} \bigg)\,,
        \end{equation}
        where $p_k\!=\!\proba{\sum_{j=1}^m \!\!B_j\one_\set{S_j=T_j} \!=\! k \big|(X_i)_{i\in[m]}}$ is the probability of $k$ correct tampered guesses.
    \end{enumerate}
\end{theorem}

\Cref{thm:conditional_audit_tampered} shows that the number of correct guesses retained \textit{after tampering} is stochastically controlled by the standard benchmark for One-Run auditing without distribution shift, for both $(\eps,\delta)$-DP and $f$-DP formulations.
This allows us to pay only \textit{pointwise} propensity penalties via the tampering parameters $b_i$, rather than the global penalty of the adaptive composition approach of \Cref{sec:worstcase_compo}.
The cost of distribution shift can be interpreted as a reduction in the audit's \textit{effective sample size}, from $m$ in the absence of distribution shift to $\sum_i b_i$ after pointwise correction.

Under $(\eps,\delta)$-DP, the pointwise tampering parameters $b_i$ are naturally chosen so that the success probability of the event $\set{T_i'=S_i}$ is $\frac{e^\eps}{1+e^\eps}$, by noticing that the probability of this event is bounded by $p_i(\eps)b_i$. 
For $f$-DP, the corresponding $b_i$ is obtained as the limit of the $(\eps,\delta)$-DP tampering parameters when $\eps\to\infty$.
In both cases, in the absence of distribution shift we have $B_i=1$ almost surely, recovering the standard One-Run results \citep{steinke2023privacy,mahloujifar2024auditing}.

\begin{figure*}[t]
\centering

\begin{minipage}{0.60\textwidth}
\begin{algorithm}[H]
    \caption{Conditional Zero-Run $f$-DP auditing
    \label{alg:audit_f_DP}
    }
    \begin{algorithmic}[1]
      \STATE \textbf{Inputs:} 
      Instance $\theta=\mathcal{A}(\dtrain)$ of the algorithm $\cA$ to audit; MIA $\cM$; $\cD_1\subset \dtrain$; $\cD_0\subset \cX\setminus\dtrain$; error $p$; $\bar f $; number of guesses $r$; tampering parameters $(b_i)$
      \STATE $\set{X_1,\dots,X_m}\gets\cD_1\cup\cD_0$
      \STATE $S_i\gets 2\one_\set{X_i\in\cD_1} -1 \in\set{-1,1}$
      \STATE Run \Cref{alg:mia} on $(\theta, \cD_1, \cD_0, \cM)$ to obtain $(T_i)$
      \STATE $c\gets \sum_{j=1}^m\mathrm{Bernoulli}(b_i) \one_\set{T_i=S_i}$
      \STATE Run \Cref{alg:f_DP_counts} on $(p,\bar f,r,c)$ to obtain $r[0] + h[0]$
      \STATE \textbf{Return} \textsc{False} \textbf{if} $r[0] + h[0] \geq \frac{r}{m}$
    \end{algorithmic}
\end{algorithm}
\end{minipage}
\hfill
\begin{minipage}{0.38\textwidth}
\begin{algorithm}[H]
    \caption{$f$-DP success counts
    \label{alg:f_DP_counts}
    }
    \begin{algorithmic}[1]
      \STATE \textbf{Inputs:} 
      Error $p$; $\bar f $; number of guesses $r$, number of correct guesses $c$
      \STATE Set $h[k]=r[k]=0$ for $k\in[c]$
      \STATE  $r[c]\gets\frac{pc}{m}$, $h[c]\gets\frac{p(r-c)}{m}$
      \FOR{$k\in\set{c-1,\ldots,0}$}
        \STATE $h[k] \gets \bar f^{-1}(r[k\!+\!1]) $
        \STATE $r[k] \gets r[k\!+\!1] + \frac{k(h[k]-h[k+1])}{r-k}$
      \ENDFOR
      \STATE \textbf{Return} $r[0] + h[0]$
    \end{algorithmic}
\end{algorithm}
\end{minipage}

\end{figure*}

\paragraph{Zero-Run Auditing via Pointwise Corrections.}
For $(\eps,\delta)$-DP, \Cref{eq:bound_thm_eps_delta} directly bounds the right tail of the tampered MIA's success rate under the privacy hypothesis.
Thus, an $(\eps,\delta)$-DP privacy hypothesis can be confidently rejected when the tampered MIA achieves unusually high success, as formalized in \Cref{thm:conditional_audit_eps_delta}.\ref{thm:conditional_audit_tampered:eps_delta}. The corresponding privacy auditing algorithm is given in \Cref{alg:audit_approx_DP_tampered} (\Cref{sec:audit_approx_DP_tampered}).
While the proof of \Cref{thm:conditional_audit_tampered}.\ref{thm:conditional_audit_tampered:eps_delta} (\Cref{app:proof_pure_DP_tampered}) draws inspiration from \citet{steinke2023privacy}, our treatment of $\delta$ differs substantially, using conditioning of privacy loss random variables and their relation to Hockey-stick divergences and $(\eps,\delta)$-DP.
The error term $\alpha_{r,\eps}(v)$ differs slightly from that of \citet{steinke2023privacy}, but is expected to have the same order, $\cO(1/r)$, for reasonable parameter values.

For an $f$-DP privacy hypothesis, we can apply the indirect bound in \eqref{eq:bound_thm_fdp} recursively to subsets of the form $\set{r',r'+1,\ldots,r}$ for $r'\leq r$, thereby upper-bounding the right tail of the tampered MIA's success rate. This yields \Cref{alg:audit_f_DP}, which confidently rejects the $f$-DP hypothesis when the tampered MIA achieves unusually high success rates.
The proof of \Cref{thm:conditional_audit_tampered}.\ref{thm:conditional_audit_tampered:fDP} (\Cref{app:proof_f_DP_cond}) is inspired by \cite{mahloujifar2024auditing}.

\begin{remark}[Propensity-aware MIA]
\label{remark:propensity_aware_MIA}
    MIAs (\Cref{alg:mia}) produce membership guesses $T_i$ by thresholding a (non)member confidence score $c_i$ and retaining the top-$r$ most confident guesses, while abstaining on the rest.
    To improve our conditional audit, \Cref{thm:conditional_audit_tampered} suggests a propensity-aware alternative that directly optimizes the \emph{tampered} MIA success rate: rank guesses by their adjusted confidence $c_i\times b_i$, where $b_i$ is the corresponding tampering probability. This prioritizes examples with stronger algorithmic signal while downweighting those whose membership can be readily inferred from distribution shift.
\end{remark}

\begin{remark}[Extracting empirical privacy parameters]
\label{remark:empirical_privacy_parameters}
   Hypothesis testing yields high-confidence empirical privacy lower bounds. Let $b_\mu=\textsc{False}$ denote \Cref{alg:audit_f_DP} rejecting $\cH_\mu$: ``$\cA$ is $\mu$-GDP''. Because rejecting $\mu$-GDP implies rejecting all $\mu'\leq\mu$, defining $\hat{\mu} = \sup\set{\mu\geq0 \mid b_\mu=\textsc{False}}$ guarantees $\mu^\star\geq\hat{\mu}$ with probability at least $1-p$, given an audit failure probability $p$. A similar search applies to pure $\eps$-DP and monotonic $(\eps,\delta)$-DP hypotheses via \Cref{alg:audit_approx_DP_tampered}.
\end{remark}

\section{Propensity Score Estimation and Uncertainty Quantification}
\label{sec:propensity_estimation}

The propensity scores $\pi(X_i)$ defined in \Cref{def:propensity} (\Cref{sec:zero-run-auditing-framework}) are critical for isolating distribution shift from true algorithmic privacy leakage in our Zero-Run auditing algorithms introduced in \Cref{sec:worstcase_compo,sec:conditional}.
This section details our methodology for estimating them and quantifying their associated uncertainty. 

\paragraph{Validity with Oracle or Conservative Propensity Scores.}
First, our algorithms yield valid DP rejections provided that the estimated scores $\hat\pi$ are either exact or pessimistic, 
as formalized in \Cref{cor:validity_oracle}.
Thus, the auditor only needs a one-sided guarantee: overestimating overlap can invalidate the audit by failing to correct enough, whereas underestimating overlap merely attributes more signal to distribution shift, making the audit more conservative while preserving validity.\looseness=-1
\begin{corollary}\label{cor:validity_oracle}
    If $\min(\hat \pi_i,1-\hat\pi_i)\leq \min( \pi(X_i),1-\pi(X_i))$ for all $i\in[m]$ and if the composition or conditional auditing algorithm returns \textsc{False} for $(\eps,\delta)$-DP (respectively, for $f$-DP), then $\cA$ is not $(\eps,\delta)$-DP (respectively, $\cA$ is not $f$-DP) with probability $1-p$.
\end{corollary}
\Cref{thm:composition} and Corollary~\ref{cor:validity_oracle} yield valid finite-sample audits when the propensity score function $\pi$ is known exactly, or bounded conservatively pointwise. In practice, as described next, we estimate the propensity scores from data and quantify estimation uncertainty via bootstrapping; under the regularity assumptions stated below, this yields an asymptotically conservative approximation.

\paragraph{Bootstrapped Privacy Lower Bounds.}
We consider a separate dataset $(X_j', S_j')_{j\in[m']}$, distinct from the audit dataset $(X_i, S_i)_{i\in[m]}$, for propensity score estimation.
To quantify the statistical uncertainty on empirical privacy lower bounds due to estimated propensity scores, \emph{we bootstrap the final empirical privacy lower bound by resampling this dataset}. For ease of exposition, we focus on $\mu$-GDP auditing and consider the empirical privacy lower bound $\hat \mu$ (see \Cref{remark:empirical_privacy_parameters}) as a function of propensity scores $(\pi(X_i))_i$, all the rest (model $\theta$, MIA guesses $T_i$, data points $X_i$, and memberships $S_i$) being fixed, and proceed as follows for $K$ bootstrapped propensity scores.
For $k\in[K]$: 
\begin{enumerate}
\item \textit{Bootstrap sample:} we draw $({X_j'}^{(k)}, {S_j'}^{(k)}){j\in[m']}$ by sampling $m'$ pairs with replacement from $(X_j',S_j'){j\in[m']}$;
\item \textit{Propensity score training:} we train a probabilistic binary classifier $\hat{\pi}^{(k)}:\cX\rightarrow[0,1]$ on $({X_j'}^{(k)}, {S_j'}^{(k)})_{j\in[m']}$ using cross-fitting (see \Cref{alg:propensity_learning}, \Cref{app:propensity_learning});
 \item \textit{Privacy lower bound:} we feed $(\hat\pi^{(k)}(X_i),T_i,S_i)_{i\in[m]}$ into the auditing algorithm to obtain an empirical privacy lower bound $\hat{\mu}^{(k)}$.

\end{enumerate}
Finally, we output the $p'$-lower quantile $\hat\mu$ of $(\hat{\mu}^{(k)})_{k\in[K]}$, yielding a robust empirical privacy bound with probability $1-p-p'$, , where $p$ and $p'$ allocate the error budget to the audit and bootstrap, respectively. Formal bootstrap validity requires standard, though practically unverifiable, regularity conditions—such as Hadamard differentiability \citep[Chapter~23]{van2000asymptotic}—on the mapping from propensity scores $(\pi(X_i))_i$ to the privacy lower bound. The full procedure is detailed in \Cref{alg:bootstrap} (\Cref{app:bootstrap}).

We emphasize that bootstrap-based uncertainty quantification is only asymptotically valid under these assumptions. Unlike the exact finite-sample validity obtained with oracle or certified conservative propensity bounds (\Cref{cor:validity_oracle}), our bootstrapped bounds should therefore be viewed as practical, asymptotic empirical privacy lower bounds.

\paragraph{Auditing and Falsifiability of Propensity Score Models.}
Our Zero-Run audits rely on estimated propensity score bounds. Although these bounds cannot generally be formally certified, \textit{they can be empirically falsified}: they can themselves be treated as a hypothesis and subjected to an audit. Analogous to testing a DP claim, a propensity-score model $\hat \pi$ can be rejected when othe bserved evidence contradicts its claimed conservative bounds $\max(\pi(X_i),1-\pi(X_i))\leq \max(\hat\pi(X_i),1-\hat\pi(X_i))$; see \Cref{remark:propensity_rejection} and \Cref{thm:pure_DP_auditing_tampered} (\Cref{app:proof_pure_DP_tampered})
This test is necessarily one-sided: passing it does not establish that the bounds are correct, but only indicates that no violation was detected given the available data and auditing power.

This motivates releasing the propensity score models and bounds, audit dataset, and auditing procedure alongside the final audit result. A skeptic can then challenge the claimed  distribution-shift correction by constructing a stronger propensity score that exploits distribution shift beyond what the released bounds allow. Thus, while propensity models remain formally unverifiable, openness and adversarial testing make the audit empirically falsifiable and thereby more trustworthy in practice.

\section{Experiments}
\label{sec:experiments}

\paragraph{Auditing Setup and Datasets.}
We quantify empirical privacy leakage in terms of $\mu$-GDP. To account for uncertainty in propensity-score estimation, we compute empirical lower bounds $\hat{\mu}$ using $K=600$ bootstrap resamples. We target an overall confidence level of $0.95$ by splitting the total error budget of $0.05$ equally between hypothesis testing ($p=0.025$) and bootstrap uncertainty quantification ($p'=0.025$). In contrast, standard IID/One-Run baselines neither account for distribution shift nor incorporate bootstrap uncertainty. We evaluate Zero-Run auditing in three experimental settings:
\begin{enumerate}
    \item \emph{Controlled Gaussian Mechanism:} We consider a synthetic noisy-sum Gaussian mechanism, with a tunable parameter to control the magnitude of the distribution shift.
    
    \item \emph{Vision Data:} We audit a ResNet architecture using CIFAR-10 data \cite{krizhevsky2009learning}, where we generate synthetic examples to serve as non-members. 
    
    \item \emph{Text Data:} We audit a GPT-2-small model using WikiText data \cite{merity2016pointer}, where we generate synthetic examples to serve as non-members following the generation procedure of \texttt{Panoramia} \cite{kazmi2024panoramia}.
\end{enumerate}

\paragraph{Propensity Score Estimation.}
For each audit, we
estimate \(\pi(x)=\Pr(S_i=1\mid X_i=x)\) from fixed
representations of \(X_i\) produced by encoders independent of the audited
model. We compare six propensity pipelines: \(L_2\)-regularized logistic
regression with either sigmoid or isotonic calibration, a linear SVM with
sigmoid calibration, random forests, extremely randomized tree ensembles, and
histogram gradient boosting, each with sigmoid calibration \cite{scikit-learn}. All propensity
scores are computed out-of-sample using five stratified outer folds and three
inner folds for calibration and classifier fitting. Due to lack of space, we defer detailed propensity score model comparisons (e.g., AUC-ROC, effective sample size and calibration error) to \Cref{app:propensity-estimation-details}.

\looseness=-1 \paragraph{Membership Inference Attacks (MIAs).}
Our instantiation of MIA (\Cref{alg:mia}) leverages the combination of a standard MIA score (cross-entropy loss for images; perplexity for text) and the estimated propensity score $\hat{\pi}(X_i)$ in a two-dimensional histogram. The histogram estimates a cellwise member-to-non-member likelihood ratio, which we use to search for a classification rule targeting a specified false-positive rate on held-out tuning data; the selected rule is then fixed for final evaluation. Providing the MIA with $\hat{\pi}(X_i)$ is consistent with the DP threat model: propensities are estimated from observable covariates using a separate propensity model, independently of the target model and the final audit labels. This allows the MIA to exploit all available information when forming its initial guesses, before the subsequent correction for distribution shift.

We also incorporate propensity-aware confidence ranking as described in \Cref{remark:propensity_aware_MIA} to favor confident guesses in regions where members and non-members have greater distributional overlap. 
In particular, the adversary accounts for the probability that the tampering procedure will retain a guess at each propensity level, allowing it to abstain strategically and prioritize guesses that are more informative about algorithmic leakage.


\paragraph{Controlled Gaussian Mechanism Results.} We first consider a simple sum query on synthetic data privatized with the Gaussian mechanism (see \Cref{appendix:gaussian_syntethic_data}) as a controlled illustration of our distribution-shift corrections. Rather than modeling a realistic image or text distribution, this setup cleanly separates mechanism-induced leakage from features-only membership information. Across all shift levels, we keep the membership assignments, Gaussian release, and mechanism-derived scores fixed, so the true privacy parameter remains $\mu_{\mathrm{true}}=2/3$. Only an auxiliary feature correlated with membership changes, ensuring that any increase in apparent leakage is attributable to confounding rather than the mechanism itself.

\Cref{fig:decomposition_main} summarizes the results.
At the IID baseline (no shift), corrected and uncorrected audits behave similarly. As the auxiliary feature becomes more predictive of membership, the uncorrected audit progressively attributes this features-only signal to the mechanism, eventually exceeding the known theoretical privacy level and thus producing an invalid privacy lower bound. In contrast, our corrected audits maintain validity across all shift levels. The global correction quickly loses power because it applies a uniform worst-case penalty, whereas the propensity-aware pointwise correction remains informative under moderate shifts by adapting to each observation's estimated propensity. Under the strongest shifts, it also loses power as member/non-member overlap deteriorates.

Overall, the experiment demonstrates that propensity correction removes spurious membership signal arising from distribution shift while preserving evidence of mechanism-induced leakage when sufficient overlap remains.

\begin{figure*}[t]
    \centering

    \begin{subfigure}[t]{0.48\textwidth}
        \centering
        \includegraphics[width=\linewidth]
        {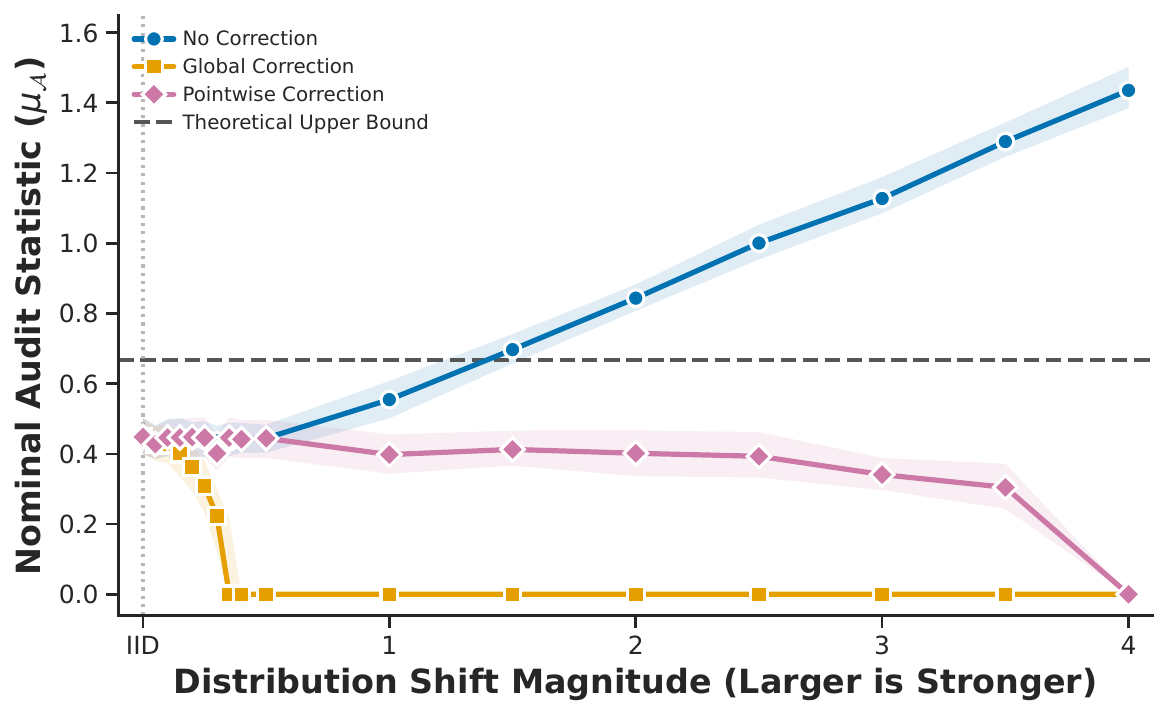}
        \caption{Comparison of different auditing algorithms for the Controlled Gaussian Mechanism experiment.}
        \label{fig:decomposition_main}
    \end{subfigure}
    \hfill
    \begin{subfigure}[t]{0.48\textwidth}
        \centering
        \includegraphics[width=\linewidth]
        {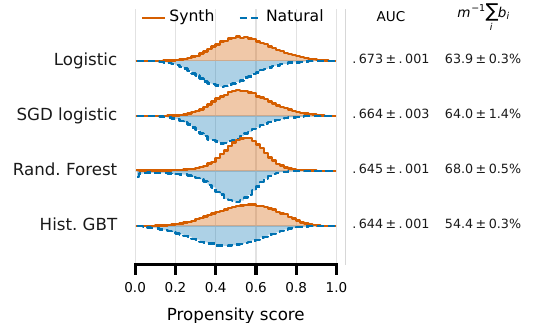}
        \caption{Performance of four propensity-score estimators in
        distinguishing real member CIFAR-10 images from synthetic non-member
        images. Entries report five-fold means $\pm$ standard deviation.}
        \label{fig:cifar10-propensities}
    \end{subfigure}

    \caption{Auditing and propensity-score estimation results on synthetic
    Gaussian and CIFAR-10 data.}
    \label{fig:combined-results}
\end{figure*}

\begin{figure*}
    \centering
    \includegraphics[width=\linewidth]{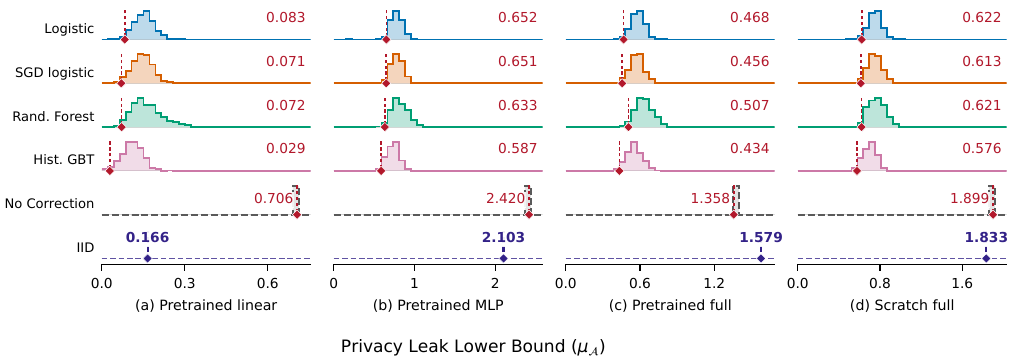}
        \caption{Epoch-100 \(\mu\)-GDP lower-bound distributions for four
    Wide ResNet-50-2 training regimes on CIFAR-10. Distributions labeled by a
    propensity estimator use propensity-corrected synthetic non-members.
    \emph{No Correction} uses the same synthetic non-members with
    \(\hat{\pi}(x)=0.5\), while \emph{IID} uses held-out real CIFAR-10
    non-members. Dotted lines mark the \(0.025\) quantiles.}
\label{fig:cifar10-audit-final} 
\end{figure*}

\paragraph{CIFAR-10 Results.}
We first attempted to reuse \pano's \cite{kazmi2024panoramia} synthetic non-member generator. However, propensity analyses based on independent ViT embeddings revealed a severe generator-induced distribution shift, leaving insufficient overlap for a meaningful privacy lower bound; we discuss this failure mode in \Cref{sec:failure-modes}. We therefore developed a class-conditional diffusion and selection procedure, described in \Cref{app:synth-data-cifar10}, which more closely matches the coverage, diversity, and difficulty of real CIFAR-10 images.

Using these synthetic non-members, we audit the epoch-100 checkpoints of four Wide ResNet-50-2 training regimes: a frozen ImageNet-pretrained backbone with a linear head (20.5K trainable parameters), the same frozen backbone with an MLP head (16.8M), full fine-tuning from pretrained weights (66.9M), and full training from scratch (66.9M). \Cref{fig:cifar10-propensities} quantifies the remaining real--synthetic distribution shift using four propensity-score estimators, while \Cref{sec:propensity_cifar} reports their calibration and overlap diagnostics. We observe reasonable agreement across the different propensity score models.

\Cref{fig:cifar10-audit-final} shows the bootstrap distributions of the resulting \(\mu\)-GDP lower bounds, alongside two reference baselines: \emph{No Correction}, an uncorrected synthetic-data audit that fixes \(\hat{\pi}(x)=0.5\) and therefore ignores confounding, and \emph{IID}, a benchmark based on held-out real CIFAR-10 images as non-members. The substantially larger uncorrected bounds illustrate how ignoring the real--synthetic shift can severely overstate algorithmic privacy leakage.

Across the four propensity estimators, the conservative epoch-100 \(\mu\)-GDP lower bounds range from \(0.587\) to \(0.652\) for the pretrained MLP head, from \(0.434\) to \(0.507\) for full fine-tuning, and from \(0.576\) to \(0.622\) for training from scratch. The frozen pretrained linear head yields substantially smaller bounds, ranging from \(0.029\) to \(0.083\). The close agreement across propensity models indicates that the qualitative ordering of the four training regimes is robust to the model choice: higher-capacity heads and end-to-end training exhibit substantially more certifiable membership leakage than the linear-head regime. \Cref{app:cifar10-experiments} provides full details of the synthetic-data construction, audited model training, and propensity estimation, while \Cref{fig:cifar10-audit-trajectories} traces the corrected leakage over the full 100-epoch training trajectory.

\begin{figure*}[t]
    \centering
    \includegraphics[width=\linewidth]{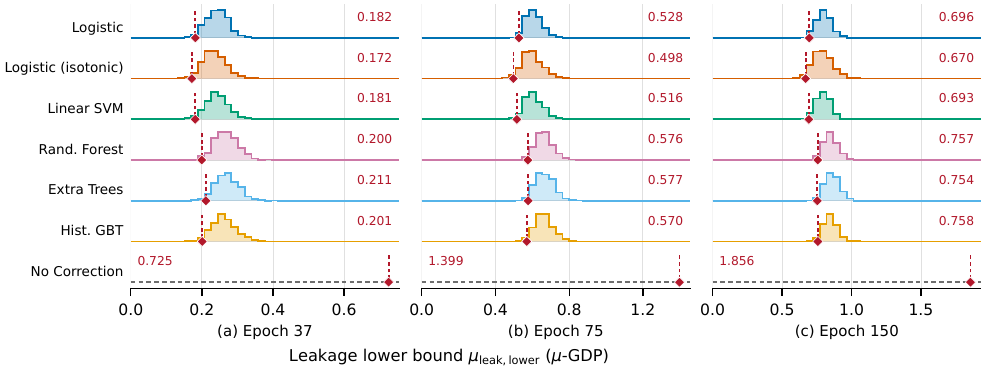}
    \caption{\(\mu\)-GDP lower-bound distributions for three checkpoints of a GPT-2
target model.}
\label{fig:wikitext-audit-final} 
\end{figure*}

\paragraph{WikiText Results.} Following the WikiText-103 protocol in \Cref{app:wikitext}, we audit epochs \(37\), \(75\), and \(150\) from a single \(200\)-epoch fine-tuning trajectory of a \(124.4\)M-parameter GPT-2-small model. The audit contrasts known \(64\)-token training chunks with \(64\)-token synthetic non-members from a separately trained GPT-2-small generator. Although the two models share the same architecture and pretrained initialization, the generated suffixes were never used to train the audited model and have no exact overlap with its training or held-out sequences, reproducing the generation technique of \texttt{Panoramia} \cite{kazmi2024panoramia}. Propensity diagnostics in \Cref{fig:wikitext-propensity-summary} show a consistent but moderate real--synthetic shift across all six estimators (held-out ROC AUC \(0.64\)--\(0.68\)), while the conditional correction retains approximately \(60\%\)--\(70\%\) of the audit sample, indicating sufficient local overlap for a non-vacuous audit.

As shown in \Cref{fig:wikitext-audit-final}, the conservative propensity-corrected \(\mu\)-GDP lower bounds increase from \(0.172\)--\(0.211\) at epoch \(37\), to \(0.498\)--\(0.577\) at epoch \(75\), and \(0.670\)--\(0.758\) at epoch \(150\). Even the most conservative estimator follows \(0.172\rightarrow0.498\rightarrow0.670\), and the widest cross-estimator range (\(0.088\)) is smaller than either adjacent-checkpoint increase. Thus, the temporal trend dominates variation across propensity models. Because the audit populations and their distribution shift are fixed across checkpoints, the persistence of this increase after correction, together with worsening held-out loss and perplexity, is consistent with increasing model-dependent membership leakage as the audited model overfits, rather than with a static generator artifact. \emph{The no-correction controls are much larger (\(0.725\), \(1.399\), and \(1.856\)); these are invalid as privacy bounds under the observed shift} and serve only to illustrate how strongly confounding can inflate apparent leakage, not to provide an additive decomposition of its sources.

\section{Failure Modes}
\label{sec:failure-modes}
We next discuss settings in which post-hoc auditing may have limited statistical power or may be fundamentally impossible.

\paragraph{Overlap failures.}
Observational auditing becomes impossible when members and non-members are too dissimilar to provide sufficient overlap.
A representative example is the setup of \citet{hayes_measuring_2025}, which estimates extraction probabilities for Pythia models \cite{biderman_pythia_2023} on the Enron email subset of the Pile \cite{gao2020pile}, using the TREC spam corpus as a reference non-member dataset. Reusing this setup for privacy auditing requires substantial care: there is no reason to assume that Enron and TREC emails are exchangeable, or even that their feature distributions have sufficient overlap. Consequently, an attack may distinguish members from non-members using corpus-specific features rather than privacy leakage from the model. \Cref{fig:enron-trec-main-separation} illustrates this failure mode. Across five propensity-score estimators, the two corpora are highly distinguishable, with AUCs close to $0.98$; the random forest achieves an AUC of $0.982 \pm 0.002$. This severe separability indicates a fundamental lack of overlap. Ignoring the distribution shift would produce \emph{strong but spurious privacy lower bound}. After accounting for the shift, however, no non-trivial lower bound can be obtained: \emph{the corrected lower bound collapses to 0}, consistent with our impossibility result for non-overlapping distributions (\Cref{prop:impossibility}, \Cref{app:impossibility})

\paragraph{When IID data is not truly IID.}
The iWildCam \cite{marklund2020wilds} benchmark consists of wildlife images collected by different cameras. Its test split is IID only in a
restricted sense: it contains images from camera locations that are
also represented in the training set. While this removes the explicit geographic shift targeted by the OOD split, it does
not imply the training and test examples are drawn independently from the same distribution.  As shown in
Figure~\ref{fig:iwildcamera_train_test_iid}, nonlinear propensity-score
estimators distinguish the training set from the nominally IID test set
with AUCs of up to \(0.847\), while the corresponding effective sample
size falls to \(14.0\%\). The two sets are therefore not exchangeable
despite sharing camera locations, with differences arising, for instance, in label distributions and camera-animal distances. Consequently, the nominally IID test split cannot be used directly as an non-member sample for a valid
One-Run audit: an MIA may exploit these split-specific differences to improve its membership predictions.
The effective lesson is that even held-out test should be treated with care: without evidence that train and test examples are exchangeable, our Zero Run framework should be employed.

\paragraph{When $\P_1=\P_0$ is not enough.} Even when the member and non-member distributions are perfectly matched, i.e., $\P_1=\P_0$, this alone does not rule out strong dependence between individual samples. The iWildCam dataset again provides a concrete example: it contains sequences of near-duplicate images of the same animal, corresponding to consecutive frames and hence exhibiting substantial sample-level correlation. Constructing a held-out non-member set $\cD_0$ by uniformly sampling such data can therefore produce examples that are highly correlated with members in $\cD_1\subset\cD_{\mathrm{train}}$. Although there is no marginal distribution shift, these correlations can substantially weaken a Membership Inference Attack (MIA): non-members may behave statistically like nearby training members, obscuring the distinction between $S_i=1$ and $S_i=-1$ and causing the audit to report little or no empirical privacy leakage even when the model memorizes its training data. 

This illustrates a more general limitation of One-Run auditing. If the audited party knows in advance which examples may be used for evaluation, it can include those examples (or sufficiently correlated examples) in $\cD_{\mathrm{train}}$, thereby weakening the membership signal on which the audit relies. Importantly, this issue need involve malicious behavior: in realistic datasets, latent grouping, temporal structure, repeated observations, and near-duplicates can create substantial sample-level dependence that is difficult to identify and control, even when the member and non-member marginals appear well matched. Using an OOD non-member dataset $\cD_0$ that was unavailable to the audited party at training time mitigates this failure mode.

\paragraph{Challenges with Synthetic Non-Members.}
Synthetic data provide an attractive way to construct a non-member dataset \(\cD_0\) when naturally occurring non-members are unavailable. For privacy auditing, however, perceptual or semantic realism alone is insufficient: the generated distribution must also exhibit substantial overlap with the member distribution, ideally satisfying \(\P_0 \approx \P_1\). Even subtle systematic artifacts introduced by a generator can allow a propensity-score model to distinguish real from synthetic examples using their features alone. An MIA may then exploit the generation process rather than solely the privacy leakage from the audited model.

We observed this failure when reproducing \pano's \cite{kazmi2024panoramia} CIFAR-10 experiments with StyleGAN2-ADA synthetic non-members. Using representations from an independent ViT encoder, every propensity model we evaluated distinguished real from generated images with a ROC AUC of approximately \(0.915\)--\(0.975\) (\Cref{fig:panoramia-synth-cifar10}, indicating negligible overlap between the two populations. Thus, despite their high perceptual quality, these generated images are unsuitable substitutes for IID non-members in a privacy audit: without appropriate correction, the audit primarily measures generator-induced distribution shift rather than model-induced membership leakage.

A second, potentially harder challenge is that synthetic data must reproduce not only the marginal distribution of individual examples but also relevant dependencies \emph{between} examples. Real datasets often contain latent relational structure, including repeated entities, near-duplicates, temporal dependencies, or groups of observations originating from the same underlying source. The Enron email corpus provides a concrete example. Names, email addresses, organizations, and participants recur across messages, while many emails belong to longer reply or forwarding threads and therefore share quotations, headers, vocabulary, and conversational context. Generating each email independently may reproduce the appearance of an individual Enron message while failing to reproduce this dataset-level structure, leaving strong signals that distinguish real from synthetic samples. This matters for privacy auditing because such dependencies can also affect the membership-inference signal itself: matching single-example statistics does not guarantee that synthetic non-members behave like genuine held-out observations relative to $\cD_{\mathrm{train}}$. Constructing synthetic datasets that simultaneously preserve realistic sample-level distributions, sufficient propensity overlap, and relevant cross-sample dependencies therefore remains a important open challenge.

\begin{table*}[t]
    \centering
    \small
    \setlength{\tabcolsep}{3.5pt}
    \renewcommand{\arraystretch}{1.08}
    \caption{
        Practical requirements and guidelines for Zero-Run privacy auditing.
    }
    \label{tab:zerorun_requirements}

    \begin{tabularx}{\textwidth}{@{}
        >{\raggedright\arraybackslash}p{0.47\textwidth}
        >{\raggedright\arraybackslash}X
    @{}}
        \toprule
        \textbf{Requirement and rationale}
        &
        \textbf{Recommended practice}
        \\
        \midrule

        \textbf{Reliable membership labels.}
        The audit requires known members and non-members; uncertainty or
        errors in these labels introduce an additional source of bias not
        covered by the current analysis.
        &
        Use examples whose membership status is established independently
        of the attack, e.g., from dataset manifests, known training/test
        splits, or synthetic generation.
        \\

        \subtlerule

        \textbf{Well-defined audit sampling.}
        The propensity score
        \mbox{$\pi(x)=\proba{S=1\mid X=x}$}
        depends on how members and non-members are sampled into the audit
        population.
        &
        Specify the member and non-member sampling procedures and estimate
        propensities under the same audit design.
        \\

       \subtlerule

        \textbf{Local overlap.}
        When $\pi(x)$ is close to $0$ or $1$, membership can already be
        inferred from $x$ alone; without overlap, algorithmic leakage cannot
         be distinguished from distribution-shift leakage.
        &
        As an \textit{overlap diagnostic}, report
        (i) the propensity-score distribution using a histogram of
        $\hat{\pi}(X_i)$, and
        (ii) the effective audit sample size $\sum_i b_i$ remaining after the Zero-Run
        correction.
        \\

       \subtlerule

        \textbf{Leakage-free propensity features.}
        Features used to estimate $\pi(x)$ should capture distribution shift,
        not information leaked by the audited model itself; otherwise, true
        model leakage may be incorrectly attributed to the shift.
        &
        When estimating propensities, avoid target-model losses, logits, or
        representations. For foundation models, use metadata or
        representations from an independent public encoder.
        \\

        \subtlerule

        \textbf{Conservative propensity uncertainty.}
        Underestimating distribution-shift leakage can lead to an invalid
        audit, whereas overestimating it only makes the resulting privacy
        bound more conservative.
        &
        Account for propensity-estimation uncertainty, e.g., using
        conservative confidence bounds or the bootstrap procedure in
        \Cref{sec:propensity_estimation}. Release the corresponding propensity
        model and bounds to support falsifiable propensity scores and
        trustworthy audits.
        \\

        \subtlerule

        \textbf{Robustness to propensity-model choice.}
        Reported results depend on the propensity-score model used, including
        the text/image embedding and classification models.
        &
        Perform a \textit{sensitivity analysis} by reporting results across
        different propensity-score models and varying both the embedding and
        classification models.
        \\

        \subtlerule

        \textbf{Different sources of membership leakage.}
        Membership information can leak through both distribution shift and
        algorithmic leakage. Zero-Run removes features-only leakage so that
        the remaining privacy leakage is attributed only to the model.
        &
        As a \textit{sanity check}, report attack performance and empirical
        privacy leakage for
        (i) uncorrected results and
        (ii) features $X_i$ alone.
        The former should yield greater leakage than the corrected results;
        the latter quantifies leakage attributable to distribution shift
        alone.
        \\

        \bottomrule
    \end{tabularx}
\end{table*}

\section{Related Work}
\label{sec:related}

\paragraph{Interventional Privacy Auditing.}
\citet{ding2018detecting,jagielski2020auditing} introduced privacy auditing methods that rely on repeatedly running the audited algorithm with and without a specific data point. These \emph{Multi-Run} techniques were later refined to obtain tighter guarantees across various settings \citep{nasr2023tight,auditing_cebere}, see \citet{annamalai2025hitchhiker} for an overview. To reduce the computational cost of rerunning the audited algorithm, \citet{steinke2023privacy} proposed a \emph{One-Run} auditing method for $\eps$-DP and $(\eps,\delta)$-DP, later extended to $f$-DP by \citet{mahloujifar2024auditing}. While our proof techniques draw on these foundational works, their validity inherently depends on randomized membership assignments controlled by the auditor. This interventional requirement motivates the development of \emph{Zero-Run} auditing methods that can estimate post hoc privacy leakage from a deployed model.

\paragraph{Observational Privacy Auditing.} \citet{kalemaj2025observational} also audits without randomized insertion, but targets the weaker label-DP, assumes identical member and non-member feature distributions, and cannot correct for general feature-level shifts. More closely related is \pano \citep{kazmi2024panoramia}, which audits $(\eps,\delta)$-DP using non-members drawn from a generative model. \pano acts as a special case of our composition framework, which we strictly generalize in two ways: we support GDP auditing and employ milder overlap assumptions (\Cref{hyp:approx_overlap}) by permitting an approximation error $\delta_\DS>0$. When $\delta_\DS=0$, this recovers \pano's $c$-closeness assumption where $c=\log((1-\eta)/\eta)$. Furthermore, rather than learning a full generative model, our approach 
reduces the correction problem to estimating the propensity score function, i.e., a probabilistic binary classifier.

Specifically for LLMs, another observational approach \cite{rossi2026natural} leverages unique identifiers naturally present in the training data, such as cryptographic hashes, for which IID non-members can be generated directly, and then applies One-Run auditing.

\paragraph{Causal Inference and Distribution Shift.} 
Recent work \citet{even206membership} studies the evaluation of membership inference attacks under distribution shift using causal principles, primarily by reweighting population-level attack metrics such as AUC to remove distribution-shift effects. In contrast, our focus is on privacy auditing, where the goal is to obtain valid lower bounds on DP leakage.

More generally,
unlike causal reweighting approaches \citep{even2026membership} or generative approximations \citep{kazmi2024panoramia}, we do not attempt to construct non-members that approximately match the member distribution. Instead, we explicitly model distribution-shift leakage conditional on the observed examples and correct the audit accordingly. This requires only the data-generation assumption in \Cref{hyp:data_generation} and relies on local overlap between members and non-members, rather than requiring a global overlap condition.
As shown in \Cref{prop:impossibility}, Zero-Run auditing is impossible when there is no local overlap, i.e., when the two distributions have disjoint support.

\section{Discussion \& Conclusion}

\paragraph{Practical guidelines.}
We summarize practical guidelines and requirements for valid and reliable Zero-Run auditing  in \Cref{tab:zerorun_requirements}.
In particular, given the uncertainty inherent to observational inference, we emphasize the importance of sensitivity analyses, sanity checks, and assumption diagnostics. These practices are central to the \textit{Veridical Data Science} framework of \citet{yu2020veridical}, which turns toolbox methods into end-to-end processes spanning data generation, analysis, and data-derived conclusions.

\paragraph{Limitations.} Our zero-run framework relies on accurately estimating the propensity score. 
Improving the reliability of propensity estimation, although orthogonal to our contribution, is a fundamental challenge in causal inference, and advances in this area would directly benefit the tightness and reliability of our audits.

While our framework successfully enables post-hoc privacy evaluation, it naturally inherits the limitations of existing One-Run methods. First, as observed by the baseline gap in \Cref{fig:decomposition_main}, Zero-Run auditing is limited by the finite-sample inefficiency intrinsic to all One-Run methods, so empirical bounds will generally fall below the theoretical privacy level \cite{steinke2023privacy, mahloujifar2024auditing}. Improvements to One-Run auditing would directly benefit Zero-Run. Second, audit power depends critically on the underlying MIA. In heavily restricted black-box deployments, particularly for DP models \citep{nasr2023tight}, loss-based attacks may provide too little signal to obtain meaningful bounds, independently of our distribution-shift correction.
\\

Overall, the main lesson of our work is that post-hoc privacy auditing is possible, but it is fundamentally an observational inference problem. Distribution shift cannot be ignored: it is itself a source of membership information. Propensity scores provide a natural language for quantifying this source, and our results show how propensity bounds can be converted into valid empirical lower bounds on algorithmic privacy leakage.

\section*{Ethical Considerations}

We identify no ethical concerns specific to this study. All data and model checkpoints used in our experiments are publicly available, and the study involves no human-subject recruitment, private-data collection, or access to proprietary training systems. Zero-Run operates passively on fixed model outputs and datasets with independently established membership labels; it neither modifies nor retrains the audited model. The evaluation therefore does not expose previously inaccessible records or infer the membership of unidentified individuals. 

More broadly, the technique is intended to improve the ethics of model governance. Model providers currently control both their training pipelines and much of the evidence needed to substantiate privacy claims, leaving users and data contributors dependent on provider assurances. Zero-Run enables independent third parties to evaluate those claims after deployment, using public artifacts and without requiring privileged access or voluntary cooperation from the model trainer. This reduces the existing information and power asymmetry, makes privacy claims empirically falsifiable, and strengthens accountability to people whose data may have been used for training. We therefore view the work as introducing minimal additional risk while providing a clear privacy and public-interest benefit.

\section*{Open Science}

For review, we provide an anonymized artifact at
\url{https://anonymous.4open.science/r/zero_run_usenix-DE61/README.md}.
It provides a minimal reproducibility path containing the Zero-Run implementation
and propensity score estimation, including \texttt{zero\_run\_audit} (the audit
core) and \texttt{PropEstimator} (cross-fitted propensity estimation).
Datasets and frozen configurations are documented under \texttt{experiments/},
with additional details available in the repository README. We will also provide
complete tables of hyperparameter values in the repository to facilitate
reproducibility, as well as the code used to generate the synthetic data for
\texttt{Cifar10} and \texttt{WikiText}.

Because the full set of artifacts is large (over 500\,GB), we provide a minimal
subset of results for reproducibility under \texttt{minimal\_artifacts/}, together
with corresponding documentation. The artifact's \texttt{REPRODUCE.md} provides
the exact command sequence for reproducing most of the results. Full reproduction,
however, requires substantial distributed computation due to cross-fitting,
bootstrapping, and model training. We are committed to open-sourcing the complete
implementation and making the full artifacts publicly available.

\bibliography{references}
\bibliographystyle{plainurl}

\newpage
\appendix
\onecolumn

\section{Additional General Content and Proofs }

\subsection{Privacy Loss Random Variable}

\begin{definition}\label{def:PLRV}
    Define the Privacy Loss Random Variable (PLRV for short) as \citep{sommer2018privacy}:
    \begin{equation}
        \PLRV_{\cD,\cD'}^\cA(\theta) = \log\left(\frac{\proba{\cA(\cD) = \theta}}{\proba{\cA(\cD') = \theta}}\right)\,,
    \end{equation}
    for two adjacent datasets $\cD$ and $\cD'$ and any mechanism $\cA$. 
\end{definition}

\subsection{Membership Inference Attack Meta-Algorithm}
\label{app:mia}

\Cref{alg:mia} describes the meta-algorithm for MIAs: a confidence score $Y_i=c(\theta, X_i)\in\R$ is first computed for every data point $X_i$ (e.g., the opposite of the loss of $\theta$ at $X_i$). The larger $Y_i$, the more confident the attacker is about $X_i$ begin a member (i.e., $S_i=1$).
Guesses $g_i$ are then made using these scores for all data points. For instance, a threshold $t$ can be computed, and the guess can be set as $g_i=\one_\set{Y_i\geq t} - \one_\set{Y_i<t}$. In order to keep data points for which we are most confident about their membership or non-membership, an \textit{abstention} step is finally performed, by setting low-confidence guesses to 0, thereby obtaining a vector of guesses and abstentions $T\in\set{-1,0,1}$.

\begin{algorithm}[H]
    \caption{Membership inference attack}
    \label{alg:mia}
    \begin{algorithmic}[1]
        \STATE \textbf{Inputs:} 
        Instance $\theta=\mathcal{A}(\mathcal{D}_{\mathrm{train}})\in\Theta$ of the algorithm $\cA$ to audit;
        $\mathcal{D}_1\subset \mathcal{D}_{\mathrm{train}}$ and
        $\mathcal{D}_0\subset \mathcal{X}\setminus\mathcal{D}_{\mathrm{train}}$ with $|\cD_1\cup\cD_0|=m$; MIA $\mathcal{M}=(c,g,\phi)$ with confidence
        score function $c:\Theta\times \mathcal{X}\to \mathbb{R}$, 
        guess function $g:\mathcal{X}\times \mathbb{R}\to \{-1,1\}$, 
        abstention function $\phi:(\{-1,1\}\times \mathbb{R}\times \mathcal{X})^m \to \{-1,0,1\}^m$
        
        \STATE Set $\{X_i : i\in[m]\} = \mathcal{D}_1 \cup \mathcal{D}_0$
        
        \FOR{$i \in [m]$}
            \STATE $Y_i \gets c(\theta, X_i)$
            \STATE $g_i \gets g(X_i, Y_i)$
        \ENDFOR
        \STATE $(T_i)_{i\in[m]} \gets 
        \phi\!\left((g_i)_{i\in[m]}, (Y_i)_{i\in[m]}, (X_i)_{i\in[m]}\right)\in\set{-1,0,1}^m$
        
        \STATE \textbf{Return} $(T_i)_{i\in[m]}$
    \end{algorithmic}
\end{algorithm}

\subsection{The Central Role of Propensity Scores in Privacy Auditing with Distribution Shift}
\label{app:role_propensity}

Next results motivate the role of $\pi$ in our analysis by showing that when $\P_1\ne\P_0$, the classical inequality relating true positive and true negative rates (TPR and FPR) for a MIA on an $f$-DP algorithm is altered, as shown by the following result.
The proof is particularly instructive due to its simplicity, illustrating how propensity scores naturally arise when auditing under distribution shift.

\begin{proposition}
\label{prop:inequality_f_DP_independent}
    Assume that \Cref{hyp:data_generation} holds and that $(X_i,S_i,T_i)_{i\in[m]}$ are independent.
    Then,
    \begin{equation}
    \begin{aligned}
        &\esp{\frac{A_i}{\pi(X_i)}\proba{T_i=1|S_i=1,X_i}} \\
        &\leq \bar f\left( \esp{\frac{1-A_i}{1-\pi(X_i)}\proba{T_i=1|S_i=-1,X_i}}\right)\,,
    \end{aligned}
    \end{equation}
    where $A_i=\one_\set{S_i=1}=\frac{S_i+1}{2}$.
\end{proposition}

\begin{proof}[Proof of \Cref{prop:inequality_f_DP_independent}]
Denoting $\bar X = (X_i)_{i\in[m]}$, $\bar S = (S_i)_{i\in[m]}$ and $\bar X^{(-i)},\bar S^{(-i)}$ the same vectors with coordinate $i$ removed, we have:
\begin{align*}
    \proba{T_i=1|S_i=1,\bar S^{(-i)},\bar X} \leq \bar f\left(\proba{T_i=1|S_i=-1,\bar S^{(-i)},\bar X}\right)\,,
\end{align*}
by definition of $f-$DP.
Now, let's take the average over $(\bar S^{(-i)},\bar X^{(-i)})$ and use the independence assumption made in \Cref{prop:inequality_f_DP_independent}:
\begin{align*}
    &\E_{\bar S^{(-i)},\bar X^{(-i)}}\proba{T_i=1|S_i=1,\bar S^{(-i)},\bar X}\\
    &\leq \E_{\bar S^{(-i)},\bar X^{(-i)}}\bar f\left(\proba{T_i=1|S_i=-1,\bar S^{(-i)},\bar X}\right)\\
    &\leq\bar f\left( \E_{\bar S^{(-i)},\bar X^{(-i)}}\proba{T_i=1|S_i=-1,\bar S^{(-i)},\bar X}\right)\,,
\end{align*}
by concavity of $\bar f$, leading to:
\begin{align*}
    \proba{T_i=1|S_i=1,X_i}\leq \bar f\left(\proba{T_i=1|S_i=-1,X_i}\right)\,.
\end{align*}
Now, taking the average on both sides for $X_i$ sampled over the whole population:
\begin{align*}
        \E\proba{T_i=1|S_i=1,X_i}&\leq  \E\bar f\left(\proba{T_i=1|S_i=-1,X_i}\right)\\
        &\leq   \bar f\left(\E\proba{T_i=1|S_i=-1,X_i}\right)\,.
\end{align*}
We now use the fact that $\esp{\frac{A_i}{\pi(X_i)}|X_i}=\esp{\frac{1--A_i}{1-\pi(X_i)}|X_i}=1$, to obtain the final result.
\end{proof}

\subsection{Impossibility Result}
\label{app:impossibility}

A $p$-valid audit algorithm for the privacy hypothesis '$f$-DP' is abstracted as a mapping:
\begin{equation*}
    \cB_f: (\cD_0,\cD_1,\theta)\in\cP(\cX)\times\cP(\cX)\times\Theta\mapsto b\in\set{\textsc{True},\textsc{False}}\,,
\end{equation*}
which satisfies, for any $f$-DP algorithm $\cA$, for any $\dtrain\in\cP(\cX)$, any $\cD_1\subset\dtrain$ and $\cD_0\subset\cX\setminus\dtrain$:
\begin{equation*}
    \proba{\cB_f(\cD_0,\cD_1,\cA(\cD)) = \textsc{False}}\leq p\,.
\end{equation*}
In other words, $\cB_f$ confidently rejects the privacy hypothesis with probability $1-p$, if \textsc{False} is outputted.

Next proposition shows that if the overlap assumption is no satisfied, no $p$-valid audit algorithm can make meaningful audits.

\begin{proposition}[Impossibility]
\label{prop:impossibility}
   For any two non-overlapping distributions $\P_1,\P_0$, there exists an algorithm $\cA'$ which does not satisfy $\eps'$-DP for any $\eps'>0$, and such that if $\cD_1$ and $\cD_0$ are respectively samples drawn from $\P_1$ and $\P_0$, any $p$-valid audit algorithm for any $\eps$-DP hypothesis yields:
    \begin{equation*}
    \proba{\cB_{\eps}(\cD_0,\cD_1,\cA'(\cD_1)) = \textsc{False}}\leq p\,.
    \end{equation*}
    In other words, no audit algorithm can reject confidently any privacy level $\eps$, despite $\cA'$ not satisfying any $\eps'$-DP.. 
\end{proposition}

\begin{proof}
    To simplify, assume that $\cX=\R$.
    Let $\cX=\cX_0\cup\cX_1$ with $\cX_0\cap\cX_1=\emptyset$, and $\P_0$ and $\P_1$ be distributions having respective supports $\cX_0$ and $\cX_1$.
    Let $\cA(\cD)$ be the algorithm defined as, for any dataset $\cD\subset\cX$:
    \begin{equation*}
        \cA(\cD):x\in\cX\mapsto\one_\set{x\in\cX_1}\,.
    \end{equation*}
    $\cA$ is $0$-DP since it is deterministic and does not depend on the dataset $\cD$.
    Let then $\cA'(\cD)$ be the algorithm defined as, for any datset $\cD\subset\cX$:
    \begin{equation*}
        \cA'(\cD):x\in\cX\mapsto\one_\set{x\in\cX_1\text{ or }x\in\cD}\,.
    \end{equation*}
    $\cA'$ is not $\eps$-DP for any $\eps'>0$ (take $\cD=\set{x}$ with $x\in\cX_1$ and $x\in\cX_0$). 
    However, for $\cD_1\subset\cX_1$, we always have that $\cA(\cD_1)=\cA(\cD'_1)$.
    Thus,
    \begin{equation*}
    \proba{\cB_{\eps}(\cD_0,\cD_1,\cA'(\cD_1)) = \textsc{False}}=\proba{\cB_{\eps}(\cD_0,\cD_1,\cA(\cD_1)) = \textsc{False}}\leq p\,,
    \end{equation*}
    for any $p$-valid audit algorithm $\cB_\eps$ for the $\eps$-DP privacy hypothesis.

\end{proof}

\section{Composition-based Privacy Auditing}
\label{app:composition-audit}

\subsection{Propensity Scores and Privacy Parameters of $\cM_\DS$}

We now show the following result, linking the propensity score with privacy parameters of the distribution shift mechanism $\cM_\DS$.

\begin{proposition}
    \label{prop:DP_of_M_DS}
    $\cM_\DS$ is $\eps_\DS$-DP with $\eps_\DS = \sup_{x\in\cX} \left| \log\left(\frac{\pi(x)}{1-\pi(x)}\right)\right|$ and $(\eps, \delta_\DS(\eps))$-DP with
    $\delta_\DS(\eps) = \max_{a\in\set{0,1}}\P_a\left(\pi(X_i)^{2a-1}(1-\pi(X_i))^{1-2a}>e^\eps\right)$, for all $\eps>0$.
\end{proposition}

\begin{proof}[Proof of \Cref{prop:DP_of_M_DS}]
    $\cM_\DS$ is $(\eps,\delta)$-DP with:
    \begin{equation*}
        \delta = \sup_{s\sim s'}\esp{(1-e^{\eps-\PLRV_{s,s'}^{\cM_\DS}(X)})_+}\,,
    \end{equation*}
    for any two adjacent bit datasets $s\sim s'$. Now, for $s\sim s'$ adjacent on bit $i\in[m]$, we have that:
    \begin{equation*}
        \esp{(1-e^{\eps-\PLRV_{s,s'}^{\cM_\DS}(X)})_+} \leq \max(\E_1[(1-e^{\eps-\eps_\DS(X_i)})_+],\E_0[(1-e^{\eps-\eps_\DS(X_i)})_+] )\,,
    \end{equation*}
    for $X_i\sim \cP_1$ if $s_i=1$, $X_i\sim\cP_0$ otherwise, since $\PLRV_{s,s'}^{\cM_\DS}(X)\leq \eps_\DS(X_i)$. Finally, $\max(\E_1[(1-e^{\eps-\eps_\DS(X_i)})_+],\E_0[(1-e^{\eps-\eps_\DS(X_i)})_+] )\leq \delta_\DS(\eps)$.
\end{proof}

\subsection{Proof of \Cref{thm:composition}}
\label{app:proof_composition_thm}

\begin{proof}[Proof of \Cref{thm:composition}]
    Under \Cref{hyp:data_generation}, we have that $X_i\sim \cD_1$ if $S_i=1$ and $X_i\sim\cD_0$ if $S_i=-1$, and $\dtrain = \cD\cup\set{X_i:S_i=1}$.
    Thus, we have that $(X_i)_{i\in[m]}\sim \cM_\DS(S)$ and $\theta\eqdef \cA(\dtrain)\sim\cM_\cA(S)$, so that we indeed have $(X,\theta)\sim\cM_\Tot(S)$.
    The randomness in the mechanisms $\cM_\DS$ and $\cM_\cA$ are independent from eachother, and thus $f_\Tot\geq f_{\cM_\DS}\otimes f_{\cA}$  using \citet[Theorem 3.2]{Dong22GDP} (adaptive composition theorem).
\end{proof}

\subsection{Proof of \Cref{cor:worst-case_audit}}

\begin{proof}[Proof of \Cref{cor:worst-case_audit}]
    The proof follows from \citep[Section 3]{Dong22GDP}:
    \begin{equation*}
        f_{\eps,\delta} \otimes f_{\eps',\delta'} = f_{\eps+\eps',1-(1-\delta)(1-\delta')}\,,
    \end{equation*}
    together with
    \begin{equation*}
        G_{\mu}\otimes G_{\mu'} = G_{\sqrt{\mu^2+{\mu'}^2}}\,,
    \end{equation*}
    using \Cref{prop:DP_of_M_DS} and \Cref{thm:composition}.
\end{proof}

\subsection{$(\eps,\delta)$-DP Auditing Algorithm}
\label{app:worst-case_bounds_algo_eps_delta}

We use \Cref{cor:worst-case_audit} and \citet{steinke2023privacy}'s audit algorithm for $(\eps,\delta)$-DP. To simplify the algorithm, we assume that the MIA outputs guesses $(T_i)_{i\in[m]}$ which satisfy $\NRM{T}_0=r$ almost surely, for some fixed $r\in[m]$.
This can be ensured by simply maximizing $\sum_i Y_i T_i$ over $T\in\set{-1,0,1}^m$ verifying $\set{T_i=\pm1}=k_\pm$, for $k_++k_-=r$.
\Cref{alg:audit_approx_DP_worst_case} instantiates \citet{steinke2023privacy}'s audit algorithm under this assumption.
\Cref{prop:audit_approx_DP_worst_case} then follows from \citet[Corollary 5.4]{steinke2023privacy}.

\begin{algorithm}[H]
    \caption{Composition-based Zero-Run $\eps$-DP auditing
    \label{alg:audit_approx_DP_worst_case}
    }
    \begin{algorithmic}[1]
      \STATE \textbf{Inputs:} Instance $\theta=\mathcal{A}(\dtrain)$ of the algorithm $\cA$ to audit; MIA $\cM$; $\cD_1\subset \dtrain$; $\cD_0\subset \cX\setminus\dtrain$; probability of error $p$; number of guesses $r\in[m]$; privacy parameter $\eps,\delta>0$, $\bar\eps_\DS,\delta_\DS>0$
      \STATE $\set{X_1,\dots,X_m}\gets\cD_1\cup\cD_0$
      \STATE $S_i\gets 2\one_\set{X_i\in\cD_1} -1 \in\set{-1,1}$
      \STATE Run \Cref{alg:mia} on $(\theta, \cD_1, \cD_0, \cM)$ to obtain $(T_i)_{i\in[m]}$
      \STATE Compute $W=\sum_{i=1}^m \max(0,T_iS_i)$
      \STATE Estimate $\hat v\in\R$ such that:
      $$ g(\hat v)+  2m(\delta+\delta_\DS-\delta\delta_\DS)\max_{i\in[m]}\set{\frac{g(\hat v)-g(\hat v-i)}{i}} \leq p \,,$$
      where
      $$g(u)\eqdef \proba{\mathrm{Binomial}\left(r,\frac{e^{\eps+\bar\eps_\DS}}{1+e^{\eps+\bar\eps_\DS}} \right) \geq u}$$
      \IF{ $W\geq \hat v$}
      \STATE \textbf{Return} \textsc{False}
      \ENDIF
    \end{algorithmic}
\end{algorithm}

\begin{proposition}\label{prop:audit_approx_DP_worst_case}
    Let $\eps,\delta>0$.
    Assume that \Cref{hyp:data_generation,hyp:approx_overlap} hold and that $\NRM{T}_0=r$ almost surely. Then, \Cref{alg:audit_approx_DP_worst_case} is $p$-valid:
    if \Cref{alg:audit_approx_DP_worst_case} outputs \textsc{False}, $\cA$ is not $(\eps,\delta)$-DP with probability $1-p$. 
\end{proposition}

\subsection{$\mu$-GDP Auditing Algorithm}
\label{app:worst-case_bounds_algo_GDP}

We use \Cref{cor:worst-case_audit} and \citet{mahloujifar2024auditing}'s audit algorithm for general $f$-DP. We assume that the MIA outputs guesses $(T_i)_{i\in[m]}$ which satisfy $\NRM{T}_0=r$ almost surely, for some fixed $r\in[m]$.
This can be ensured by simply maximizing $\sum_i Y_i T_i$ over $T\in\set{-1,0,1}^m$ verifying $\set{T_i=\pm1}=k_\pm$, for $k_++k_-=r$.
\Cref{alg:audit_approx_DP_worst_case} instantiates \citet{mahloujifar2024auditing}'s audit algorithm under this assumption, for the tradeoff function $G'=G_{\sqrt{\mu^2+\mu_\DS^2}} \otimes f_{\delta_\DS}$, which writes as:
\begin{equation}\label{eq:G_prime}
    G'(p)=\left\{
    \begin{aligned}
        &(1-\delta_\DS)G_{\sqrt{\mu^2+\mu_\DS^2}}(\tfrac{p}{1-\delta_\DS})\quad \text{if} \quad 0\leq p \leq 1-\delta_\DS\,,\\
        &0\qquad \text{else.}
    \end{aligned}
    \right.
\end{equation}
\Cref{prop:audit_GDP_worst_case} then follows from \citet[Theorem 10]{mahloujifar2024auditing}.

\begin{algorithm}
    \caption{Composition-based Zero-Run $\mu$-GDP auditing I
    \label{alg:audit_GDP_worst_case}
    }
    \begin{algorithmic}[1]
      \STATE \textbf{Inputs:} 
      Instance $\theta=\mathcal{A}(\dtrain)$ of the algorithm $\cA$ to audit; MIA $\cM$; $\cD_1\subset \dtrain$; $\cD_0\subset \cX\setminus\dtrain$; probability of errors $p$ and $p'$; number of guesses $r$; number of correct guesses $c$
      \STATE $\set{X_1,\dots,X_m}\gets\cD_1\cup\cD_0$
      \STATE $S_i\gets 2\one_\set{X_i\in\cD_1} -1 \in\set{-1,1}$
      \STATE Run \Cref{alg:mia} on $(\theta, \cD_1, \cD_0, \cM)$ to obtain $(T_i)_{i\in[m]}$
      \STATE Set $h[k]=r[k]=0$ for $k\in[c]$
      \STATE Set $r[c]=\frac{pc}{m}$ and $h[c]=\frac{p(r-c)}{m}$
      \FOR{$k\in\set{c-1,\ldots,0}$}
        \STATE $h[k] = \bar {G'}^{-1}(r[k+1]) $ with $\bar {G'}^{-1}$ the inverse of 
        $\bar {G'}=1-G'$ in the tradeoff function sense ($G'$ defined in \Cref{eq:G_prime})
        \STATE $r[k] = r[k+1] + \frac{k}{r-k}(h[k]-h[k+1])$
      \ENDFOR
      \IF{ $r[0] + h[0] \geq \frac{r}{m}$}
      \STATE Return \textsc{False} if the number of correct guesses is equal to or larger than $c$
      \ENDIF
    \end{algorithmic}
  \end{algorithm}

\begin{proposition}\label{prop:audit_GDP_worst_case}
    Let $\mu>0$.
    Assume that \Cref{hyp:data_generation,hyp:approx_overlap} hold and that $\NRM{T}_0=r$ almost surely. Then, \Cref{alg:audit_GDP_worst_case} is $p$-valid:
    if \Cref{alg:audit_GDP_worst_case} outputs \textsc{False}, $\cA$ is not $\mu$-GDP with probability $1-p$. 
\end{proposition}

The error probability $\delta_\DS$ can also be put in the error probability of the audit, allowing us to use \citet{mahloujifar2024auditing}'s auditing algorithm for the GDP tradeoff function, yielding a $(p+\delta_\DS)$-valid audit (instead pf $p$-valid). This is what we describe in the discussion after \Cref{cor:worst-case_audit}, and what next algorithm and corollary formalize.

\begin{algorithm}
    \caption{Composition-based Zero-Run $\mu$-GDP auditing II 
    \label{alg:audit_GDP_worst_case_error_proba}
    }
    \begin{algorithmic}[1]
      \STATE \textbf{Inputs:} 
      Instance $\theta=\mathcal{A}(\dtrain)$ of the algorithm $\cA$ to audit; MIA $\cM$; $\cD_1\subset \dtrain$; $\cD_0\subset \cX\setminus\dtrain$; probability of errors $p$ and $p'$; parameter $\mu$; number of guesses $r$; number of correct guesses $c$
      \STATE $\set{X_1,\dots,X_m}\gets\cD_1\cup\cD_0$
      \STATE $S_i\gets 2\one_\set{X_i\in\cD_1} -1 \in\set{-1,1}$
      \STATE Run \Cref{alg:mia} on $(\theta, \cD_1, \cD_0, \cM)$ to obtain $(T_i)_{i\in[m]}$
      \STATE Set $h[k]=r[k]=0$ for $k\in[c]$
      \STATE Set $r[c]=\frac{pc}{m}$ and $h[c]=\frac{p(r-c)}{m}$
      \FOR{$k\in\set{c-1,\ldots,0}$}
        \STATE $h[k] = \bar {G}^{-1}_\mu(r[k+1]) $ with $\bar {G}_\mu^{-1}$ the inverse of 
        $\bar {G}_\mu=1-G_\mu$ in the tradeoff function sense ($G_\mu$ is the $\mu$-GDP tradeoff function)
        \STATE $r[k] = r[k+1] + \frac{k}{r-k}(h[k]-h[k+1])$
      \ENDFOR
      \IF{ $r[0] + h[0] \geq \frac{r}{m}$}
      \STATE Return \textsc{False} if the number of correct guesses is equal to or larger than $c$
      \ENDIF
    \end{algorithmic}
  \end{algorithm}

\begin{proposition}\label{prop:audit_GDP_worst_case_error_proba}
    Let $\mu>0$.
    Assume that \Cref{hyp:data_generation,hyp:approx_overlap} hold and that $\NRM{T}_0=r$ almost surely. Then, \Cref{alg:audit_GDP_worst_case} is $(p+\delta_\DS)$-valid:
    if \Cref{alg:audit_GDP_worst_case_error_proba} outputs \textsc{False}, $\cA$ is not $\mu$-GDP with probability $1-p-\delta_\DS$. 
\end{proposition}

\section{Conditional $(\eps,\delta)$-Auditing}
\label{app:proof_eps_delta_cond_tampered}

\subsection{Conditional $\eps$-DP Auditing}
\label{app:proof_pure_DP_tampered}

We first present a relaxed version of \Cref{thm:conditional_audit_tampered}.\ref{thm:conditional_audit_tampered:eps_delta} with $\delta=0$. We include this simplified result for its proof simplicity and clarity of exposure, as well as for its use for rejecting propensity score models, as explained in next paragraph.

\begin{remark}[Rejecting propensity score models]
\label{remark:propensity_rejection}
    \Cref{thm:pure_DP_auditing_tampered} below can additionally be used to audit a propensity score model $\hat\pi$, to reject the following hypothesis:
    \begin{equation*}
        \cH_{\hat\pi}\,:\qquad \max(\pi(x),1-\pi(x)) \leq \max(\hat\pi(x),1-\hat\pi(x))\,, \qquad \forall x\,.
    \end{equation*}
    Suppose a third-party auditor who wishes to reject an audit performed with $\hat\pi$ manages to build a stronger propensity score model $\tilde \pi$.
    Let the $\tilde\pi$-MIA be defined as thresholding on scores $\tilde\pi(X_i)$ and attributing $T_i=\one_\set{\tilde\pi(X_i)<p_-}+\one_\set{\tilde\pi(X_i)>p_+}$.
    
    Setting $\eps=0$ and $b_i$ obtained using $\hat\pi$ in \Cref{thm:pure_DP_auditing_tampered} below, having a large success rate for the tampered $\tilde\pi$-MIA (which only hinges on distribution shift) means that $\hat\pi$ does not not penalize enough distribution shifts, therefore rejecting $\cH_{\hat\pi}$.

\end{remark}

\begin{theorem}[Tampered $\eps$-DP Auditing]\label{thm:pure_DP_auditing_tampered}
    Assume that \Cref{hyp:data_generation} holds.
    and let $(B_i)_{i\in[m]}$ be independent Bernoulli random variables of parameters $(\frac{1+e^{-\eps-\eps_\DS(X_i)}}{1+e^{-\eps}})_{i\in[m]}$, independent of all the rest.
    Then, for any $v>0$, if $\cA$ is $\eps$-DP:
    \begin{align*}
        \proba{\sum_{i=1}^m B_i \one_\set{T_i=S_i} \geq v \Big| (X_i,T_i)_{i\in[m]}} &\leq   \proba{ \mathrm{Binomial}\left(\NRM{T}_1 , \frac{e^\eps}{1+e^\eps} \right) \geq v\Big| (X_i,T_i)_{i\in[m]}}\,.
    \end{align*}
\end{theorem}

\begin{proof}[Proof of \Cref{thm:pure_DP_auditing}]
	The proof is largely inspired by that of \citet[Proposition 5.1]{steinke2023privacy}.
	Let ${\theta\in\cE}$ be an event of non-null probability.
	Fix some $i\in[m]$ and $s^{(-i)}\in\set{-1,1}^m$.
	We first start by proving an upper-bound on:
	\begin{equation*}
		\proba{S_i=1 | \theta\in\cE, S^{(-i)} = s^{(-i)},\bar X }\,.
	\end{equation*} 
	Using Bayes' theorem,
	\begin{align*}
		&\proba{S_i=1 | \theta\in\cE, S^{(-i)} s^{(-i)},\bar X } \\
        &= \frac{ \proba{ \theta\in\cE | S_i=1 , S^{(-i)} = s^{(-i)},\bar X } \times \proba{S_i=1 | S^{(-i)} = s^{(-i)},\bar X }}{ \proba{ \theta\in\cE | S^{(-i)} = s^{(-i)},\bar X } } \\
		\quad& =  \frac{ \proba{ \theta\in\cE | S_i=1 , S^{(-i)} = s^{(-i)},\bar X } \times \proba{S_i=1 | S^{(-i)} = s^{(-i)},\bar X }}{ \sum_{s\in\set{-1,1}}\proba{ \theta\in\cE | S_i=s, S^{(-i)} = s^{(-i)},\bar X } \proba{ S_i=s | S^{(-i)} = s^{(-i)},\bar X } } \\
		\quad& =  \frac{ \pi(X_i) \proba{ \theta\in\cE | S_i=1 , S^{(-i)} = s^{(-i)},\bar X } }{ \pi(X_i)\proba{ \theta\in\cE| S_i=1, S^{(-i)} = s^{(-i)},\bar X }  + (1-\pi(X_i)\proba{ \theta\in\cE| S_i=-1, S^{(-i)} = s^{(-i)},\bar X }  } \,,
	\end{align*}
	since $\proba{ S_i=1 | S^{(-i)} = s^{(-i)},\bar X } = \pi(X_i)$   and  $\proba{ S_i=-1 | S^{(-i)} = s^{(-i)},\bar X } = 1-\pi(X_i)$.
	Then, using the fact that due to $\eps$-DP we have that:
	\begin{equation*}
		R \eqdef  \frac{\proba{ \theta\in\cE | S_i=1 , S^{(-i)} = s^{(-i)},\bar X } }{ \proba{ \theta\in\cE | S_i=-1 , S^{(-i)} = s^{(-i)},\bar X }} \in \left[ \frac{1}{e^\eps} , e^\eps\right]\,,
	\end{equation*}
	we have:
	\begin{align*}
		\proba{S_i=1 | \theta\in\cE, S^{(-i)} = s^{(-i)},\bar X } & \leq \frac{ \pi(X_i) }{ \pi(X_i)  + (1-\pi(X_i))R^{-1} } \\
        & \leq \frac{ \pi(X_i) }{ \pi(X_i)  + (1-\pi(X_i))e^{-\eps} } \\
		&\leq \frac{1}{ 1  + \frac{1-\pi(X_i)}{\pi(X_i)}e^{-\eps} }\\
		&\leq \frac{1}{1+e^{-\eps-\eps_i}}\,,
	\end{align*}
    where $e^{\eps_i}=\max(\frac{\pi(X_i)}{1-\pi(X_i)},\frac{1-\pi(X_i)}{\pi(X_i)})$.
	Similarly, we have that:
	\begin{align*}
		\proba{S_i=-1 | \theta\in\cE, S^{(-i)} = s^{(-i)},\bar X } & \leq \frac{1}{1+e^{-\eps-\eps_i}}\,.
	\end{align*}
	We can thus prove by induction that, conditionally on $(X_j,T_j)_{j\in[m]}$, for all $i\in[m]$, $W_i = \sum_{j\leq i} \max(0,T_jS_j)$ is stochastically dominated by $\tilde W_j = \sum_{j\leq i} |T_j|\xi_j$, where $\xi_j$ are independent Bernoulli random variables of parameters $\frac{1}{1+e^{-\eps-\eps_j}}$.
	For $i=1$ this is of course true.
	Then, if we assume that the result holds at some $i\leq n$, we have that $W_{i+1}=W_i +  \one_\set{S_{i+1} = \mathrm{sign}(T_{i+1})}$ and $\tilde W_{i+1}=\tilde W_i +  \xi_{i+1}$.
    Conditionally on the event $\set{W_i=w_i,\tilde W_i=\tilde w_i,\bar T=\bar t}$ for $w_i,\tilde w_i,\bar t$ such that this event is not of probability 0:
    \begin{align*}
        &\esp{\one_\set{S_{i+1} = \mathrm{sign}(T_{i+1})} |W_i=w_i,\tilde W_i=\tilde w_i,\bar T=\bar t}\\
        &= \proba{S_{i+1} = \mathrm{sign}(T_{i+1}) |W_i=w_i,\tilde W_i=\tilde w_i,\bar T=\bar t}\\
        &=  \proba{S_{i+1} = T_{i+1}|W_i=w_i,\tilde W_i=\tilde w_i,\bar T=\bar t}\\
        &=  \proba{S_{i+1} = 1|W_i=w_i,\tilde W_i=\tilde w_i,\bar T=\bar t,T_{i+1}=1}\proba{T_{i+1}=1|W_i=w_i,\tilde W_i=\tilde w_i,\bar T=\bar t}\\
        &\quad+ \proba{S_{i+1} = -1|W_i=w_i,\tilde W_i=\tilde w_i,\bar T=\bar t,T_{i+1}=-1}\proba{T_{i+1}=-1|W_i=w_i,\tilde W_i=\tilde w_i,\bar T=\bar t}\\
        &\leq \frac{1}{1+e^{-\eps-\eps_i}}\proba{T_{i+1}=1|W_i=w_i,\tilde W_i=\tilde w_i,\bar T=\bar t}\\
        &\quad+ \frac{1}{1+e^{-\eps-\eps_i}}\proba{T_{i+1}=-1|W_i=w_i,\tilde W_i=\tilde w_i,\bar T=\bar t}\,,
    \end{align*}
    using the above computations with the adequate event $\cE$, which leads to:
        \begin{align*}
        \esp{\one_\set{S_{i+1} = T_{i+1}} |W_i=w_i,\tilde W_i=\tilde w_i,\bar T=\bar t}&\leq \frac{|t_{i+1}|}{1+e^{-\eps-\eps_i}}\,.
    \end{align*}
    Then, noticing that:
    \begin{align*}
        \esp{\xi_{i+1}|W_i=w_i,\tilde W_i=\tilde w_i,\bar T=\bar t} &=\frac{|t_{i+1}|}{1+e^{-\eps}}\\
        &=\frac{|t_{i+1}|}{1+e^{-\eps}}\times \frac{1+e^{-\eps}}{1+e^{-\eps-\eps_i}}\\
        &=\frac{|t_{i+1}|}{1+e^{-\eps}}\times \frac{1}{\esp{B_{i+1}|W_i=w_i,\tilde W_i=\tilde w_i,\bar T=\bar t}}\\
        &=\frac{\esp{|T_{i+1}|\xi_{i+1}|W_i=w_i,\tilde W_i=\tilde w_i,\bar T=\bar t}}{\esp{B_{i+1}|W_i=w_i,\tilde W_i=\tilde w_i,\bar T=\bar t}}\,,
    \end{align*}
    we have:
    \begin{align*}
    &\esp{B_{i+1}\one_\set{S_{i+1} = T_{i+1}} |W_i=w_i,\tilde W_i=\tilde w_i,\bar T=\bar t}\\
        &\esp{B_{i+1}|W_i=w_i,\tilde W_i=\tilde w_i,\bar T=\bar t}\esp{\one_\set{S_{i+1} = T_{i+1}} |W_i=w_i,\tilde W_i=\tilde w_i,\bar T=\bar t}\\
        &\leq \esp{|T_{i+1}|\xi_{i+1}|W_i=w_i,\tilde W_i=\tilde w_i,\bar T=\bar t}\,.
    \end{align*}
    Thus, $B_{i+1}\one_\set{S_{i+1} = T_{i+1}}$ is stochastically dominated by $|T_{i+1}|\zeta_{i+1}$ conditionally on $W_i$, $\tilde W_i$ and $T_{\leq i+1}$. Using \citet[Lemma 4.9]{steinke2023privacy}: $W_{i+1}$ is stochastically dominated by $\tilde W_{i+1}$, conditionally on $W_i$, $\tilde W_i$ and $T_{\leq i+1}$.
    We conclude by unrolling the induction.
\end{proof}

\subsection{Proof of \Cref{thm:conditional_audit_tampered}.\ref{thm:conditional_audit_tampered:eps_delta}}

We now proceed to the proof of \Cref{thm:conditional_audit_tampered}.\ref{thm:conditional_audit_tampered:eps_delta}, starting with the following intermediate result.

\subsubsection{Intermediate Result}

We first start with the following result, from which \Cref{thm:conditional_audit_tampered}.\ref{thm:conditional_audit_tampered:eps_delta} will directly follow, byt plugging in $\gamma=\eps$ and $\delta'=(1+e^{-1})\delta$, as we then prove in \Cref{prop:hockey}.

\begin{theorem}[Intermediate Result]
\label{thm:unified_auditing_better_tampered}
    Let $\gamma>0$ and $\delta'>0$ such that for all adjacent datasets $\cD\sim\cD'$,
    \begin{equation*}
        \esp{(1-e^{\gamma-|\PLRV_{\cD,\cD'}^\cA(\cA(\cD))|})_+} \leq \delta'\,.
    \end{equation*}
    Assume that \Cref{hyp:data_generation} holds.
    Let $(B_i)_{i\in[m]}$ be independent Bernoulli random variables of parameters $(\frac{1+e^{-\eps-\eps_\DS(X_i)}}{1+e^{-\eps}} )_{i\in[m]}$.
    Assume that the attacker makes at most $r$ guesses.
    Then, for all $v> 0$,
    \begin{align*}
        \proba{\sum_{i=1}^m B_i \one_\set{T_i=S_i} \geq v \Big| (X_i)_{i\in[m]}} &\leq   \proba{ \mathrm{Binomial}\left(r , \frac{e^\eps}{1+e^\eps} \right) \geq v}+\alpha m\delta'\,,
    \end{align*}
    where
    \begin{equation*}
    \alpha = \sum_{i\in[m]}\frac{1}{i}  \proba{ v\geq \mathrm{Binomial}\left(r , \frac{e^\eps}{1+e^\gamma} \right) = v-i}\,.
    \end{equation*}
\end{theorem}

\begin{proof}[Proof of \Cref{thm:unified_auditing_better_tampered}]
    	Fix some $i\in[m]$, let $\cD_i=\cD\cup\set{X_j,S_j=1,j\ne i}\cup\set{X_i}$ and $\cD'_i=\cD\cup\set{X_j,X_j=1,j\ne i}$, respectively corresponding to the (random) datasets obtained by setting $S_i$ to 1 and $-1$.
     If $S_i=1$, we have $\dtrain=\cD_i$, while if $S_i=-1$ we have $\dtrain = \cD'_i$.
    Let us define the following (random, since $S_j,j\ne i$ and $S_i$ are random) set:
\begin{equation*}
\cE_{\gamma,i} = \set{\theta\in\Theta:|\PLRV^\cA_{\cD_i,\cD'_i}(\theta)|\leq \gamma}\,.
\end{equation*}	
We first start by proving an upper-bound on:
	\begin{equation*}
		\proba{S_i=1 | \theta\in\cE\cap\cE_{\gamma,i}, S^{(-i)} = s^{(-i)},\bar X }\,,
	\end{equation*} 
 for $\cE$ such that $\theta\in\cE\cap\cE_{\gamma,i}$ is of non-null probability.
	Using Bayes' theorem,
	\begin{align*}
		&\proba{S_i=1 | \theta\in\cE\cap\cE_{\gamma,i}, S^{(-i)} = s^{(-i)},\bar X }\\
  &= \frac{ \proba{ \theta\in\cE\cap\cE_{\gamma,i} | S_i=1 , S^{(-i)} = s^{(-i)},\bar X } \times \proba{S_i=1 | S^{(-i)} = s^{(-i)},\bar X }}{ \proba{ \theta\in\cE\cap\cE_{\gamma,i} | S^{(-i)} = s^{(-i)},\bar X } } \\
		\quad& =  \frac{ \proba{ \theta\in\cE\cap\cE_{\gamma,i} | S_i=1 , S^{(-i)} = s^{(-i)},\bar X } \times \proba{S_i=1 | S^{(-i)} = s^{(-i)},\bar X }}{ \sum_{s\in\set{-1,1}}\proba{ \theta\in\cE\cap\cE_{\gamma,i} | S_i=s, S^{(-i)} = s^{(-i)},\bar X } \proba{ S_i=s | S^{(-i)} = s^{(-i)},\bar X } } \\
		\quad& =  \frac{ \pi(X_i) \proba{ \theta\in\cE\cap\cE_{\gamma,i} | S_i=1 , S^{(-i)} = s^{(-i)},\bar X } }{ \pi(X_i)\proba{ \theta\in\cE\cap\cE_{\gamma,i}| S_i=1, S^{(-i)} = s^{(-i)},\bar X }  + (1-\pi(X_i)\proba{ \theta\in\cE| S_i=-1, S^{(-i)} = s^{(-i)},\bar X }  } \,,
	\end{align*}
	since $\proba{ S_i=1 | S^{(-i)} = s^{(-i)},\bar X } = \pi(X_i)$ and  $\proba{ S_i=-1 | S^{(-i)} = s^{(-i)},\bar X } = 1-\pi(X_i)$.
    Then, let	
	\begin{equation*}
		R_i \eqdef  \frac{\proba{ \theta\in\cE\cap\cE_{\gamma,i} | S_i=1 , S^{(-i)} = s^{(-i)},\bar X } }{ \proba{ \theta\in\cE\cap\cE_{\gamma,i} | S_i=-1 , S^{(-i)} = s^{(-i)},\bar X }}\,.
	\end{equation*}
    We have that:
    \begin{align*}
        \proba{ \theta\in\cE\cap\cE_{\gamma,i} | S_i=1 , S^{(-i)} = s^{(-i)},\bar X } &= \E_{S_{>i}}\int \one_\set{\theta\in\cE}\one_\set{|\PLRV_{\cD_i,\cD_i'}(\theta)|\leq \gamma} \dd\P_{\cA(\cD_i)}(\theta)\\
        &\leq e^\gamma\E_{S_{>i}}\int \one_\set{\theta\in\cE}\one_\set{|\PLRV_{\cD_i,\cD_i'}(\theta)|\leq \gamma} \dd\P_{\cA(\cD_i')}(\theta)\\
        & = e^\gamma \proba{ \theta\in\cE\cap\cE_{\gamma,i} | S_i=-1 , S^{(-i)} = s^{(-i)},\bar X }\,,
    \end{align*}
    since $\PLRV_{\cD_i,\cD_i'}(\theta)\leq \gamma$ implies that $\dd\P_{\cA(\cD_i)}(\theta)\leq e^\gamma \dd\P_{\cA(\cD_i')}(\theta)$.
    Thus, we have that $R_i\leq e^\gamma$, and:
	\begin{align*}
		\proba{S_i=1 | \theta\in\cE\cap\cE_{\gamma,i}, S^{(-i)} = s^{(-i)},\bar X } & \leq \frac{ \pi(X_i) }{ \pi(X_i)  + (1-\pi(X_i))R_i^{-1} } \\
        & \leq \frac{ \pi(X_i) }{ \pi(X_i)  + (1-\pi(X_i))e^{-\gamma} } \\
		&\leq \frac{1}{ 1  + \frac{1-\pi(X_i)}{\pi(X_i)}e^{-\gamma} }\\
		&\leq \frac{1}{1+e^{-\gamma-\eps_i}}\,,
	\end{align*}
    where $e^{\eps_i}=\max(\frac{\pi(X_i)}{1-\pi(X_i)},\frac{1-\pi(X_i)}{\pi(X_i)})$.
	Similarly, since $|\PLRV_{\cD_i,\cD_i'}(\theta)|=|\PLRV_{\cD_i',\cD_i}(\theta)|$ we have that:
	\begin{align*}
		\proba{S_i=-1 | \theta\in\cE\cap\cE_{\gamma,i}, S^{(-i)} = s^{(-i)},\bar X } & \leq \frac{1}{1+e^{-\gamma-\eps_i}}\,.
	\end{align*}
 As a consequence, for all $i$ and event $\theta\in\cE$:
 \begin{align*}
   		&\proba{S_i=T_i | \theta\in\cE\cap\cE_{\gamma,i}, S^{(-i)} = s^{(-i)},\bar X }\\
     & = \proba{S_i=1 | \theta\in\cE\cap\cE_{\gamma,i}, T_i=1, S^{(-i)} = s^{(-i)},\bar X }\proba{T_i=1 | \theta\in\cE\cap\cE_{\gamma,i}, S^{(-i)} = s^{(-i)},\bar X }\\
     &\quad+ \proba{S_i=-1 | \theta\in\cE\cap\cE_{\gamma,i}, T_i=-1, S^{(-i)} = s^{(-i)},\bar X }\proba{T_i=-1 | \theta\in\cE\cap\cE_{\gamma,i}, S^{(-i)} = s^{(-i)},\bar X }\\
     &\leq \frac{1}{1+e^{-\gamma-\eps_i}}\,.  
 \end{align*}
Thus,
\begin{align*}
    &\esp{B_i\one_\set{S_i=T_i} | \theta\in\cE\cap\cE_{\gamma,i}, S^{(-i)} = s^{(-i)},\bar X }\\
    &=\proba{S_i=T_i ,B_i=1| \theta\in\cE\cap\cE_{\gamma,i}, S^{(-i)} = s^{(-i)},\bar X }\\
    &=\proba{B_i=1}\proba{S_i=T_i | \theta\in\cE\cap\cE_{\gamma,i}, S^{(-i)} = s^{(-i)},\bar X }\\
    &=\frac{1+e^{-\gamma-\eps_i}}{1+e^{-\gamma}}\times\frac{1}{1+e^{-\gamma-\eps_i}}\\
    &=\frac{1}{1+e^{-\gamma}}\,.
\end{align*}
Now, for any $\gamma'>\gamma$, let:
\begin{equation*}
\dd\cE_{\gamma,i} = \set{\theta\in\Theta:|\PLRV_{\cD_i,\cD'_i}(\theta)|= \gamma'}\,.
\end{equation*}	
Although this event can be of null probability if the PLRV is continuous, we can still condition on $\dd\cE_{\gamma,i}$, as it amounts to conditioning on a real valued random variable.
Using the same reasoning as above,
	\begin{align*}
		\proba{S_i=T_i| \theta\in\cE\cap\dd\cE_{\gamma',i}, S^{(-i)} = s^{(-i)},\bar X } & \leq \frac{1}{1+e^{-\gamma'-\eps_i}}\,,
	\end{align*}
and
\begin{align*}
		&\esp{B_i\one_\set{S_i=T_i}| \theta\in\cE\cap\dd\cE_{\gamma',i}, S^{(-i)} = s^{(-i)},\bar X } \\
        & \leq \frac{1}{1+e^{-\gamma'-\eps_i}}\times \frac{1+e^{-\gamma-\eps_i}}{1+e^{-\gamma}}\\
  & \leq \frac{1}{1+e^{-\gamma'}}\times \frac{1+e^{-\gamma'}}{1+e^{-\gamma'-\eps_i}}\times \frac{1+e^{-\gamma-\eps_i}}{1+e^{-\gamma}}\\
  & \leq \frac{1}{1+e^{-\gamma'}}\,,
\end{align*}
since $\frac{1+e^{-\gamma'}}{1+e^{-\gamma'-\eps_i}}\leq \frac{1+e^{-\gamma}}{1+e^{-\gamma-\eps_i}}$ for $\gamma'\geq \gamma$.

Let $\xi_j(\gamma')$ be independent Bernoulli random variables of parameters $\frac{1}{1+e^{-\gamma'}}$, independent of all the rest and coupled such that $\xi_j(\gamma')\leq\xi_j(\gamma'')$ for all $\gamma''\geq \gamma'$. Such a coupling can be obtained by sampling $(U_i)_{i\in[m]}$ i.i.d. random variables uniformly in $[0,1]$ and setting:
\begin{equation*}
    \xi_j = \one_\set{U_j\leq \frac{1}{1+e^{-\gamma'}}}\,.
\end{equation*}
Let $\xi_j(\gamma)$ be independent Bernoulli random variables of parameters $\frac{1}{1+e^{-\gamma}}$, independent of all the rest and of the $\xi_j(\gamma')$.
Let then:
\begin{equation*}
    \xi_j' = \one_\set{\cA(\cD)\in\cE_{\gamma,i}} \xi_j(\gamma) + \int_{\gamma'>\gamma} \one_\set{\cA(\cD) \in\dd\cE_{\gamma',i}}\xi_j(\gamma')\,,
\end{equation*}
which is a Bernoulli random variable, as at most only one of the $\zeta_j$ is ‘‘activated''.
	Let us prove by induction that, conditionally on $(X_j,T_j)_{j\in[m]}$, for all $i\in[m]$, $W_i \eqdef \sum_{j\leq i} B_j\one_\set{T_j=S_j}$ is stochastically dominated by $\tilde W_i \eqdef \sum_{j\leq i} |T_j|\xi_j'$.
	For $i=0$ this is of course true since both quantities are null..
	Then, if we assume that the result holds at some $ i\leq n$, we have that $W_{i+1}=W_i + B_{i+1} \one_\set{S_{i+1} = T_{i+1}}$ and $\tilde W_{i+1}=\tilde W_i +  \xi_{i+1}'$.
    Conditionally on the event $\set{W_i=w_i,\tilde W_i=\tilde w_i,\bar T=\bar t,\bar X=\bar x}$ for $w_i,\tilde w_i,\bar t,\bar x$ such that this event is not of probability 0, if $t_{i+1}\ne 0$, using the above computations and writing $\theta=\cA(\cD)$ as the output of the training algorithm:
    \begin{align*}
        &\esp{B_{i+1}\one_\set{S_{i+1} = T_{i+1}} |W_i=w_i,\tilde W_i=\tilde w_i,\bar T=\bar t,\bar X=\bar x}\\
        &=  \esp{B_{i+1}\one_\set{S_{i+1} = T_{i+1}} ,\theta\in\cE_{\gamma,i+1}|W_i=w_i,\tilde W_i=\tilde w_i,\bar T=\bar t,\bar X=\bar x}\\
        &\quad +  \int_{\gamma'>\gamma} \esp{B_{i+1}\one_\set{S_{i+1} = T_{i+1}} ,\theta\in\dd\cE_{\gamma',i+1}|W_i=w_i,\tilde W_i=\tilde w_i,\bar T=\bar t,\bar X=\bar x}\dd\gamma'\\
        &\leq  \esp{B_{i+1}\one_\set{S_{i+1} = T_{i+1}} |\theta\in\cE_{\gamma,i+1},W_i=w_i,\tilde W_i=\tilde w_i,\bar T=\bar t,\bar X=\bar x}\\
        &\qquad\times\proba{\theta\in\cE_{\gamma,i+1}|W_i=w_i,\tilde W_i=\tilde w_i,\bar T=\bar t,\bar X=\bar x}\\
        &\quad + \int_{\gamma'>\gamma}\esp{B_{i+1}\one_\set{S_{i+1} = T_{i+1}} |\theta\in\dd\cE_{\gamma',i+1},W_i=w_i,\tilde W_i=\tilde w_i,\bar T=\bar t,\bar X=\bar x}\\
        &\qquad\times\proba{\theta\in\dd\cE_{\gamma',i+1}|W_i=w_i,\tilde W_i=\tilde w_i,\bar T=\bar t,\bar X=\bar x}\dd\gamma'\\
        &\leq  \frac{e^{\gamma}}{1+e^{\gamma}}\times\proba{\theta\in\cE_{\gamma,i+1}|W_i=w_i,\tilde W_i=\tilde w_i,\bar T=\bar t,\bar X=\bar x}\\
        &\quad + \int_{\gamma'>\gamma} \frac{e^{\gamma'}}{1+e^{\gamma'}}\times\proba{\theta\in\dd\cE_{\gamma',i+1}|W_i=w_i,\tilde W_i=\tilde w_i,\bar T=\bar t,\bar X=\bar x}\dd\gamma'\\
        &= \esp{\xi'_{i+1}|W_i=w_i,\tilde W_i=\tilde w_i,\bar T=\bar t,\bar X=\bar x}\,,
    \end{align*}
    using the above computations with the adequate event $\cE$.
    Thus,
    \begin{align*}
        &\esp{B_{i+1}\one_\set{S_{i+1}=T_{i+1}} |W_i=w_i,\tilde W_i=\tilde w_i,\bar T=\bar t,\bar X=\bar x}\\
        &\leq \esp{|T_{i+1}|\xi_{i+1}'|W_i=w_i,\tilde W_i=\tilde w_i,\bar T=\bar t,\bar X=\bar x}\,,
    \end{align*}
    leading to $\max(0,T_{i+1}S_{i+1})$ stochastically dominated by $|T_{i+1}|\xi_{i+1}'$ since stochastic domination for Bernoulli random variables is equivalent to domination of means, conditionally on $W_i$, $\tilde W_i$, $X$ and $T_{\leq i+1}$.
    Using \citet[Lemma 4.9]{steinke2023privacy}: $W_{i+1}$ is stochastically dominated by $\tilde W_{i+1}$, conditionally on $W_i$, $\tilde W_i$ and $T_{\leq i+1}$.
    Thus, $W_n$ is stochastichally dominated by:
    \begin{equation*}
        W_n'=\sum_{i=1}^m|T_i|\xi_i'\,.
    \end{equation*}
    Recall that:
    \begin{equation*}
    \xi_j' = \one_\set{\cA(\cD)\in\cE_{\gamma,i}} \xi_j(\gamma) + \int_{\gamma'>\gamma} \one_\set{\cA(\cD) \in\dd\cE_{\gamma',i}}\xi_j(\gamma')\,,
\end{equation*}
    so that:
    \begin{equation*}
        W_m' = A_m + B_m\,,
    \end{equation*}
    where
    \begin{equation*}
        A_m = \sum_{i=1}^m|T_i|\xi_i(\gamma)\,,
        \end{equation*}    
and
    \begin{align*}
        B_m &\eqdef  \sum_{j=1}^m|T_j|\left(\int_{\gamma'>\gamma} \one_\set{\cA(\cD) \in\dd\cE_{\gamma',i}}\xi_j(\gamma') - \one_\set{\cA(\cD)\notin\cE_{\gamma,i}} \xi_j(\gamma)\right)\\
        & = \sum_{j=1}^m|T_j|\left(\int_{\gamma'>\gamma} \one_\set{\cA(\cD) \in\dd\cE_{\gamma',i}}\big(\xi_j(\gamma') -  \xi_j(\gamma)\big)\right)\\
    \end{align*}
    and we are left with upper bounding this second sum. Note that under our coupling of $\xi_j(\gamma),\xi_j(\gamma')$, each element of the sum in $B_m$ is non-negative, enabling us to use Markov inequality later in the proof.
Let
\begin{align*}
            B_m' = \sum_{j=1}^m\left(\int_{\gamma'>\gamma} \one_\set{\cA(\cD) \in\dd\cE_{\gamma',i}}\big(\xi_j(\gamma') -  \xi_j(\gamma)\big)\right)\,,
\end{align*}
so that $B_m'\geq B_m$ almost surely, leading to 
\begin{align*}
    \proba{W_m' \geq v}&=\proba{A_m+B_m\geq v}\\
    &\leq \proba{A_m+B_m'\geq v}\\
    &\leq \proba{A_m\geq v}+\proba{v>A_m\geq v-B_m'}\,.
\end{align*}
Our problem now is that $A_m,B_m'$ are apriori not independent. However, due to our coupling of the random variables $(\xi_i(\gamma'))_{\gamma'\geq\gamma}$, we can easily characterize the distribution of $B_m'$ conditionally on the random variables $\xi_i(\gamma)$.
Indeed, for $U_i$ uniform in $[0,1]$, we have $\xi_i(\gamma)=\one_\set{U_i\leq \frac{1}{1+e^{-\gamma}}}$ and $\xi_i(\gamma')-\xi_i(\gamma)=\one_\set{\frac{1}{1+e^{-\gamma'}}\geq U_i> \frac{1}{1+e^{-\gamma}}}$. 
We are thus going to upper bound $\esp{B_m'|A_m}$.
First,
\begin{align*}
    \esp{\int_{\gamma'>\gamma} \one_\set{\cA(\cD) \in\dd\cE_{\gamma',j}}\xi_j(\gamma') - \one_\set{\cA(\cD)\notin\cE_{\gamma,j}} \xi_j(\gamma) \,\Big|\, \xi_j(\gamma) = 1 } =0\,,
\end{align*}
since $\xi_j(\gamma) = 1 \implies \xi_j(\gamma') = 1$. Then, under $\xi_j(\gamma)=0$, we have that $\xi_j(\gamma')-\xi_j(\gamma)$ is a Bernoulli random variable of parameter:
\begin{align*}
    \frac{\frac{1}{1+e^{-\gamma'}}-\frac{1}{1+e^{-\gamma}}}{1- \frac{1}{1+e^{-\gamma}}} &= \frac{\frac{1+e^{-\gamma}}{1+e^{-\gamma'}}}{e^{-\gamma}}\\
    & = e^\gamma \frac{e^{-\gamma}-e^{-\gamma'}}{1+e^{-\gamma'}}\\
    & \leq e^\gamma (e^{-\gamma}-e^{-\gamma'})\\
    & = 1-e^{\gamma-\gamma'}\,,
\end{align*}
leading to:
\begin{align*}
        &\esp{\int_{\gamma'>\gamma} \one_\set{\cA(\cD) \in\dd\cE_{\gamma',j}}\xi_j(\gamma') - \one_\set{\cA(\cD)\notin\cE_{\gamma,j}} \xi_j(\gamma)\,\Big|\, \xi_j(\gamma) = 0 }\\
        &= \esp{ 1-e^{\gamma-|\PLRV^\cA_{\cD_j,\cD_j'}(\cA(\cD))|} \Big| |\PLRV^\cA_{\cD_j,\cD_j'}(\cA(\cD))| >\gamma }\proba{\cA(\cD)\notin\cE_{\gamma,j}}\\
        &= \esp{ \left(1-e^{\gamma-|\PLRV^\cA_{\cD_j,\cD_j'}(\cA(\cD))|}\right)_+ }\qquad (\text{where}\qquad \forall z\in\R\,,\quad z_+ = \max(0,z))\,.
    \end{align*}
    Now, under $S_j=1$ we have $\cD=\cD_j$, while under $S_j=-1$ we have $\cD=\cD_j'$. Either way, we have, using our assumption on the mechanism $\cA$:
    \begin{align*}
                &\esp{\int_{\gamma'>\gamma} \one_\set{\cA(\cD) \in\dd\cE_{\gamma',j}}\xi_j(\gamma') - \one_\set{\cA(\cD)\notin\cE_{\gamma,j}} \xi_j(\gamma)\,\Big|\, \xi_j(\gamma) = 0 }\\
                &\leq \max\left(\esp{ \left(1-e^{\gamma-|\PLRV^\cA_{\cD_j,\cD_j'}(\cA(\cD_j))|}\right)_+ },\esp{ \left(1-e^{\gamma-|\PLRV^\cA_{\cD_j,\cD_j'}(\cA(\cD_j'))|}\right)_+ } \right)\\
                &\leq \delta'\,.
    \end{align*}
    Thus, almost surely:
    \begin{align*}
        \esp{\int_{\gamma'>\gamma} \one_\set{\cA(\cD) \in\dd\cE_{\gamma',j}}\xi_j(\gamma') - \one_\set{\cA(\cD)\notin\cE_{\gamma,j}} \xi_j(\gamma)\,\Big|\, \xi_j(\gamma) } \leq \delta'\,,
    \end{align*}
    and:
    \begin{align*}
        \esp{B_m'|A_m}\leq m\delta'\,.
    \end{align*}
    Using conditional Markov inequality, for any $t>0$, $\proba{B_m'\geq t|A_m}\leq \frac{m\delta'}{t}$. Using this,
    \begin{align*}
        \proba{v>A_m\geq v-B_m'} & = \esp{\one_\set{A_m<v} \proba{B_m\geq v-A_m|A_m}}\\
        & \leq \esp{\one_\set{A_m<v} \times \frac{m\delta'}{v-A_m}}\\
        & = m\delta'\times \sum_{k<v}\proba{A_m = v-k}k^{-1}\,.
    \end{align*}
    We then conclude using the fact that under our assumption on the attacker, $A_m$ is a Binomial random variable.
\end{proof}

\subsubsection{Proof of \Cref{thm:conditional_audit_tampered}.\ref{thm:conditional_audit_tampered:eps_delta}}

\begin{proof}[Proof of \Cref{thm:conditional_audit_eps_delta}]
    The proof follows by combining the following proposition with \Cref{thm:unified_auditing_better}.
\begin{proposition}\label{prop:hockey}
    If $\cA$ is $(\eps,\delta)$-DP, then for any two adjacent datasets $\cD\sim\cD'$, we have:
    \begin{equation*}
        \esp{\left( 1-e^{\eps - |\PLRV_{\cD,\cD'}(\cA(\cD))|} \right)_+}\leq (1+e^{-\eps})\delta\,.
    \end{equation*}
\end{proposition}
\begin{proof}
    First, note that $(\eps,\delta)$-DP directly implies that:
    \begin{equation*}
        \esp{\left( 1-e^{\eps - \PLRV_{\cD,\cD'}(\cA(\cD))} \right)_+}\leq \delta\,.
    \end{equation*}
    Then,
    \begin{align*}
        &\esp{\left( 1-e^{\eps - \PLRV_{\cD,\cD'}(\cA(\cD))} \right)_+}\\
        &\leq \esp{\left( 1-e^{\eps - \PLRV_{\cD,\cD'}(\cA(\cD))} \right)\one_\set{\PLRV_{\cD,\cD'}(\cA(\cD))\geq \eps} }\\
        &\quad + \esp{\left( 1-e^{\eps + \PLRV_{\cD,\cD'}(\cA(\cD))} \right)\one_\set{\PLRV_{\cD,\cD'}(\cA(\cD))\leq -\eps} }\\
        &\leq \delta  + \esp{\left( 1-e^{\eps + \PLRV_{\cD,\cD'}(\cA(\cD))} \right)\one_\set{\PLRV_{\cD,\cD'}(\cA(\cD))\leq -\eps} }\\
        &= \delta  + \esp{e^{\PLRV_{\cD,\cD'}(\cA(\cD'))}\left( 1-e^{\eps + \PLRV_{\cD,\cD'}(\cA(\cD'))} \right)\one_\set{\PLRV_{\cD,\cD'}(\cA(\cD'))\leq -\eps} } \\
        &\qquad\qquad (\text{change of measure})\\
        &= \delta  + \esp{e^{\PLRV_{\cD,\cD'}(\cA(\cD'))}\left( 1-e^{\eps - \PLRV_{\cD',\cD}(\cA(\cD'))} \right)\one_\set{\PLRV_{\cD,\cD'}(\cA(\cD'))\leq -\eps} } \\
        &\qquad\qquad (\text{since}\quad \PLRV_{\cD',\cD}(\cA(\cD'))=-\PLRV_{\cD,\cD'}(\cA(\cD'))\quad)\\
        &\leq \delta  + \esp{e^{-\eps}\left( 1-e^{\eps - \PLRV_{\cD',\cD}(\cA(\cD'))} \right)\one_\set{\PLRV_{\cD,\cD'}(\cA(\cD'))\leq -\eps} } \\
        &\leq \delta + e^{-\eps}\delta\,,
        \end{align*}
    concluding the proof.
\end{proof}
\end{proof}

\subsection{Conditional Zero-Run $(\eps,\delta)$-DP Auditing Algorithm}
\label{sec:audit_approx_DP_tampered}

\begin{algorithm}[t]
    \caption{Conditional Zero-Run $(\eps,\delta)$-DP auditing}
    \label{alg:audit_approx_DP_tampered}
    \begin{algorithmic}[1]
        \STATE \textbf{Inputs:}
        Instance $\theta=\cA(\dtrain)$ of the algorithm $\cA$ to audit;
        MIA $\cM$;
        $\cD_1\subseteq\dtrain$;
        $\cD_0\subseteq\cX\setminus\dtrain$;
        probability of error $p$;
        privacy parameters $\eps,\delta>0$;
        guess budget $r\in[m]$ such that $\NRM{T}_0\leq r$
        almost surely

        \STATE $\set{X_1,\ldots,X_m}\gets\cD_1\cup\cD_0$
        \STATE
        $S_i\gets
        2\one_\set{X_i\in\cD_1}-1
        \in\set{-1,1}$
        for all $i\in[m]$

        \STATE Run \Cref{alg:mia} on
        $(\theta,\cD_1,\cD_0,\cM)$
        to obtain $(T_i)_{i\in[m]}$

        \STATE Run \Cref{alg:propensity_learning} on
        $(X_i,S_i)_{i\in[m]}$
        to obtain $(\hat\pi_i)_{i\in[m]}$

        \FOR{$i\in[m]$}
            \STATE
            $\displaystyle
            b_i\gets
            \frac{
                \min\{\hat\pi_i,1-\hat\pi_i\}
            }{
                \max\{\hat\pi_i,1-\hat\pi_i\}
            }$

            \STATE Sample
            $B_i\sim\mathrm{Bernoulli}(b_i)$
            independently
        \ENDFOR

        \STATE
        $\displaystyle
        c\gets
        \sum_{i=1}^{m}
        B_i\one_\set{T_i=S_i}$

        \STATE Set $\bar f\gets 1-f_{\eps,\delta}$

        \STATE Run \Cref{alg:f_DP_counts} on
        $(p,\bar f,r,c)$
        to obtain $a=r[0]+h[0]$

        \IF{$a\geq r/m$}
            \STATE \textbf{Return} \textsc{False}
        \ELSE
            \STATE \textbf{Return} \textsc{True}
        \ENDIF
    \end{algorithmic}
\end{algorithm}

\section{Conditional Approximate DP Auditing \textit{Without Tampering}}

In this section, we show that approximate DP can be audited in Zero-Run, without tampering.
In order to do so, instead of multiplying MIA guesses by independent Bernoulli variables in order to tamper with their confidence, we adapt the success rate, which need to take into account MIA overconfidence.
While in the pure DP setting, both results would yield equivalent results, in the approximate DP case the tampering results are stronger, as they yield milder $\delta$ dependence.
However, the obtained MIA success rate bounds here are conditioned also on $T_i$.

\subsection{$\eps$-DP Conditional Auditing without Tampering}
\label{app:proof_pure_DP}

\begin{theorem}[$\eps$-DP Auditing]\label{thm:pure_DP_auditing}
  Let $\eps>0$ and assume that $\cA$  satisfies $\eps$-DP.
	For any $v\in\R$, the success rate of the guesses of the attack is upper bounded as:
	\begin{equation}
        \begin{aligned}
		      &\proba{\sum_{i=1}^m \max(0,T_iS_i) \geq v\,\Big|\, (X_i,T_i)_{i\in[m]}} \\
        &\leq \proba{\sum_{i=1}^m  \xi_i |T_i|
        \geq v\,\Big|\, (X_i,T_i)_{i\in[m]}}\,,
        \end{aligned}
	\end{equation}
 where $\xi_i$ are independent Bernoulli random variables of respective parameters given by:
 \begin{equation*}
     p_i(\eps) = \frac{e^{\eps_i}}{1+e^{\eps_i}}\,,\qquad \text{where}\qquad \eps_i=\eps+\left|\log\left(\frac{\pi(X_i)}{1-\pi(X_i)}\right)\right|\,.
 \end{equation*}
\end{theorem}

\begin{proof}[Proof of \Cref{thm:pure_DP_auditing}]
	The proof is largely inspired by that of \citet[Proposition 5.1]{steinke2023privacy}.
	Let ${\theta\in\cE}$ be an event of non-null probability.
	Fix some $i\in[m]$ and $s^{(-i)}\in\set{-1,1}^m$.
	We first start by proving an upper-bound on:
	\begin{equation*}
		\proba{S_i=1 | \theta\in\cE, S^{(-i)} = s^{(-i)},\bar X }\,.
	\end{equation*} 
	Using Bayes' theorem,
	\begin{align*}
		&\proba{S_i=1 | \theta\in\cE, S^{(-i)} = s^{(-i)},\bar X } = \frac{ \proba{ \theta\in\cE | S_i=1 , S^{(-i)} = s^{(-i)},\bar X } \times \proba{S_i=1 | S^{(-i)} = s^{(-i)},\bar X }}{ \proba{ \theta\in\cE | S^{(-i)} = s^{(-i)},\bar X } } \\
		\quad& =  \frac{ \proba{ \theta\in\cE | S_i=1 , S^{(-i)} = s^{(-i)},\bar X } \times \proba{S_i=1 | S^{(-i)} = s^{(-i)},\bar X }}{ \sum_{s\in\set{-1,1}}\proba{ \theta\in\cE | S_i=s, S^{(-i)} = s^{(-i)},\bar X } \proba{ S_i=s | S^{(-i)} = s^{(-i)},\bar X } } \\
		\quad& =  \frac{ \pi(X_i) \proba{ \theta\in\cE | S_i=1 , S^{(-i)} = s^{(-i)},\bar X } }{ \pi(X_i)\proba{ \theta\in\cE| S_i=1, S^{(-i)} = s^{(-i)},\bar X }  + (1-\pi(X_i)\proba{ \theta\in\cE| S_i=-1, S^{(-i)} = s^{(-i)},\bar X }  } \,,
	\end{align*}
	since $\proba{ S_i=1 | S^{(-i)} = s^{(-i)},\bar X } = \pi(X_i)$   and  $\proba{ S_i=-1 | S^{(-i)} = s^{(-i)},\bar X } = 1-\pi(X_i)$.
	Then, using the fact that due to $\eps$-DP we have that:
	\begin{equation*}
		R \eqdef  \frac{\proba{ \theta\in\cE | S_i=1 , S^{(-i)} = s^{(-i)},\bar X } }{ \proba{ \theta\in\cE | S_i=-1 , S^{(-i)} = s^{(-i)},\bar X }} \in \left[ \frac{1}{e^\eps} , e^\eps\right]\,,
	\end{equation*}
	we have:
	\begin{align*}
		\proba{S_i=1 | \theta\in\cE, S^{(-i)} = s^{(-i)},\bar X } & \leq \frac{ \pi(X_i) }{ \pi(X_i)  + (1-\pi(X_i))R^{-1} } \\
        & \leq \frac{ \pi(X_i) }{ \pi(X_i)  + (1-\pi(X_i))e^{-\eps} } \\
		&\leq \frac{1}{ 1  + \frac{1-\pi(X_i)}{\pi(X_i)}e^{-\eps} }\\
		&\leq \frac{1}{1+e^{-\eps-\eps_i}}\,,
	\end{align*}
    where $e^{\eps_i}=\max(\frac{\pi(X_i)}{1-\pi(X_i)},\frac{1-\pi(X_i)}{\pi(X_i)})$.
	Similarly, we have that:
	\begin{align*}
		\proba{S_i=-1 | \theta\in\cE, S^{(-i)} = s^{(-i)},\bar X } & \leq \frac{1}{1+e^{-\eps-\eps_i}}\,.
	\end{align*}
	We can thus prove by induction that, conditionally on $(X_j,T_j)_{j\in[m]}$, for all $i\in[m]$, $W_i = \sum_{j\leq i} \max(0,T_jS_j)$ is stochastically dominated by $\tilde W_j = \sum_{j\leq i} |T_j|\xi_j$, where $\xi_j$ are independent Bernoulli random variables of parameters $\frac{1}{1+e^{-\eps-\eps_j}}$.
	For $i=1$ this is of course true.
	Then, if we assume that the result holds at some $i\leq n$, we have that $W_{i+1}=W_i +  \one_\set{S_{i+1} = \mathrm{sign}(T_{i+1})}$ and $\tilde W_{i+1}=\tilde W_i +  \xi_{i+1}$.
    Conditionally on the event $\set{W_i=w_i,\tilde W_i=\tilde w_i,\bar T=\bar t}$ for $w_i,\tilde w_i,\bar t$ such that this event is not of probability 0:
    \begin{align*}
        &\esp{\one_\set{S_{i+1} = \mathrm{sign}(T_{i+1})} |W_i=w_i,\tilde W_i=\tilde w_i,\bar T=\bar t}\\
        &= \proba{S_{i+1} = \mathrm{sign}(T_{i+1}) |W_i=w_i,\tilde W_i=\tilde w_i,\bar T=\bar t}\\
        &=  \proba{S_{i+1} = T_{i+1}|W_i=w_i,\tilde W_i=\tilde w_i,\bar T=\bar t}\\
        &=  \proba{S_{i+1} = 1|W_i=w_i,\tilde W_i=\tilde w_i,\bar T=\bar t,T_{i+1}=1}\proba{T_{i+1}=1|W_i=w_i,\tilde W_i=\tilde w_i,\bar T=\bar t}\\
        &\quad+ \proba{S_{i+1} = -1|W_i=w_i,\tilde W_i=\tilde w_i,\bar T=\bar t,T_{i+1}=-1}\proba{T_{i+1}=-1|W_i=w_i,\tilde W_i=\tilde w_i,\bar T=\bar t}\\
        &\leq \frac{1}{1+e^{-\eps-\eps_i}}\proba{T_{i+1}=1|W_i=w_i,\tilde W_i=\tilde w_i,\bar T=\bar t}\\
        &\quad+ \frac{1}{1+e^{-\eps-\eps_i}}\proba{T_{i+1}=-1|W_i=w_i,\tilde W_i=\tilde w_i,\bar T=\bar t}\,,
    \end{align*}
    using the above computations with the adequate event $\cE$, which leads to:
    \begin{align*}
        \esp{\one_\set{S_{i+1} = \mathrm{sign}(T_{i+1})} |W_i=w_i,\tilde W_i=\tilde w_i,\bar T=\bar t}&\leq \esp{\zeta_{i+1}|W_i=w_i,\tilde W_i=\tilde w_i,\bar T=\bar t}\,.
    \end{align*}
    Thus,
    \begin{align*}
        \esp{\max(0,T_{i+1}S_{i+1}) |W_i=w_i,\tilde W_i=\tilde w_i,\bar T=\bar t}&\leq \esp{|T_{i+1}|\zeta_{i+1}|W_i=w_i,\tilde W_i=\tilde w_i,\bar T=\bar t}\,.
    \end{align*}
    leading to $\max(0,T_{i+1}S_{i+1})$ stochastically dominated by $|T_{i+1}|\zeta_{i+1}$ conditionally on $W_i$, $\tilde W_i$ and $T_{\leq i+1}$. We conclude using \citet[Lemma 4.9]{steinke2023privacy}: $W_{i+1}$ is stochastically dominated by $\tilde W_{i+1}$, conditionally on $W_i$, $\tilde W_i$ and $T_{\leq i+1}$.
\end{proof}

\subsection{Conditional $(\eps,\delta)$-DP Auditing without Tampering}
\label{app:proof_eps_delta_cond}

\begin{proposition}
    For any  $S\sim_i S'$ two bits datasets adjacent on $i\in[m]$ and any $x\in\cX^m$, we have:
    \begin{equation*}
        \left|\log\left( \frac{\proba{\cM_\DS(S)=x}}{\proba{\cM_\DS(S')=x}}\right)\right| = \eps_\DS(x_i)\,.
    \end{equation*}
\end{proposition}
Given the prior that members and non-members respectively follow distributions $\P_1$ and $\P_0$, the virtual conditional mechanism $\cM_{\DS|x_i}^{(i)}$ for data point $i\in[m]$ is an exact $\eps_\DS(x_i)$-DP mechanism.
By the chain rule of the Privacy Loss Random Variable, the total instance-specific leakage $\cM_{\Tot|x_i}^{(i)}$ for data point $i$, having observed $x_i$, is bounded by the sequential composition of the two conditionally independent mechanisms $\cM_{\DS|x_i}^{(i)}$ and $\cA$.
Thus, if $\cA$ is $\eps$-DP, $\cM_{\Tot|x_i}^{(i)}$ is $(\eps+\eps_\DS(x_i))$-DP, and the probability of making a good guess (i.e., of the event $\set{T_i=S_i}$) on the observed data point $X_i$ is upper-bounded by:
\begin{equation}\label{eq:p_i_eps}
    p_i(\eps)=\frac{e^{\eps+\eps_\DS(X_i)}}{1+e^{\eps+\eps_\DS(X_i)}}\,.
\end{equation}
Next result uses and formalizes this heuristics to upper bound the success rate of any MIA, capturing instance-specific privacy leakage. 

\begin{theorem}[$(\eps,\delta)$-Conditional Auditing]
\label{thm:conditional_audit_eps_delta}
    Assume that \Cref{hyp:data_generation} holds.
    Let $(\zeta_i)_{i\in[m]}$ be independent Bernoulli random variables of parameters $(p_i(\eps))_{i\in[m]}$, as defined in \Cref{eq:p_i_eps}.
    Then, for any $u,v>0$, with probability $1-\sqrt{\frac{m(1+e^{-\eps})\delta}{v}}$ over the randomness of $(T_i)_{i\in[m]}$, we have:
    \begin{align*}
        \proba{\sum_{i=1}^m \one_\set{T_i=S_i} \geq v \Big| (X_i,T_i)_{i\in[m]}} &\leq   \proba{\sum_{i=1}^m  |T_i|\xi_i
        \geq v-u\,\Big|\, (X_i,T_i)_{i\in[m]}}+\sqrt{\frac{m(1+e^{-\eps})\delta}{v}}\\
        & \eqdef \Phi_{\eps,\delta,u}\big(v\,\big|\,(\pi(X_i),T_i)_{i\in[m]}\big)
        \,.
    \end{align*}
\end{theorem}

\subsection{Intermediate Result}
We first start with the following result, from which \Cref{thm:conditional_audit_eps_delta} will directly follow, by plugging in $\gamma=\eps$ and $\delta'=(1+e^{-1})\delta$, as we then prove in \Cref{prop:hockey}.

\begin{theorem}[Intermediate Result]
\label{thm:unified_auditing_better}
    Let $\gamma>0$ and $\delta'>0$ such that for all adjacent datasets $\cD\sim\cD'$,
    \begin{equation*}
        \esp{(1-e^{\gamma-|\PLRV_{\cD,\cD'}^\cA(\cA(\cD))|})_+} \leq \delta'\,.
    \end{equation*}
    Assume that \Cref{hyp:data_generation} holds.
    Let $(\zeta_i)_{i\in[m]}$ be independent Bernoulli random variables of parameters $(p_i(\gamma))_{i\in[m]}$, defined as:
    \begin{equation*}
        p_i(\gamma)=\frac{e^{\gamma+\eps_i}}{1+e^{\gamma+\eps_i}}\,,
    \end{equation*}
    where $\eps_i=|\log(\pi(X_i)/(1-\pi(X_i)))$.
    Then, for all $u,v> 0$, with probability $1-\sqrt{m\delta'/u}$ over the randomness of $(T_i)_{i\in[m]}$:
    \begin{align*}
        \proba{\sum_{i=1}^m \max(0,T_iS_i) \geq v \Big| (X_i,T_i)_{i\in[m]}} &\leq   \proba{\sum_{i=1}^m  |T_i|\xi_i
        \geq v-u\,\Big|\, (X_i,T_i)_{i\in[m]}} \\
        &\qquad+ \sqrt{m\delta'/u}\,.
    \end{align*}
\end{theorem}

\begin{proof}[Proof of \Cref{thm:unified_auditing_better}]
    	Fix some $i\in[m]$, let $\cD_i=\cD\cup\set{X_j,S_j=1,j\ne i}\cup\set{X_i}$ and $\cD'_i=\cD\cup\set{X_j,X_j=1,j\ne i}$, respectively corresponding to the (random) datasets obtained by setting $S_i$ to 1 and $-1$.
     If $S_i=1$, we have $\dtrain=\cD_i$, while if $S_i=-1$ we have $\dtrain = \cD'_i$.
    Let us define the following (random, since $S_j,j\ne i$ and $S_i$ are random) set:
\begin{equation*}
\cE_{\gamma,i} = \set{\theta\in\Theta:|\PLRV^\cA_{\cD_i,\cD'_i}(\theta)|\leq \gamma}\,.
\end{equation*}	
We first start by proving an upper-bound on:
	\begin{equation*}
		\proba{S_i=1 | \theta\in\cE\cap\cE_{\gamma,i}, S^{(-i)} = s^{(-i)},\bar X }\,,
	\end{equation*} 
 for $\cE$ such that $\theta\in\cE\cap\cE_{\gamma,i}$ is of non-null probability.
	Using Bayes' theorem,
	\begin{align*}
		&\proba{S_i=1 | \theta\in\cE\cap\cE_{\gamma,i}, S^{(-i)} = s^{(-i)},\bar X }\\
  &= \frac{ \proba{ \theta\in\cE\cap\cE_{\gamma,i} | S_i=1 , S^{(-i)} = s^{(-i)},\bar X } \times \proba{S_i=1 | S^{(-i)} = s^{(-i)},\bar X }}{ \proba{ \theta\in\cE\cap\cE_{\gamma,i} | S^{(-i)} = s^{(-i)},\bar X } } \\
		\quad& =  \frac{ \proba{ \theta\in\cE\cap\cE_{\gamma,i} | S_i=1 , S^{(-i)} = s^{(-i)},\bar X } \times \proba{S_i=1 | S^{(-i)} = s^{(-i)},\bar X }}{ \sum_{s\in\set{-1,1}}\proba{ \theta\in\cE\cap\cE_{\gamma,i} | S_i=s, S^{(-i)} = s^{(-i)},\bar X } \proba{ S_i=s | S^{(-i)} = s^{(-i)},\bar X } } \\
		\quad& =  \frac{ \pi(X_i) \proba{ \theta\in\cE\cap\cE_{\gamma,i} | S_i=1 , S^{(-i)} = s^{(-i)},\bar X } }{ \pi(X_i)\proba{ \theta\in\cE\cap\cE_{\gamma,i}| S_i=1, S^{(-i)} = s^{(-i)},\bar X }  + (1-\pi(X_i)\proba{ \theta\in\cE| S_i=-1, S^{(-i)} = s^{(-i)},\bar X }  } \,,
	\end{align*}
	since $\proba{ S_i=1 | S^{(-i)} = s^{(-i)},\bar X } = \pi(X_i)$ and  $\proba{ S_i=-1 | S^{(-i)} = s^{(-i)},\bar X } = 1-\pi(X_i)$.
    Then, let	
	\begin{equation*}
		R_i \eqdef  \frac{\proba{ \theta\in\cE\cap\cE_{\gamma,i} | S_i=1 , S^{(-i)} = s^{(-i)},\bar X } }{ \proba{ \theta\in\cE\cap\cE_{\gamma,i} | S_i=-1 , S^{(-i)} = s^{(-i)},\bar X }}\,.
	\end{equation*}
    We have that:
    \begin{align*}
        \proba{ \theta\in\cE\cap\cE_{\gamma,i} | S_i=1 , S^{(-i)} = s^{(-i)},\bar X } &= \E_{S_{>i}}\int \one_\set{\theta\in\cE}\one_\set{|\PLRV_{\cD_i,\cD_i'}(\theta)|\leq \gamma} \dd\P_{\cA(\cD_i)}(\theta)\\
        &\leq e^\gamma\E_{S_{>i}}\int \one_\set{\theta\in\cE}\one_\set{|\PLRV_{\cD_i,\cD_i'}(\theta)|\leq \gamma} \dd\P_{\cA(\cD_i')}(\theta)\\
        & = e^\gamma \proba{ \theta\in\cE\cap\cE_{\gamma,i} | S_i=-1 , S^{(-i)} = s^{(-i)},\bar X }\,,
    \end{align*}
    since $\PLRV_{\cD_i,\cD_i'}(\theta)\leq \gamma$ implies that $\dd\P_{\cA(\cD_i)}(\theta)\leq e^\gamma \dd\P_{\cA(\cD_i')}(\theta)$.
    Thus, we have that $R_i\leq e^\gamma$, and:
	\begin{align*}
		\proba{S_i=1 | \theta\in\cE\cap\cE_{\gamma,i}, S^{(-i)} = s^{(-i)},\bar X } & \leq \frac{ \pi(X_i) }{ \pi(X_i)  + (1-\pi(X_i))R_i^{-1} } \\
        & \leq \frac{ \pi(X_i) }{ \pi(X_i)  + (1-\pi(X_i))e^{-\gamma} } \\
		&\leq \frac{1}{ 1  + \frac{1-\pi(X_i)}{\pi(X_i)}e^{-\gamma} }\\
		&\leq \frac{1}{1+e^{-\gamma-\eps_i}}\,,
	\end{align*}
    where $e^{\eps_i}=\max(\frac{\pi(X_i)}{1-\pi(X_i)},\frac{1-\pi(X_i)}{\pi(X_i)})$.
	Similarly, since $|\PLRV_{\cD_i,\cD_i'}(\theta)|=|\PLRV_{\cD_i',\cD_i}(\theta)|$ we have that:
	\begin{align*}
		\proba{S_i=-1 | \theta\in\cE\cap\cE_{\gamma,i}, S^{(-i)} = s^{(-i)},\bar X } & \leq \frac{1}{1+e^{-\gamma-\eps_i}}\,.
	\end{align*}
 As a consequence, for all $i$ and event $\theta\in\cE$:
 \begin{align*}
   		&\proba{S_i=T_i | \theta\in\cE\cap\cE_{\gamma,i}, S^{(-i)} = s^{(-i)},\bar X }\\
     & = \proba{S_i=1 | \theta\in\cE\cap\cE_{\gamma,i}, T_i=1, S^{(-i)} = s^{(-i)},\bar X }\proba{T_i=1 | \theta\in\cE\cap\cE_{\gamma,i}, S^{(-i)} = s^{(-i)},\bar X }\\
     &\quad+ \proba{S_i=-1 | \theta\in\cE\cap\cE_{\gamma,i}, T_i=-1, S^{(-i)} = s^{(-i)},\bar X }\proba{T_i=-1 | \theta\in\cE\cap\cE_{\gamma,i}, S^{(-i)} = s^{(-i)},\bar X }\\
     &\leq \frac{1}{1+e^{-\gamma-\eps_i}}\,.  
 \end{align*}
Now, for any $\gamma'>\gamma$, let:
\begin{equation*}
\dd\cE_{\gamma,i} = \set{\theta\in\Theta:|\PLRV_{\cD_i,\cD'_i}(\theta)|= \gamma'}\,.
\end{equation*}	
Although this event can be of null probability if the PLRV is continuous, we can still condition on $\dd\cE_{\gamma,i}$, as it amounts to conditioning on a real valued random variable.
Using the same reasoning as above,
	\begin{align*}
		\proba{S_i=T_i| \theta\in\cE\cap\dd\cE_{\gamma',i}, S^{(-i)} = s^{(-i)},\bar X } & \leq \frac{1}{1+e^{-\gamma'-\eps_i}}\,.
	\end{align*}
Let $\xi_j(\gamma')$ be independent Bernoulli random variables of parameters $\frac{1}{1+e^{-\gamma'-\eps_j}}$, independent of all the rest and coupled such that $\xi_j(\gamma)\leq\xi_j(\gamma')$ for all $\gamma'\geq \gamma$. Such a coupling can be obtained by sampling $(U_i)_{i\in[m]}$ i.i.d. random variables uniformly in $[0,1]$ and setting, for any $\gamma'\geq \gamma$:
\begin{equation*}
    \xi_j = \one_\set{U_j\leq p_j(\gamma')}\,.
\end{equation*}
Let then:
\begin{equation*}
    \xi_j' = \one_\set{\cA(\cD)\in\cE_{\gamma,i}} \xi_j(\gamma) + \int_{\gamma'>\gamma} \one_\set{\cA(\cD) \in\dd\cE_{\gamma',i}}\xi_j(\gamma')\,,
\end{equation*}
which is a Bernoulli random variable, as at most only one of the $\zeta_j$ is ‘‘activated''.
	Let us prove by induction that, conditionally on $(X_j,T_j)_{j\in[m]}$, for all $i\in[m]$, $W_i = \sum_{j\leq i} \max(0,T_jS_j)$ is stochastically dominated by $\tilde W_i = \sum_{j\leq i} |T_j|\xi_j'$.
	For $i=0$ this is of course true since both quantities are null..
	Then, if we assume that the result holds at some $ i\leq n$, we have that $W_{i+1}=W_i +  \one_\set{S_{i+1} = \mathrm{sign}(T_{i+1})}$ and $\tilde W_{i+1}=\tilde W_i +  \xi_{i+1}'$.
    Conditionally on the event $\set{W_i=w_i,\tilde W_i=\tilde w_i,\bar T=\bar t,\bar X=\bar x}$ for $w_i,\tilde w_i,\bar t,\bar x$ such that this event is not of probability 0, if $t_{i+1}\ne 0$, using the above computations and writing $\theta=\cA(\cD)$ as the output of the training algorithm:
    \begin{align*}
        &\esp{\one_\set{S_{i+1} = \mathrm{sign}(T_{i+1})} |W_i=w_i,\tilde W_i=\tilde w_i,\bar T=\bar t,\bar X=\bar x}\\
        &= \proba{S_{i+1} = \mathrm{sign}(T_{i+1}) |W_i=w_i,\tilde W_i=\tilde w_i,\bar T=\bar t,\bar X=\bar x}\\
        &=  \proba{S_{i+1} = T_{i+1}|W_i=w_i,\tilde W_i=\tilde w_i,\bar T=\bar t,\bar X=\bar x}\\
        &=  \proba{S_{i+1} = T_{i+1},\theta\in\cE_{\gamma,i+1}|W_i=w_i,\tilde W_i=\tilde w_i,\bar T=\bar t,\bar X=\bar x}\\
        &\quad +  \int_{\gamma'>\gamma} \proba{S_{i+1} = T_{i+1},\theta\in\dd\cE_{\gamma',i+1}|W_i=w_i,\tilde W_i=\tilde w_i,\bar T=\bar t,\bar X=\bar x}\dd\gamma'\\
        &\leq  \proba{S_{i+1} = T_{i+1}|\theta\in\cE_{\gamma,i+1},W_i=w_i,\tilde W_i=\tilde w_i,\bar T=\bar t,\bar X=\bar x}\\
        &\qquad\times\proba{\theta\in\cE_{\gamma,i+1}|W_i=w_i,\tilde W_i=\tilde w_i,\bar T=\bar t,\bar X=\bar x}\\
        &\quad + \int_{\gamma'>\gamma}\proba{S_{i+1} = T_{i+1}|\theta\in\dd\cE_{\gamma',i+1},W_i=w_i,\tilde W_i=\tilde w_i,\bar T=\bar t,\bar X=\bar x}\\
        &\qquad\times\proba{\theta\in\dd\cE_{\gamma',i+1}|W_i=w_i,\tilde W_i=\tilde w_i,\bar T=\bar t,\bar X=\bar x}\dd\gamma'\\
        &\leq  \frac{e^{\gamma+\eps_i}}{1+e^{\gamma+\eps_i}}\\
        &\qquad\times\proba{\theta\in\cE_{\gamma,i+1}|W_i=w_i,\tilde W_i=\tilde w_i,\bar T=\bar t,\bar X=\bar x}\\
        &\quad + \int_{\gamma'>\gamma} \frac{e^{\gamma'+\eps_i}}{1+e^{\gamma'+\eps_i}}\\
        &\qquad\times\proba{\theta\in\dd\cE_{\gamma',i+1}|W_i=w_i,\tilde W_i=\tilde w_i,\bar T=\bar t,\bar X=\bar x}\dd\gamma'\\
        &= \esp{\xi'_{i+1}|W_i=w_i,\tilde W_i=\tilde w_i,\bar T=\bar t,\bar X=\bar x}\,,
    \end{align*}
    using the above computations with the adequate event $\cE$.
    Thus,
    \begin{align*}
        &\esp{\max(0,T_{i+1},S_{i+1}) |W_i=w_i,\tilde W_i=\tilde w_i,\bar T=\bar t,\bar X=\bar x}\\
        &\leq \esp{|T_{i+1}|\xi_{i+1}'|W_i=w_i,\tilde W_i=\tilde w_i,\bar T=\bar t,\bar X=\bar x}\,,
    \end{align*}
    leading to $\max(0,T_{i+1}S_{i+1})$ stochastically dominated by $|T_{i+1}|\xi_{i+1}'$ since stochastic domination for Bernoulli random variables is equivalent to domination of means, conditionally on $W_i$, $\tilde W_i$, $X$ and $T_{\leq i+1}$.
    Using \citet[Lemma 4.9]{steinke2023privacy}: $W_{i+1}$ is stochastically dominated by $\tilde W_{i+1}$, conditionally on $W_i$, $\tilde W_i$ and $T_{\leq i+1}$.
    Thus, $W_n$ is stochastichally dominated by:
    \begin{equation*}
        W_n'=\sum_{i=1}^m|T_i|\xi_i'\,.
    \end{equation*}
    Recall that:
    \begin{equation*}
    \xi_j' = \one_\set{\cA(\cD)\in\cE_{\gamma,i}} \xi_j(\gamma) + \int_{\gamma'>\gamma} \one_\set{\cA(\cD) \in\dd\cE_{\gamma',i}}\xi_j(\gamma')\,,
\end{equation*}
    so that:
    \begin{equation*}
        W_n' = A_m + B_m\,,
    \end{equation*}
    where
    \begin{equation*}
        A_m = \sum_{i=1}^m|T_i|\xi_i(\gamma)\,,
        \end{equation*}    
and
    \begin{align*}
        B_m &\eqdef  \sum_{j=1}^m|T_j|\left(\int_{\gamma'>\gamma} \one_\set{\cA(\cD) \in\dd\cE_{\gamma',i}}\xi_j(\gamma') - \one_\set{\cA(\cD)\notin\cE_{\gamma,i}} \xi_j(\gamma)\right)\\
        & = \sum_{j=1}^m|T_j|\left(\int_{\gamma'>\gamma} \one_\set{\cA(\cD) \in\dd\cE_{\gamma',i}}\big(\xi_j(\gamma') -  \xi_j(\gamma)\big)\right)\\
    \end{align*}
    and we are left with upper bounding this second sum. Note that under our coupling of $\xi_j(\gamma),\xi_j(\gamma')$, each element of the sum in $B_m$ is non-negative, enabling us to use Markov inequality later in the proof.
    For $j\in[m]$,
    \begin{align*}
        &\esp{\int_{\gamma'>\gamma} \one_\set{\cA(\cD) \in\dd\cE_{\gamma',i}}\xi_j(\gamma')} \\
        &= \int_{\gamma'>\gamma} \proba{\cA(\cD) \in\dd\cE_{\gamma',i}}\esp{\xi_j(\gamma')}\dd\gamma'\\
        &= \int_{\gamma'>\gamma} \proba{\cA(\cD) \in\dd\cE_{\gamma',i}} \frac{e^{\gamma'+\eps_j}}{1+e^{\gamma'+\eps_j}} \dd\gamma'\\
        &= \int_{\gamma'>\gamma} \proba{|\PLRV^\cA_{\cD_j,\cD_j'}(\cA(\cD))| = \gamma'} \frac{e^{\gamma'+\eps_j}}{1+e^{\gamma'+\eps_j}} \dd\gamma'\\
        & = \esp{ \frac{e^{|\PLRV^\cA_{\cD_j,\cD_j'}(\cA(\cD))|+\eps_j}}{1+e^{|\PLRV^\cA_{\cD_j,\cD_j'}(\cA(\cD))|+\eps_j}} \Big| |\PLRV^\cA_{\cD_j,\cD_j'}(\cA(\cD))| >\gamma }\proba{\cA(\cD)\notin\cE_{\gamma,j}} \,.
    \end{align*}
    Thus,
    \begin{align*}
        &\esp{\int_{\gamma'>\gamma} \one_\set{\cA(\cD) \in\dd\cE_{\gamma',j}}\xi_j(\gamma') - \one_\set{\cA(\cD)\notin\cE_{\gamma,j}} \xi_j(\gamma)}\\
        &= \esp{ \left(\frac{e^{|\PLRV^\cA_{\cD_j,\cD_j'}(\cA(\cD))|+\eps_j}}{1+e^{|\PLRV^\cA_{\cD_j,\cD_j'}(\cA(\cD))|+\eps_j}} - \frac{e^{\gamma+\eps_j}}{1+e^{\gamma+\eps_j}}\right) \Big| |\PLRV^\cA_{\cD_j,\cD_j'}(\cA(\cD))| >\gamma }\proba{\cA(\cD)\notin\cE_{\gamma,j}}\\
        &= \esp{ \left(\frac{e^{|\PLRV^\cA_{\cD_j,\cD_j'}(\cA(\cD))|+\eps_j}}{1+e^{|\PLRV^\cA_{\cD_j,\cD_j'}(\cA(\cD))|+\eps_j}} - \frac{e^{\gamma+\eps_j}}{1+e^{\gamma+\eps_j}}\right)_+ }\qquad (\text{where}\qquad \forall z\in\R\,,\quad z_+ = \max(0,z))\\
        &\leq \esp{ \left(\frac{e^{|\PLRV^\cA_{\cD_j,\cD_j'}(\cA(\cD))|+\eps_j}}{1+e^{\gamma+\eps_j}} - \frac{e^{\gamma+\eps_j}}{1+e^{\gamma+\eps_j}}\right)_+ }\\
        &\leq \esp{ \left(\frac{e^{|\PLRV^\cA_{\cD_j,\cD_j'}(\cA(\cD))|+\eps_j}-e^{\gamma+\eps_j}}{1+e^{\gamma+\eps_j}} \right)_+ }\\
        &\leq \esp{ \left(\frac{e^{|\PLRV^\cA_{\cD_j,\cD_j'}(\cA(\cD))|+\eps_j}-e^{\gamma+\eps_j}}{e^{\gamma+\eps_j}} \right)_+ }\\
        &\leq \esp{ \left(\frac{e^{|\PLRV^\cA_{\cD_j,\cD_j'}(\cA(\cD))|}-e^{\gamma}}{e^{\gamma}} \right)_+ }\\
        &\leq \esp{ \left(1-e^{\gamma-|\PLRV^\cA_{\cD_j,\cD_j'}(\cA(\cD))|}\right)_+ }\,.\\
    \end{align*}
    Now, under $S_j=1$ we have $\cD=\cD_j$, while under $S_j=-1$ we have $\cD=\cD_j'$. Either way, we have:
    \begin{align*}
                &\esp{\int_{\gamma'>\gamma} \one_\set{\cA(\cD) \in\dd\cE_{\gamma',j}}\xi_j(\gamma') - \one_\set{\cA(\cD)\notin\cE_{\gamma,j}} \xi_j(\gamma)}\\
                &\leq \max\left(\esp{ \left(1-e^{\gamma-|\PLRV^\cA_{\cD_j,\cD_j'}(\cA(\cD_j))|}\right)_+ },\esp{ \left(1-e^{\gamma-|\PLRV^\cA_{\cD_j,\cD_j'}(\cA(\cD_j'))|}\right)_+ } \right)\\
                &\leq \delta'\,.
    \end{align*}
    Each element in the sum defining $B_m$ is of bounded first moment of order $\delta'$, but the two difficulties are that: these events are not independent, and these events are not independent of random variables $T_i$, enabling us to only be able to use a basic Markov inequality.
    We thus have, on the one hand:
    \begin{align*}
        \esp{B_m|\bar X} &\leq m\delta'\,.
    \end{align*}
    Using Markov inequality, 
    \begin{equation*}
            \proba{B_m>v|\bar X}\leq \frac{m\delta'}{v}\,,
    \end{equation*}
    and thus, for any $u,v>0$:
    \begin{align*}
        &\proba{\sum_{i=1}^m\max(0,S_iT_i)>u|\bar X, \bar T}\\
        &\leq \proba{A_m>u-B_m|\bar X,\bar T}\\
        &\leq \proba{\sum_{i=1}^m|T_i|\xi_i>u-v, B_m\leq v|\bar X,\bar T}+\proba{\sum_{i=1}^m|T_i|\xi_i>u-B_m, B_m\geq v|\bar X,\bar T}\\
        &\leq \proba{\sum_{i=1}^m|T_i|\xi_i>u-v|\bar X,\bar T}+\proba{B_m\geq v|\bar X,\bar T}\,.
    \end{align*}
Then, since using Markov inequality:
\begin{align*}
    &\esp{\proba{B_m\geq v|\bar X,\bar T}|\bar X} \\
    & = \proba{B_m\geq v|\bar X}\\
    &\leq \frac{m\delta'}{v}\,,
\end{align*}
using yet again Markov inequality, for any $v,w>0$:
\begin{align*}
    \proba{\proba{B_m\geq v|\bar X,\bar T} > w}\leq \frac{2m\delta'}{vw}\,,
\end{align*}
and we conclude by taking $w=\left(\frac{m\delta'}{v}\right)^{1/2}$.
\end{proof}

\subsection{Proof of \Cref{thm:conditional_audit_eps_delta}}

\begin{proof}[Proof of \Cref{thm:conditional_audit_eps_delta}]
    The proof follows by combining \Cref{prop:hockey} with \Cref{thm:unified_auditing_better}.
\end{proof}

\section{Conditional $f-$DP Auditing: Proof of \Cref{thm:conditional_audit_tampered}.\ref{thm:conditional_audit_tampered:fDP}}
\label{app:proof_f_DP_cond}

Let $\beta(X_i)=\max\set{\frac{\pi(X_i)}{1-\pi(X_i)} ,\frac{1-\pi(X_i)}{\pi(X_i)}} $ so that $\eps_i = \log(\beta(X_i))$, and let $(B_i)_{i\in[m]}$ be independent Bernoulli random variables of parameters $b_i\leq e^{-\eps_i}$, independent of all the rest.
Then, let $T_i'=B_iT_i$, where $T\in\set{-1,0,1}^m$ are the MIA guesses. $T'\in\set{-1,0,1}^m$ can be thought of as \textit{tampered} guesses, in order to to debias for the over confidence induced by distribution shift between members and non-members.

    We first start by proving the following lemma, inspired by \citet[Lemma 11]{mahloujifar2024auditing}. Intuitively, this lemma upper and lower bounds the influence of making a correct guess on the data distribution of the algorithm. It will serve as a back-bone to the proof of \Cref{thm:conditional_audit_tampered}.\ref{thm:conditional_audit_tampered:fDP}.

    \begin{lemma}\label{lemma:f_DP}
        For any measurable $\cE\in\Theta$, 
        \begin{equation*}
            \proba{\cA(\cD)\in\cE,T_1'=S_1 | X_1,\ldots,X_m}\leq f'(\proba{\cA(\cD)\in\cE | X_1,\ldots,X_m})\,,
        \end{equation*}
        where:
        \begin{equation*}
            f'(t) = \sup\set{\alpha: \alpha+f(t-\alpha)\leq 1}\,.
        \end{equation*}
    \end{lemma}
    \begin{proof}[Proof of \Cref{lemma:f_DP}]
        Let $p'=\proba{\cA(\cD)\in\cE, T_1' = S_1 | X_1,\ldots,X_m}$, $p=\proba{\cA(\cD)\in\cE, T_1 = S_1 | X_1,\ldots,X_m}$ and $q=\proba{\cA(\cD)\in\cE | X_1,\ldots,X_m}$.
        First, notice that:
        \begin{align*}
            p' &= \proba{\cA(\cD)\in\cE, T_1' = S_1 | X_1,\ldots,X_m}\\
            & = \proba{\cA(\cD)\in\cE, T_1 = S_1, B_1=1 | X_1,\ldots,X_m}\\
            & = \proba{\cA(\cD)\in\cE, T_1 = S_1 | X_1,\ldots,X_m}\times \proba{B_1=1|X_1}\\
            & = b_1\proba{\cA(\cD)\in\cE, T_1 = S_1 | X_1,\ldots,X_m}\\
            & \leq e^{-\eps_1}\proba{\cA(\cD)\in\cE, T_1 = S_1 | X_1,\ldots,X_m}\\
            & = e^{-\eps_1}p\,.
        \end{align*}
        We then have:
        \begin{align*}
            p &= \proba{\cA(\cD)\in\cE, T_1 = S_1 | X_1,\ldots,X_m,S_1=0} \proba{S_1=0|X_1,\ldots,X_m}\\
            &\quad+ \proba{\cA(\cD)\in\cE, T_1 = S_1 | X_1,\ldots,X_m,S_1=1} \proba{S_1=1|X_1,\ldots,X_m}\\
            &= \proba{\cA(\cD)\in\cE, T_1 = 0 | X_1,\ldots,X_m,S_1=0} (1-\pi(X_1))\\
            &\quad +\proba{\cA(\cD)\in\cE, T_1 = 1 | X_1,\ldots,X_m,S_1=1} \pi(X_1)\\
            &\leq \bar f\left(\proba{\cA(\cD)\in\cE, T_1 = 0 | X_1,\ldots,X_m,S_1=1}\right) (1-\pi(X_1))\\
            &\quad +\bar f\left(\proba{\cA(\cD)\in\cE, T_1 = 1 | X_1,\ldots,X_m,S_1=0}\right) \pi(X_1)\\
            &\leq 1 - f\left(\proba{\cA(\cD)\in\cE, T_1 = 0 | X_1,\ldots,X_m,S_1=1}\right) (1-\pi(X_1))\\
            &\quad - f\left(\proba{\cA(\cD)\in\cE, T_1 = 1 | X_1,\ldots,X_m,S_1=0}\right) \pi(X_1)\,,
        \end{align*}
        by using the $f$-DP assumption on the two adjacent datasets (removing or adding $X_1$), since $T_1$ is also a function of $\cA(\cD)$.
        Now, by convexity and decreasing property of $f$:
        \begin{align*}
            p &\leq 1- f\Big(  (1-\pi(X_1))\proba{\cA(\cD)\in\cE, T_1 = 0 | X_1,\ldots,X_m,S_1=1} \\
            &\quad + \pi(X_1)\proba{\cA(\cD)\in\cE, T_1 = 1 | X_1,\ldots,X_m,S_1=0}\Big)\\
            &\leq 1- f\Big(  \frac{1-\pi(X_1)}{\pi(X_1)}\proba{\cA(\cD)\in\cE, T_1 = 0, S_1=1 | X_1,\ldots,X_m} \\
            &\quad + \frac{\pi(X_1)}{1-\pi(X_1)}\proba{\cA(\cD)\in\cE, T_1 = 1 ,S_1=0| X_1,\ldots,X_m}\Big)\\
            & \leq 1- f\Big( \beta(X_1)\times \proba{\cA(\cD)\in\cE, T_1 \ne S_1 | X_1,\ldots,X_m}\Big)\\
            & = \bar f\Big( e^{\eps_1}\times \proba{\cA(\cD)\in\cE, T_1 \ne S_1 | X_1,\ldots,X_m}\Big)\,,
        \end{align*}
        where $\beta(X_1) = \max\Big(\frac{\pi(X_1)}{1-\pi(X_1)} , \frac{1-\pi(X_1)}{\pi(X_1)}\Big) =e^{\eps_1}$.
        Then, notice that, since we have the event inclusion $\set{T_1'=S_1}\subset\set{T_1=S_1}$:
        \begin{align*}
            &\proba{\cA(\cD)\in\cE, T_1 \ne S_1 | X_1,\ldots,X_m} \\
            &\leq \proba{\cA(\cD)\in\cE, T_1' \ne S_1 | X_1,\ldots,X_m}\\
            &= \proba{\cA(\cD)\in\cE | X_1,\ldots,X_m} - \proba{\cA(\cD)\in\cE, T_1' = S_1 | X_1,\ldots,X_m}\\
            &= q - p'\,.
        \end{align*}
        Thus, since $\bar f$ is concave, non-decreasing, and $\bar f(0)\geq 0$:
        \begin{align*}
            p' & \leq e^{-\eps_1} p\\
            &\leq e^{-\eps_1}\bar f(e^{\eps_1} (q-p'))\\
            &\leq e^{-\eps_1}\bar f(e^{\eps_1} (q-p')) +(1- e^{-\eps_1})\bar f(0)\\
            &\leq \bar f(e^{-\eps_1}e^{\eps_1} (q-p') +(1- e^{-\eps_1})\times 0)\\
            &=\bar f(q-p')\,,
        \end{align*}
        which leads to:\begin{equation*}
            p'+f(q-p')\leq 1\,,
        \end{equation*}
        giving us $p\leq f'(q)$, concluding the proof of \Cref{lemma:f_DP}.
    \end{proof}

We now prove \Cref{thm:conditional_audit_tampered}.\ref{thm:conditional_audit_tampered:fDP}. Although the remaining arguments are exactly the same as in \citet{mahloujifar2024auditing}, we include them for completeness.

\begin{proof}[Proof of \Cref{thm:conditional_audit_tampered}.\ref{thm:conditional_audit_tampered:fDP}]

    We first state the following lemma \citep[Proposition 12]{mahloujifar2024auditing}.

    \begin{lemma}\label{lem:monotone}
                $f'$ is increasing and concave.
    \end{lemma}
    \begin{proof}[Proof of \Cref{lem:monotone}]
        For $t<t'$, we have that for all $\alpha\in\R$ such that $t-\alpha$ and $t'-\alpha$ are in $[0,1]$:
        \begin{equation*}
            \alpha + f(t-\alpha) \geq \alpha + f(t-\alpha)\,,
        \end{equation*}
        since $f$ is decreasing. 
        Thus, $f(t)\leq f(t')$ and $f$ is increasing.

        For concavity, let $t,t'\in\R$. By convexity of $f$, for any $\alpha,\alpha'\in\R$ such that $t-\alpha$ and $t'-\alpha'$ are in $[0,1]$:
        \begin{equation*}
            f(\frac{t+t'}{2}-\frac{\alpha+\alpha'}{2}) \leq \frac{f(t-\alpha)+f(t'-\alpha')}{2}\,,
        \end{equation*}
        by convexity of $f$. Thus, if $\alpha,\alpha'$ satisfy $f(t-\alpha)+\alpha \leq 1$ and $f(t'-\alpha')+\alpha' \leq 1$, we have 
        \begin{equation*}
             \frac{\alpha+\alpha'}{2}+ \frac{\alpha + f(t-\alpha)+\alpha'+ f(t'-\alpha)}{2}\leq 1\,,
        \end{equation*}
        and thus $\frac{\alpha+\alpha'}{2} + f(\frac{t+t'}{2}-\frac{\alpha+\alpha'}{2})\leq 1$. We can thus conclude that $f(\frac{t+t'}{2})\geq \frac{\alpha+\alpha'}{2}$ and taking the sup over $\alpha,\alpha'$ satisfying $f(t-\alpha)+\alpha \leq 1$ and $f(t'-\alpha')+\alpha' \leq 1$, we have concavity.
    \end{proof}


    
    As in \citet[Theorem 9]{mahloujifar2024auditing}, we assume that the adversary makes guesses on all $m$ inputs with a vector $q\in\set{0,1}^m$ with exactly $r$ 1's and $m-r$ 0's. Only correct guesses that are in locations where $q_i=1$ are counted.
    If $t_i=\one_\set{S_i=T_i'}$, we thus have 
    \begin{align*}
        p_k &= \proba{\sum_{i=1}^m t_iq_i = k\,\big|\, (X_i)_{i\in[m]}}\\
        &= \proba{\sum_{i=2}^m t_iq_i=k-1\text{ and }t_1=1\text{ and }q_1=1\,\big|\, (X_i)_{i\in[m]}}\\
        &\quad+\proba{\sum_{i=2}^m t_iq_i=k\text{ and }t_1q_1=0\,\big|\, (X_i)_{i\in[m]}}\,.
    \end{align*}
    Here, we could use \Cref{lemma:f_DP}, which would yield:
    \begin{align*}
        &\proba{\sum_{i=2}^m t_iq_i=k-1\text{ and }t_1=1\text{ and }q_1=1\,\big|\, (X_i)_{i\in[m]}} \\
        &\leq f'\left( \proba{\sum_{i=2}^m t_iq_i=k-1\text{ and }q_1=1\,\big|\, (X_i)_{i\in[m]}}\right)\,.
    \end{align*}
    To obtain sharper bounds, we instead sum over all $k\in T$ and use the same reasoning:
    \begin{align*}
        \sum_{k\in T}p_k &= \sum_{k\in T} \proba{\sum_{i=1}^m t_iq_i = k\,\big|\, (X_i)_{i\in[m]}}\\
        & = \sum_{k\in T} \proba{\sum_{i=2}^m t_iq_i=k-1\text{ and }t_1=1\text{ and }q_1=1\,\big|\, (X_i)_{i\in[m]}}\\
        &\quad+\proba{\sum_{i=2}^m t_iq_i=k\text{ and }t_1q_1=0\,\big|\, (X_i)_{i\in[m]}}\\
        &= \proba{1+\sum_{i=2}^m t_iq_i\in T\text{ and }t_1=1\text{ and }q_1=1\,\big|\, (X_i)_{i\in[m]}}\\
        &\quad+\proba{\sum_{i=2}^m t_iq_i\in T\text{ and }t_1q_1=0\,\big|\, (X_i)_{i\in[m]}}\,.
    \end{align*}
    We now use \Cref{lemma:f_DP} on the first term:
    \begin{align*}
         \sum_{k\in T}p_k & \leq f'\left(\proba{1+\sum_{i=2}^m t_iq_i\in T\text{ and }q_1=1\,\big|\, (X_i)_{i\in[m]}}\right)\\
        &\quad+\proba{\sum_{i=2}^m t_iq_i\in T\text{ and }t_1q_1=0\,\big|\, (X_i)_{i\in[m]}}\,.
    \end{align*}
    The RHS of the bound is now invariant under a permutation of the order of indices, so that for a random permutation $\pi$ sampled uniformly at random in all permutations on $[n]$, we have:
    \begin{align*}
        \sum_{k\in T}p_k & \leq \E_\pi f'\left(\proba{1+\sum_{i=2}^m t_{\pi(i)}q_{\pi(i)}\in T\text{ and }q_{\pi(1)}=1\,\big|\, (X_i)_{i\in[m]}}\right)\\
        &\quad +\E_\pi\proba{\sum_{i=2}^m t_{\pi(i)}q_{\pi(i)}\in T\text{ and }t_{\pi(1)}q_{\pi(1)}=0\,\big|\, (X_i)_{i\in[m]}}\\
        & \leq f'\left( \E_\pi \proba{1+\sum_{i=2}^m t_{\pi(i)}q_{\pi(i)}\in T\text{ and }q_{\pi(1)}=1\,\big|\, (X_i)_{i\in[m]}}\right)\\
        &\quad+\E_\pi\proba{\sum_{i=2}^m t_{\pi(i)}q_{\pi(i)}\in T\text{ and }t_{\pi(1)}q_{\pi(1)}=0\,\big|\, (X_i)_{i\in[m]}}\,,
    \end{align*}
    using concavity of $f'$.
    The reasoning is then the same as \citet{mahloujifar2024auditing}: when permuting the order, $\sum_{i=2}^mt_{\pi(i)}q_{\pi(i)} = k$ and $t_{\pi(1)}q_{\pi(1)}=0$ count each instances $tq$ with exactly $k$ non-zero locations, for $(m-k)\times(m-1)!$ times, leading to:
    \begin{align*}
        \E_\pi\proba{\sum_{i=2}^m t_{\pi(i)}q_{\pi(i)}\in T\text{ and }t_{\pi(1)}q_{\pi(1)}=0\,\big|\, (X_i)_{i\in[m]}} = \sum_{k\in T}\frac{m-k}{m}p_k\,.
    \end{align*}
    Similarly for the other term, $\sum_{i=2}^mt_{\pi(i)}q_{\pi(i)} = k-1$ and $q_{\pi(1)}=1$ count each instances $t,q$ that satisfy $t_iq_i=1$ for $k-1$ indices exactly ($(r-k+1)(m-1)!$ times each) and each instances $t,q$ that satisfy $t_iq_i=1$ for $k$ indices exactly ($k(m-1)!$ times each), leading to:
    \begin{align*}
        \E_\pi\proba{1+\sum_{i=2}^m t_{\pi(i)}q_{\pi(i)}\in T\text{ and }q_{\pi(1)}=1\,\big|\, (X_i)_{i\in[m]}} = \sum_{k\in T} \frac{k}{m}p_k + \frac{r-k+1}{m}p_{k-1}\,.
    \end{align*}
    Thus,
    \begin{align*}
        \sum_{k\in T}p_k &\leq f'\left( \sum_{k\in T} \frac{k}{m}p_k + \frac{r-k+1}{m}p_{k-1} \right) + \sum_{k\in T}\frac{m-k}{m}p_k\,,
    \end{align*}
    leading to:
    \begin{align*}
        \sum_{k\in T}\frac{k}{m}p_k &\leq f'\left( \sum_{k\in T} \frac{k}{m}p_k + \frac{r-k+1}{m}p_{k-1} \right)\,.
        \end{align*}
    Thus, by definition of $f'$, 
    \begin{align*}
        \sum_{k\in T}\frac{k}{m}p_k &\leq 1-f\left( \sum_{k\in T}  \frac{r-k+1}{m}p_{k-1} -f\left( \sum_{k\in T}  \frac{r-k+1}{m}p_{k-1} \right) \right)\\
        &\leq 1-f\left( \sum_{k\in T} \frac{r-k+1}{m}p_{k-1}\right)\\
        &= \bar f\left( \sum_{k\in T} \frac{r-k+1}{m}p_{k-1}\right)\,,
    \end{align*}
    since $\bar f=1-f$ and $f$ is decreasing. This concludes the proof.
\end{proof}

\begin{proof}[Proof of \Cref{cor:validity_oracle} (conditional $f$-DP part)]
    The proof of is a direct adaptation of the proof \citet[Theorem 10]{mahloujifar2024auditing} (which is quite technically involved).
\end{proof}

\section{Propensity Score Estimation Algorithms}
\label{app:prop_learning_top}

\subsection{Propensity Score Estimation with Cross-Fitting}
    \label{app:propensity_learning}


We use \emph{cross-fitting} to estimate propensity scores without evaluating the classifier on the same examples used for training. We randomly split the data into two folds, train a propensity model on each fold, and use each model to predict membership for the other fold. Thus, every $\hat\pi_i$ is obtained from a classifier that did not see $X_i$ during training, reducing overfitting and ensuring that the estimated propensity scores reflect out-of-sample predictive performance. The procedure is summarized in \Cref{alg:propensity_learning}.

\begin{algorithm}[H]
    \caption{Propensity score learning}
    \label{alg:propensity_learning}
    \begin{algorithmic}[1]
        \STATE \textbf{Inputs:} 
        $(X_i,S_i)_{i\in[m]}\in(\cX\times \set{-1,1})^m$
                
        \STATE Construct two random balanced splits $\cI_1,\cI_2$ of $[m]$
        
        \FOR{$k \in \set{1,2}$}
            \STATE Fit classifier $\hat\pi^{(k)}$ on $(X_i,S_i)_{i\in\cI_k}$ to predict $\proba{S_i=1|X_i}$
        \ENDFOR
        \FOR{$i\in[m]$}
            \STATE $\hat\pi_i\gets \left\{\begin{aligned}
                & \hat \pi^{(2)}(X_i)\quad\text{if}\quad i\in\cI_1\\
                & \hat \pi^{(1)}(X_i)\quad\text{if}\quad i\in\cI_2
            \end{aligned}\right.$
        \ENDFOR
        
        \STATE \textbf{Return} $(\hat\pi_i)_{i\in[m]}$
    \end{algorithmic}
\end{algorithm}

\subsection{Bootstrap Algorithm}
\label{app:bootstrap}

\Cref{alg:bootstrap} formalizes our bootstrap procedure.
Given a function $\Psi:[0,1]^m\to \R$, the bootstrap procedure produces confidence intervals around:
\begin{equation*}
    \Psi((\pi(X_i))_{i\in[m]})\,,
\end{equation*}
using only estimates of the propensity scores.
We then use this by seeing empirical lower bounds as functions of the propensity scores, all the rest being fixed. 

\begin{algorithm}[H]
    \caption{Bootstrap upper confidence estimate
\label{alg:bootstrap}
    }
    \begin{algorithmic}[1]
      \STATE \textbf{Inputs:} Confidence probability $1-p'$, function $\Psi:[0,1]^m\to\R$, number of bootstraps $K\in\N^*$, $(X_i)_{i\in[m]} \cX^m$, training dataset for propensity scores $\cD'\in(\cX\times\set{-1,1})^{m'}$.
      \FOR{$k\in\set{1,\ldots,K}$}
      \STATE $\cD_k'\gets $ subsampling of $\cD'$ of size $m'$ with replacement
      \STATE Run \Cref{alg:propensity_learning} on $\cD_k'$ to obtain $(\tilde\pi_i^{(k)})_{i\in[m]}$
      \STATE $\tilde v_k\gets \Psi((\tilde\pi_i^{(k)}(X_i))_{i\in[m]})$
      \ENDFOR
      \STATE Run \Cref{alg:propensity_learning} on $\cD$ to obtain $(\hat\pi_i)_{i\in[m]}$
      \STATE $\tilde v\gets \Psi((\hat\pi(X_i))_{i\in[m]})$
      \STATE $\hat v_k \gets \tilde v_k + \tilde v - \mathrm{median}((\tilde v_k)_{k\in[K]} $
      \STATE \textbf{Return} $\mathrm{Quantile}\big((\hat v_k)_{k\in[K]},1-p'\big)$
    \end{algorithmic}
\end{algorithm}

\section{Experimental Details: Propensity Score Estimation}
\label{app:propensity-estimation-details}

\paragraph{Estimand and sampling design.}
For each audit comparison, we estimate the propensity score introduced in
\Cref{def:propensity},
\begin{equation}
    \pi(x)
    =
    \Pr(S_i=1 \mid X_i=x),
\end{equation}
where \(S_i=1\) denotes a member and \(S_i=-1\) denotes a non-member.
The propensity score is defined relative to the sampling design used to
construct the audit population. We therefore balance the two source
populations before fitting the propensity model: when
\(|\mathcal D_1| \neq |\mathcal D_0|\), we subsample the larger population
without replacement so that $ n_1 = |\mathcal D_1| = |\mathcal D_0| = n_0.$
Consequently, under the resulting audit design, $ \Pr(S_i=1) = \Pr(S_i=-1) = \frac{1}{2}$
and a calibrated classifier probability can be interpreted directly as
\(\pi(x)\) for this balanced population. Under this design, the propensity score also determines the pointwise
distribution-shift leakage used by the Zero-Run correction:
\begin{equation}
    \epsilon_{\mathrm{DS}}(x)
    =
    \left|
        \log
        \left(
            \frac{\pi(x)}{1-\pi(x)}
        \right)
    \right|.
    \label{eq:estimated-ds-leakage}
\end{equation}

Only covariates derived from \(X_i\) are supplied as propensity-model features.
The membership indicator \(S_i\) is used as the supervised classification
target, but task labels and all quantities derived from the audited model,
including its losses, logits, predictions, and internal representations, are
excluded. This separation prevents the propensity estimator from attributing
model-induced privacy leakage to distribution shift.

\paragraph{Input representations.}
For high-dimensional inputs, we use fixed representations computed by
pretrained encoders that are independent of the audited model. Text inputs are
represented using the 1,024-dimensional final-layer CLS embedding of
\texttt{BAAI/bge-m3}. The embeddings are \(L_2\)-normalized, input sequences
are truncated at 8,192 tokens, and no retrieval instruction or task-specific
prefix is added.

Images are represented using the 768-dimensional final-layer CLS embedding of
\texttt{google/vit-base-patch16-224-in21k}
\citep{dosovitskiy2021imageworth16x16words}. Images are resized to
\(224\times224\) using bilinear interpolation, rescaled by \(1/255\), and
normalized channelwise using mean \(0.5\) and standard deviation \(0.5\).

\paragraph{Propensity models.}
We compare six propensity-estimation pipelines spanning linear and nonlinear
model classes and alternative probability-calibration methods. Specifically,
we consider \(L_2\)-regularized logistic regression \citep{scikit-learn} with
either sigmoid (Platt) or isotonic calibration, a linear SVM with sigmoid
calibration, a random forest with sigmoid calibration, an extremely randomized
tree ensemble with sigmoid calibration, and histogram gradient boosting with
sigmoid calibration. Coordinatewise feature standardization is applied to the
linear models. \Cref{tab:propensity-models} summarizes the complete set of
estimators.

\begin{table}[t]
    \centering
    \small
    \renewcommand{\arraystretch}{1.15}
    \caption{
        Propensity-score model specifications. All reported propensity scores
        are obtained using five-fold outer cross-fitting and three-fold inner
        probability calibration. Optional PCA is fitted exclusively within the
        corresponding training folds.
    }
    \label{tab:propensity-models}

    \begin{tabularx}{\textwidth}{
        >{\raggedright\arraybackslash}p{0.28\textwidth}
        >{\raggedright\arraybackslash}p{0.22\textwidth}
        >{\raggedright\arraybackslash}p{0.25\textwidth}
        >{\raggedright\arraybackslash}X
    }
        \toprule
        \textbf{Estimator}
        &
        \textbf{Preprocessing}
        &
        \textbf{Uncalibrated score}
        &
        \textbf{Calibration}
        \\
        \midrule

        \(L_2\)-regularized logistic regression
        &
        Standardization
        &
        Decision function
        &
        Sigmoid
        \\

        \addlinespace

        \(L_2\)-regularized logistic regression
        &
        Standardization
        &
        Decision function
        &
        Isotonic
        \\

        \addlinespace

        Linear SVM
        &
        Standardization
        &
        Signed margin
        &
        Sigmoid
        \\

        \addlinespace

        Random forest
        &
        None
        &
        Class probability
        &
        Sigmoid
        \\

        \addlinespace

        Extremely randomized trees
        &
        None
        &
        Class probability
        &
        Sigmoid
        \\

        \addlinespace

        Histogram gradient boosting
        &
        None
        &
        Class probability
        &
        Sigmoid
        \\

        \bottomrule
    \end{tabularx}
\end{table}

\paragraph{Cross-fitting and calibration.}
All propensity scores used in the audit are computed out of sample using
five-fold stratified cross-fitting. Within each outer training fold, probability
calibration is learned using three inner folds, with sigmoid or isotonic
calibration depending on the estimator. All learned preprocessing, including
optional PCA and feature standardization, is fitted exclusively on the
corresponding training data. Consequently, every reported
\(\widehat{\pi}(X_i)\) is obtained from a pipeline that was not trained on
\(X_i\), and all downstream diagnostics and privacy calculations use these
held-out predictions.

\paragraph{Calibration and overlap diagnostics.}
We retain the estimated propensity scores without clipping or shrinking them
toward \(1/2\), since extreme values provide direct evidence of poor overlap
between \(\mathbb P_1\) and \(\mathbb P_0\). Through
\(\widehat{\epsilon}_{\mathrm{DS}}(X_i)
=|\log(\widehat{\pi}(X_i)/(1-\widehat{\pi}(X_i)))|\), such observations incur
a larger distribution-shift correction and contribute less to the conditional
audit. We evaluate the propensity estimators using held-out ROC AUC, log loss,
Brier score, expected calibration error, and calibration intercept and slope,
together with overlap diagnostics such as the propensity-score distribution,
the fraction of extreme scores, empirical common support, and the effective
sample size induced by the correction. AUC measures population separability,
whereas calibration and overlap diagnostics determine whether the estimated
probabilities provide a reliable basis for the Zero-Run correction; high AUC
alone does not imply adequate overlap.

\paragraph{Propensity-estimation uncertainty.}
Finally, we propagate uncertainty in propensity estimation through the final
privacy audit. We use \(K=600\) bootstrap replicates. In each replicate, we
resample observations with replacement within the member and non-member
populations and refit the complete propensity-estimation procedure, including
all learned preprocessing,  classifier fitting and calibration. This produces replicate-specific propensity
scores
\[
    \left(
        \widehat{\pi}^{(k)}(X_i)
    \right)_{i=1}^m,
    \qquad
    k=1,\ldots,K.
\]

For each replicate, we recompute the Zero-Run correction and extract the
corresponding empirical privacy lower bound
\(\widehat{\mu}^{(k)}\). The audited model output and the MIA scores are held
fixed across these propensity bootstrap replicates; the resampling therefore
targets uncertainty arising from estimation of the nuisance propensity
function.

As in \Cref{sec:propensity_estimation}, we split the overall error budget
between the privacy audit and propensity uncertainty. In our experiments,
\(p=p'=0.025\). We report the lower \(p'\)-quantile
\begin{equation}
    \widehat{\mu}_{\mathrm{boot}}
    =
    \operatorname{Quantile}_{p'}
    \left(
        \widehat{\mu}^{(1)},
        \ldots,
        \widehat{\mu}^{(K)}
    \right),
    \qquad
    p'=0.025.
    \label{eq:bootstrap-lower-bound}
\end{equation}
This conservative summary accounts for both finite-sample variability and
instability of the propensity-estimation pipeline. As discussed in
\Cref{sec:propensity_estimation}, its formal bootstrap interpretation requires
the usual regularity assumptions for the functional mapping propensity scores
to the resulting empirical privacy lower bound.

\section{Experimental Synthetic Distribution Shift}
\label{appendix:gaussian_syntethic_data}

We consider \(m=20{,}000\) candidate records represented by the standard basis vectors
\(e_1,\ldots,e_m\). Each record is independently assigned to be a member or a non-member:
\begin{equation}
    Y_i\sim\operatorname{Bernoulli}(1/2),
\end{equation}
where \(Y_i=1\) denotes a member. The member and non-member sets are therefore
\begin{equation}
    \mathcal{D}_{\mathrm{mem}}
    =
    \{e_i:Y_i=1\},
    \qquad
    \mathcal{D}_{\mathrm{non}}
    =
    \{e_i:Y_i=0\}.
\end{equation}

The mechanism releases one noisy sum of all member records:
\begin{equation}
    \theta_{\mathrm{priv}}
    =
    \sum_{i=1}^{m}Y_i e_i+\eta,
    \qquad
    \eta\sim\mathcal{N}(0,\sigma^2 I_m).
\end{equation}
The Gaussian vector \(\eta\) is sampled once for the entire release. Since adding or removing one record changes the sum by one basis vector, the \(L_2\) sensitivity is \(1\), and the mechanism satisfies \(\mu\)-GDP with
\begin{equation}
    \mu_{\mathrm{true}}=\frac{1}{\sigma}.
\end{equation}
We use \(\sigma=1.5\), giving \(\mu_{\mathrm{true}}=2/3\). For candidate \(i\), the membership score is its coordinate in the released vector:
\begin{equation}
    s_i
    =
    \langle e_i,\theta_{\mathrm{priv}}\rangle
    =
    Y_i+\eta_i.
\end{equation}
Thus, member scores follow \(\mathcal{N}(1,\sigma^2)\), while non-member scores follow
\(\mathcal{N}(0,\sigma^2)\). These scores are all obtained from the same noisy release. To introduce distribution shift, we randomly order the candidates once. At shift level \(\lambda\), we select the first
\begin{equation}
    m_\lambda
    =
    \operatorname{round}\left(m\frac{\lambda}{4}\right)
\end{equation}
candidates. Each selected candidate receives an additional feature that agrees with its membership label with probability
\begin{equation}
    q_\lambda=\frac{1}{1+\exp(-\lambda)}
\end{equation}
and takes the opposite value otherwise. Unselected candidates receive the value zero.

The same membership labels, noisy release, and membership scores are reused for every value of \(\lambda\). Only the additional feature changes. Therefore, increasing \(\lambda\) changes the member--non-member distribution shift without changing the underlying Gaussian mechanism. This extra feature should be seen as "advice" given to the auditor that leaks membership information and encapsulates the full distribution shift. 

\section{CIFAR-10 Experiments}
\label{app:cifar10-experiments}

\subsection{Synthetic Data Generation for CIFAR-10}
\label{app:synth-data-cifar10}

We construct a synthetic non-member dataset for the CIFAR-10 training set using a pretrained class-conditional DDO-EDM \cite{pmlr-v267-zheng25b, NEURIPS2022_a98846e9} diffusion model. We generate a pool of 104,000 candidates using the deterministic EDM Heun \cite{NEURIPS2022_a98846e9} sampler with 18 denoising steps and no stochastic churn. Each candidate is generated from a unique noise seed and conditioned on its CIFAR-10 class. We allocate 8,000 candidates to each class, except bird, cat, and dog, for which we generate 16,000 candidates because these classes exhibited the largest utility gaps in preliminary experiments. The final dataset contains 50,000 images, balanced at 5,000 examples per class.

We select candidates exclusively from real CIFAR-10 training examples. We deterministically split the 50,000 training images into a 45,000-image reference set and a 5,000-image holdout set, using 4,500 reference examples per class for selection. Images are represented using the concatenated, L2-normalized penultimate features of independently trained CIFAR-10 ResNet-18 \cite{he2016deep} classifiers. We first remove only clear semantic failures: we discard a generated image when the classifier ensemble predicts a class different from its conditioning label with probability greater than 0.9. This removes 1,083 of the 104,000 candidates; importantly, we retain low-confidence or ambiguous examples rather than filtering the pool toward artificially easy ones.

We then select 5,000 images independently for each class to match the coverage and diversity of the real training distribution. The real reference embeddings of each class are partitioned into 250 clusters, and the synthetic selection budget is distributed across clusters proportionally to their frequency in real CIFAR-10. Within each cluster, candidates are greedily selected using a facility-location criterion that rewards additional coverage of real examples, diversity relative to already selected synthetic examples, and agreement with the real data's characteristic class-confusion structure. Additional quotas over classifier entropy and prediction margin preserve the mixture of easy, ambiguous, and difficult examples found in the real training set. Thus, selection optimizes the synthetic set's properties collectively rather than choosing the most confidently classified or individually closest images.

\subsection{CIFAR-10 victim model training}
\label{app:cifar10-model-training}

We used a Wide ResNet-50-2 (\texttt{wide\_resnet50\_2}, approximately 66.9M parameters) adapted to CIFAR-10. The original $7\times7$ input convolution was retained with stride reduced from 2 to 1, the initial max-pooling layer was removed, and images were processed at their native $32\times32$ resolution. We considered four regimes designed to separate the effects of pretraining, classifier capacity, and full-model fine-tuning. Pretrained models were initialized from the torchvision ImageNet-1K V2 weights before replacing the classification head. In the two head-only settings, the backbone was frozen and kept in evaluation mode, ensuring that BatchNorm statistics were also fixed.

\begin{table}[t]
    \centering
    \small
    \begin{tabular}{llll}
        \toprule
        \textbf{Scenario} & \textbf{Classifier} & \textbf{Trainable parameters} & \textbf{Base LR} \\
        \midrule
        \texttt{pretrained\_linear}
        & Linear$(2048,10)$
        & 20,490
        & $0.1$ \\

        \texttt{pretrained\_mlp}
        & $2048{\rightarrow}4096{\rightarrow}2048{\rightarrow}10$
        & 16,803,850
        & $0.05$ \\

        \texttt{pretrained\_full}
        & Linear$(2048,10)$
        & 66,854,730
        & backbone $0.01$, head $0.1$ \\

        \texttt{scratch\_full}
        & Linear$(2048,10)$
        & 66,854,730
        & $0.2$ \\
        \bottomrule
    \end{tabular}
    \caption{Wide ResNet-50-2 victim-model configurations. The MLP classifier is
    $2048{\rightarrow}4096{\rightarrow}\mathrm{ReLU}{\rightarrow}
    2048{\rightarrow}\mathrm{ReLU}{\rightarrow}10$.}
    \label{tab:wrn_training_scenarios}
\end{table}

All models were trained for 100 epochs using only the 50,000 real CIFAR-10 training examples. Synthetic CIFAR-10 data and the official test set were reserved exclusively for evaluation. Inputs were scaled to $[0,1]$ and normalized with the ImageNet mean $(0.485,0.456,0.406)$ and standard deviation $(0.229,0.224,0.225)$. To avoid additional sources of regularization or augmentation, we used no random crops, horizontal flips, MixUp, CutMix, label smoothing, dropout, stochastic depth, or weight decay. Optimization used SGD with momentum $0.9$, batch size 512, and cross-entropy loss, resulting in 98 updates per epoch and 9,800 updates overall. The learning rate followed a five-epoch linear warm-up and cosine decay. 

To track the evolution of memorization, each model was evaluated before training and after every epoch on the real training set, synthetic evaluation set, and official CIFAR-10 test set. Evaluation was performed without augmentation and recorded per-example loss, prediction confidence, and gradient-norm statistics used by the subsequent privacy analysis. All experiments used seed 42 for initialization, shuffling, and random-number generators, together with deterministic PyTorch and cuDNN execution. These choices make the four scenarios directly comparable and ensure that synthetic examples influence only the evaluation and auditing stages, not model optimization.

\subsection{Propensity Score Analysis}
\label{sec:propensity_cifar}

\begin{figure}[H]
  \centering
  \begin{subfigure}[t]{0.485\textwidth}
    \vspace{0pt}
    \centering
    \includegraphics[width=\linewidth]{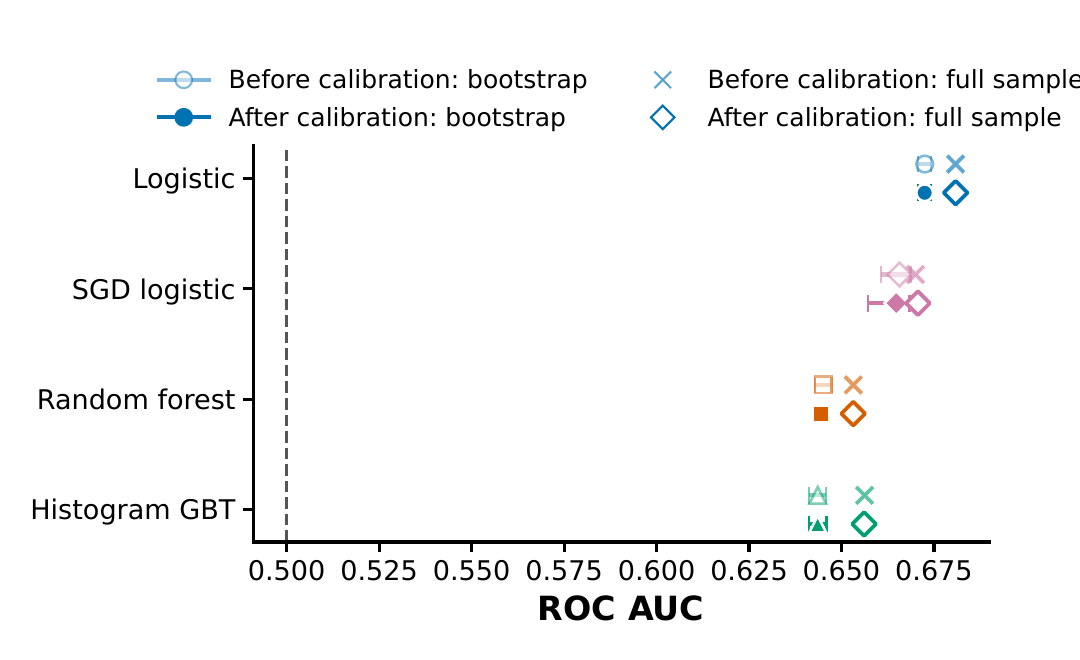}
    \caption{ROC AUC before and after calibration.}
    \label{fig:cifar10-propensity-auc}
  \end{subfigure}
  \hfill
  \begin{subfigure}[t]{0.485\textwidth}
    \vspace{0pt}
    \centering
    \includegraphics[width=\linewidth]{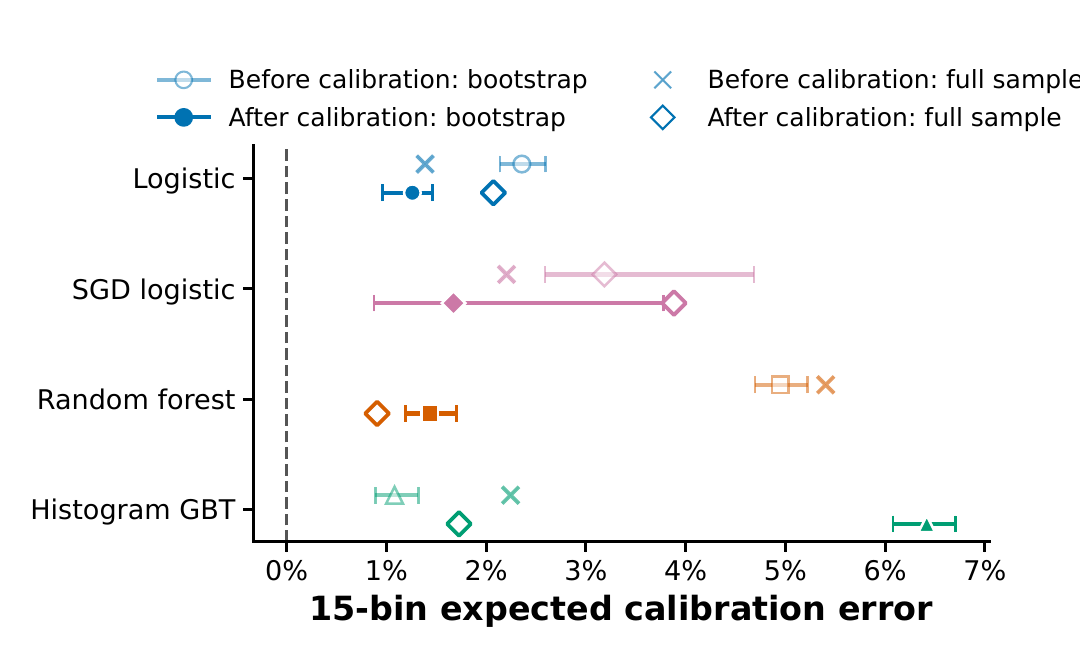}
    \caption{15-bin ECE before and after calibration.}
    \label{fig:cifar10-propensity-ece}
  \end{subfigure}

  \vspace{0.6em}

  \begin{subfigure}[t]{0.485\textwidth}
    \vspace{0pt}
    \centering
    \includegraphics[width=\linewidth]{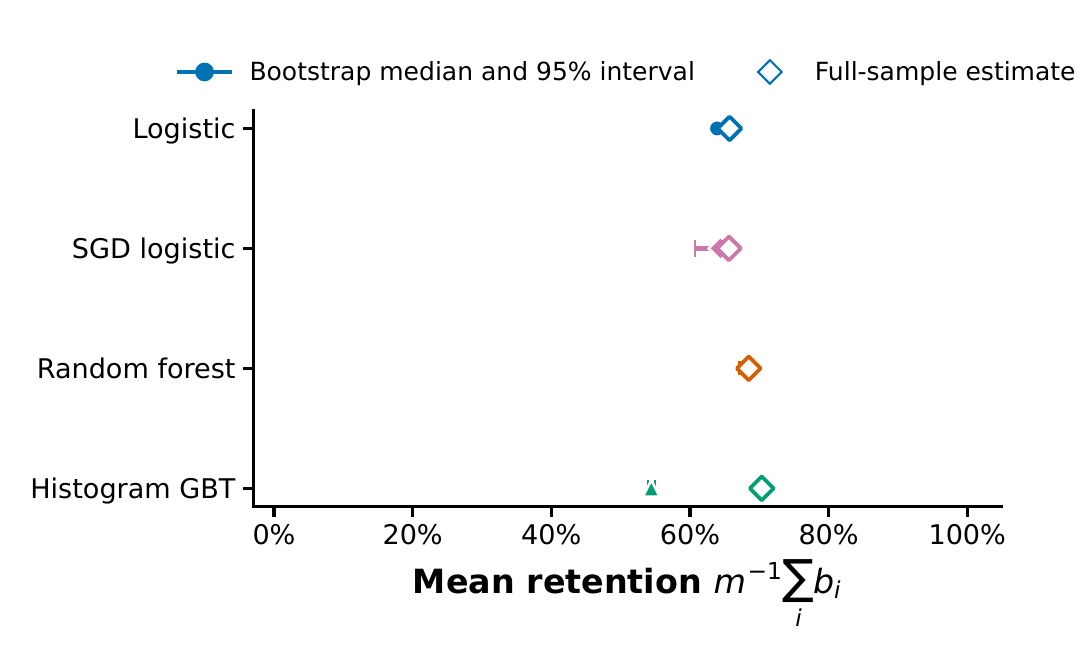}
    \caption{Mean f-DP retention probability.}
    \label{fig:cifar10-propensity-retention}
  \end{subfigure}
  \hfill
  \begin{subfigure}[t]{0.485\textwidth}
    \vspace{0pt}
    \centering
    \includegraphics[width=\linewidth]{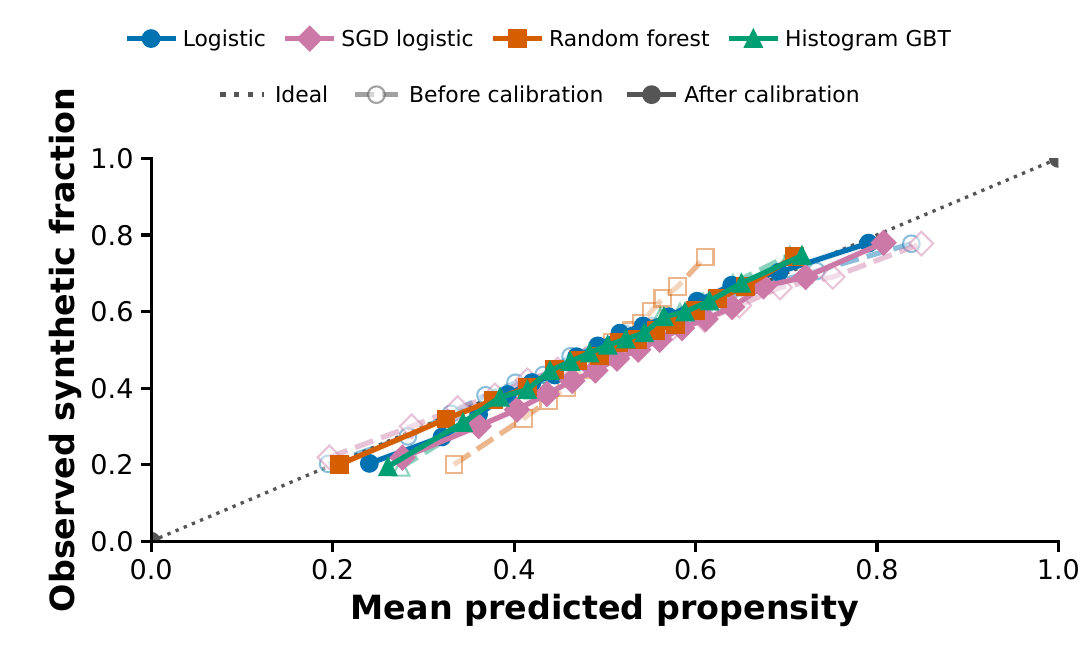}
    \caption{Reliability before and after calibration.}
    \label{fig:cifar10-propensity-reliability}
  \end{subfigure}

  \caption{Propensity-model diagnostics for synthetic versus real CIFAR-10
  embeddings without PCA. Panels (a) and (b) compare each estimator before and
  after calibration using 600 bootstrap replicates and full-sample estimates;
  panel (d) compares the corresponding full-sample reliability curves. Panel
  (c) reports the calibrated probabilities used for f-DP rejection, where
  $b_i=\min(\hat\pi_i,1-\hat\pi_i)/\max(\hat\pi_i,1-\hat\pi_i)$, so
  $m^{-1}\sum_i b_i$ is the expected fraction of images retained.}
  \label{fig:cifar10-propensity-summary}
\end{figure}

\Cref{fig:cifar10-propensity-summary} provides a sensitivity analysis of the
estimated distribution-shift correction with respect to the choice of
propensity model. Across logistic regression, SGD logistic regression, random
forest, and histogram gradient boosting, the held-out ROC AUC remains in the
relatively narrow range of approximately \(0.64\)--\(0.68\)
(\Cref{fig:cifar10-propensity-auc}). Thus, the ability to distinguish synthetic
from real CIFAR-10 examples is not an artifact of a single classifier family:
both linear and nonlinear estimators consistently identify a non-negligible
feature-level shift. At the same time, the AUC values remain far from one,
indicating that the two populations retain substantial overlap. This is the
regime in which the Zero-Run correction is most useful: treating the synthetic
examples as IID non-members would ignore measurable confounding, whereas the
overlap is still sufficient to obtain an informative audit. As expected,
calibration has little effect on AUC because it primarily changes the scale of
the predicted probabilities rather than the ranking of observations. Probability calibration is nevertheless important for the privacy correction
itself. In particular, the estimated pointwise shift,
\[
    \widehat{\epsilon}_{\mathrm{DS}}(X_i)
    =
    \left|
        \log\left(
            \frac{\hat{\pi}_i}{1-\hat{\pi}_i}
        \right)
    \right|,
\]
is sensitive to probability errors, especially when \(\hat{\pi}_i\) is close
to zero or one. Panels~(b) and~(d) therefore provide information that cannot be
recovered from AUC alone. Calibration substantially reduces the error of the
random forest, whose ECE falls from roughly \(5\%\) to near \(1\%\), while the
linear models exhibit lower post-calibration bootstrap ECEs but some
disagreement between their bootstrap and full-sample estimates. Histogram
gradient boosting displays the strongest instability: its post-calibration
full-sample ECE is below \(2\%\), whereas the corresponding bootstrap estimate
is considerably larger. Although the reliability curves generally follow the
diagonal over the central propensity range, the remaining tail deviations and
bootstrap variability show that an apparently well-calibrated full-sample fit
can understate uncertainty in the nuisance model.

This variation is propagated directly to the effective size of the audit.
As shown in \Cref{fig:cifar10-propensity-retention}, the mean retention
\(m^{-1}\sum_i b_i\) is approximately \(55\%\)--\(70\%\) across the evaluated
estimators. The correction therefore preserves a meaningful majority of the
observations, but reduces the effective evidence available to the audit by
roughly \(30\%\)--\(45\%\). The qualitative conclusion is robust across model
classes: synthetic and real images are detectably different, yet their overlap
is sufficient for a non-vacuous correction. The precise correction magnitude,
however, remains sensitive to propensity-model fitting and calibration,
particularly for the nonlinear estimator. Consequently, downstream privacy
bounds should be accompanied by bootstrap uncertainty and reported across
multiple propensity-model classes rather than selecting the specification
that produces the largest lower bound. Stability across these choices provides
stronger evidence that the final audit reflects model-induced membership
leakage rather than an artifact of the propensity estimator.

\subsection{Audit Analysis}

\Cref{fig:cifar10-audit-trajectories} reports the conservative bootstrap
bound \(\widehat{\mu}_{\mathrm{boot}}\) for checkpoints from each model's
single 100-epoch training run. At each checkpoint, the victim model and
its MIA scores are held fixed while the Zero-Run audit is recomputed over
the \(K=600\) bootstrap propensity fits; each colored point is the
resulting lower \(0.025\) quantile. The corrected trajectories begin near
zero and increase as optimization proceeds, showing that the audit
responds to model-dependent membership signal accumulated during
training rather than only to the static shift between real and synthetic
images. The close agreement among the four propensity estimators,
particularly after the initial transient, further indicates that the
temporal pattern is not tied to a single nuisance-model family.

The trajectories also distinguish the four training regimes. The frozen
pretrained linear head exhibits little corrected leakage throughout and
ends below \(\widehat{\mu}=0.09\). The pretrained MLP and the model trained
fully from scratch accumulate leakage more gradually before plateauing
near \(0.6\), whereas full fine-tuning from pretrained weights produces
an earlier increase and stabilizes around \(0.45\)--\(0.50\). After the
first few epochs, the held-out real-data control remains above the
corrected synthetic-data curves and is more variable, but it supports
the same broad conclusion: the higher-capacity regimes expose
substantially more membership leakage than the linear head. Zero-Run is
therefore conservative in magnitude while still recovering the onset and
saturation of leakage during training. Each point is a separate audit of
a fixed checkpoint; the trajectories should not be interpreted as a
privacy-composition analysis across multiple checkpoint releases.

\begin{figure}[t]
    \centering
    \includegraphics[width=\linewidth]
        {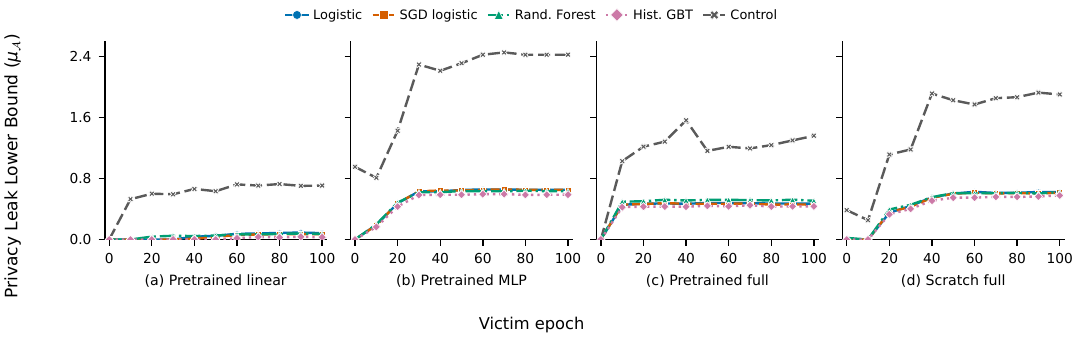}
    \caption{Conservative CIFAR-10 \(\mu\)-GDP lower bounds across 100 victim-training
    epochs. Colored curves use four propensity corrections with synthetic
    non-members; Control is the held-out real-data benchmark.}
    \label{fig:cifar10-audit-trajectories}
\end{figure}

\section{WikiText Experiments}
\label{app:wikitext}
\subsection{Synthetic Data Generation}

Following the PANORAMIA setup, we produced synthetic WikiText-103 data using a fine-tuned GPT-2-small generator. We trained this generator on 29,582 natural 64-token WikiText chunks, with another 3,697 chunks for validation. Training ran for 60 epochs with seed 42, and we selected the checkpoint with the lowest validation loss. We then used a separate set of 12,678 WikiText chunks as prompts. For each 64-token prompt, the generator sampled eight 64-token continuations using top-$k$ sampling with $k=200$, temperature $1$, and top-$p=1$. After removing the prompts, we retained 97,536 of the 101,424 generated suffixes because their decoded text re-tokenized to exactly 64 tokens.

Crucially, the generator was not the model being audited. The audited target was a separate GPT-2-small model trained in an independent run, with its own weights and checkpoint, on 95,085 WikiText chunks for 200 epochs. Thus, although the generator and target shared the same architecture and pretrained GPT-2 initialization, they were distinct model instances trained with different data configurations and for different purposes. The generator produced reference text, whereas the target model was evaluated for evidence of membership leakage.

Some of the natural data used to train the generator was also included in the target model's training population, as prescribed by the experimental data split. However, the generated 64-token suffixes themselves were not used to train the target model and had no exact sequence overlap with its natural training or held-out populations. Consequently, these suffixes served as synthetic nonmember examples: their behavior under the target model was compared with that of known training members to estimate membership-inference bounds.

\subsection{Victim Model Training Setup}

Following the WikiText-103 setup of PANORAMIA, we fine-tuned a pretrained GPT-2-small model (124.4M parameters) using standard causal language modeling. WikiText-103 documents were tokenized with the GPT-2 tokenizer and partitioned into fixed sequences of 64 tokens, each defining one membership unit. The victim training set contained 95,085 such sequences, while 10,565 sequences were held out for evaluation. Each sequence was used as both input and target, yielding 63 next-token predictions per example. No synthetic examples or differential-privacy mechanisms were used during victim training. We follow the $95{,}085/10{,}565$ partition stated in the paper; this differs from the released PANORAMIA data-module path, which produces an $84{,}520/10{,}565/10{,}565$ train/validation/test split, and this discrepancy should therefore be accounted for when reproducing the experiment.

Training was performed for 200 epochs with AdamW, a learning rate of $2\times10^{-5}$, weight decay $0.01$, and 100 warm-up steps. The per-GPU batch size was 64 across four GPUs, giving an effective batch size of 256 without gradient accumulation. This corresponds to 372 optimizer updates per epoch and 74,400 updates over the full training trajectory. Standard next-token cross-entropy was minimized, and model performance on the held-out data was evaluated once per epoch. All training randomness used seed 42. The victim was trained only once: the models audited at different levels of memorization are checkpoints extracted from this same continuous optimization trajectory rather than independently trained models.

The increase in held-out loss and perplexity at later checkpoints indicates progressively stronger overfitting, making these checkpoints useful for studying how membership leakage changes as training proceeds. During subsequent auditing, each checkpoint is kept fixed; no membership-inference model or auditing procedure participates in victim optimization.


\begin{figure}[H]
  \centering
  \begin{subfigure}[t]{0.485\textwidth}
    \vspace{0pt}
    \centering
    \includegraphics[width=\linewidth]{Figures/cifar10_synth_data/cifar10_bootstrap_auc.pdf}
    \caption{ROC AUC before and after calibration.}
    \label{fig:wikitext-propensity-auc}
  \end{subfigure}
  \hfill
  \begin{subfigure}[t]{0.485\textwidth}
    \vspace{0pt}
    \centering
    \includegraphics[width=\linewidth]{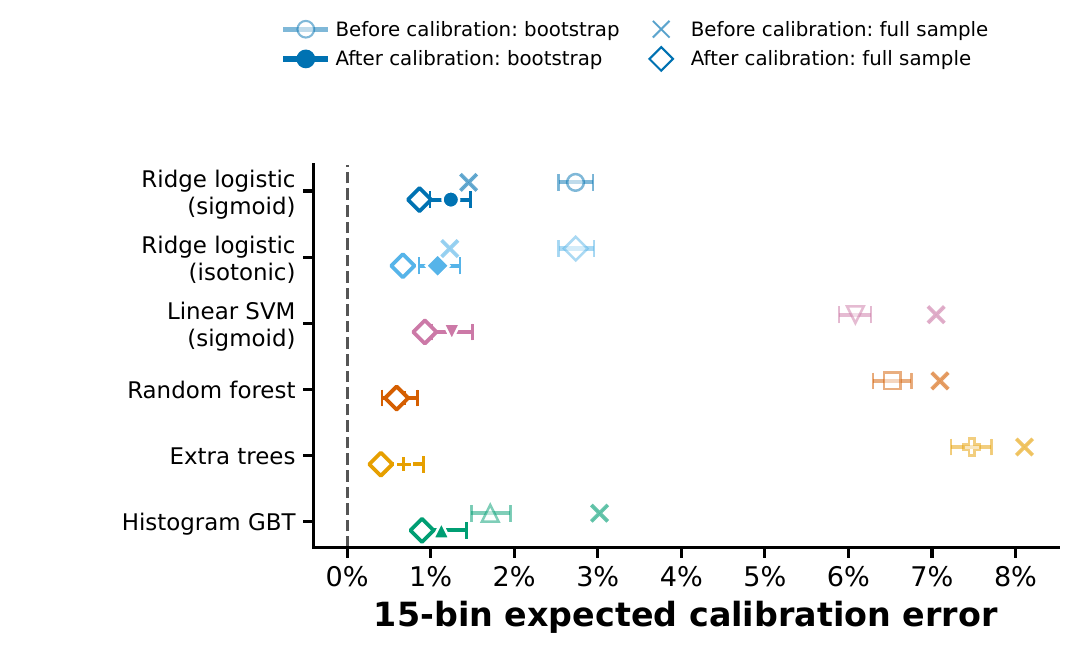}
    \caption{15-bin ECE before and after calibration.}
    \label{fig:wikitext-propensity-ece}
  \end{subfigure}

  \vspace{0.6em}

  \begin{subfigure}[t]{0.485\textwidth}
    \vspace{0pt}
    \centering
    \includegraphics[width=\linewidth]{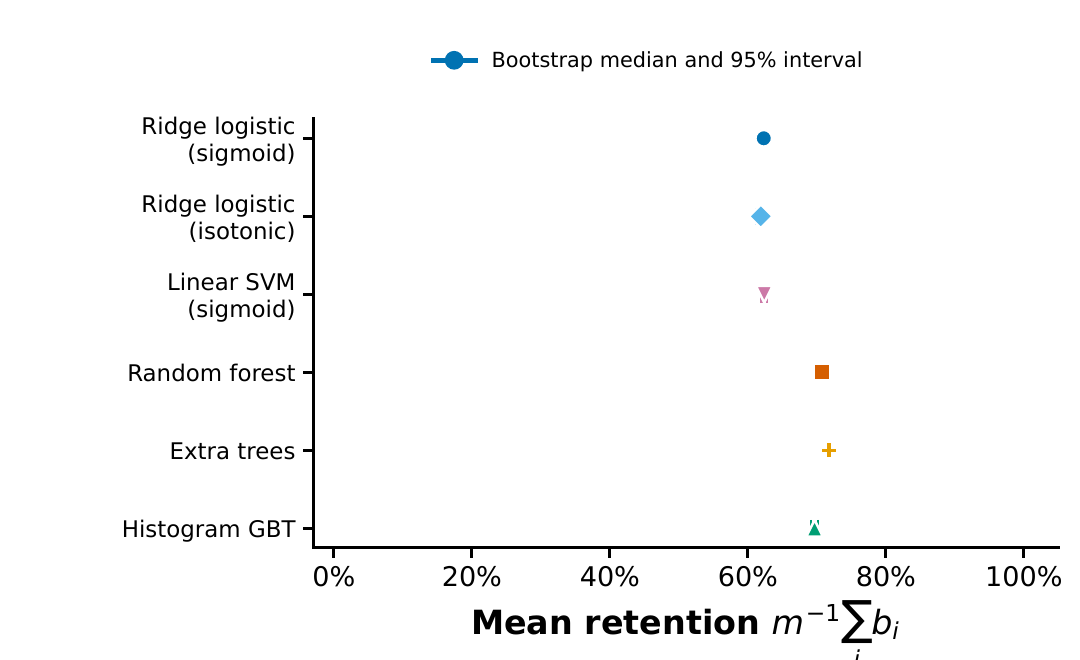}
    \caption{Mean f-DP retention probability.}
    \label{fig:wikitext-propensity-retention}
  \end{subfigure}
  \hfill
  \begin{subfigure}[t]{0.485\textwidth}
    \vspace{0pt}
    \centering
    \includegraphics[width=\linewidth]{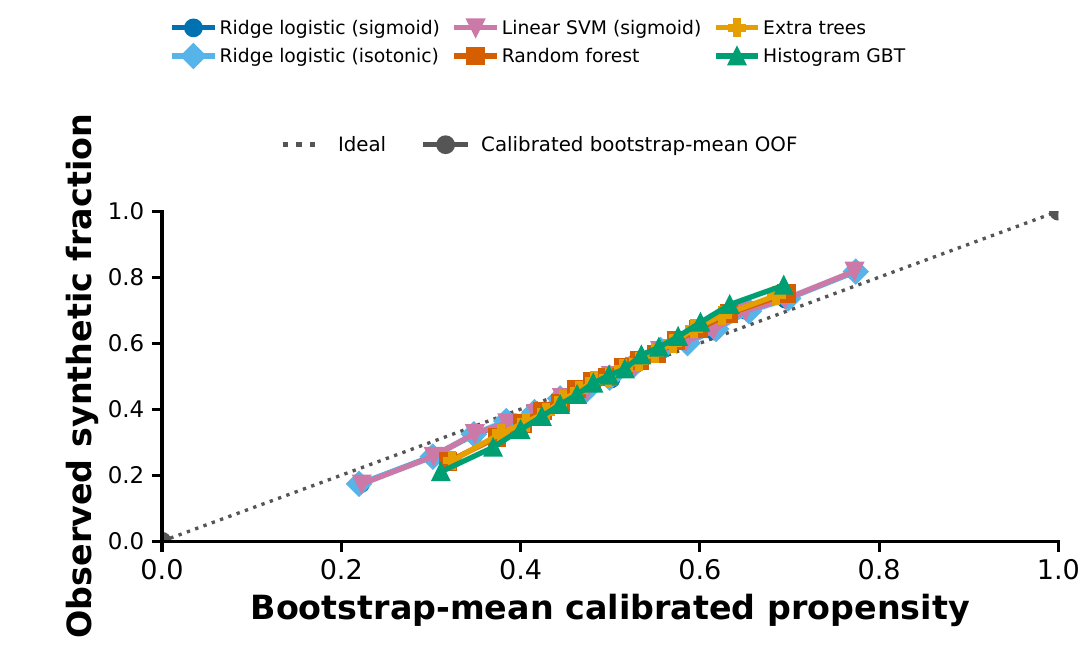}
    \caption{Calibrated bootstrap-mean reliability.}
    \label{fig:wikitext-propensity-reliability}
  \end{subfigure}

  \caption{Propensity-model diagnostics for GPT-2-generated suffixes versus
  real WikiText-103 embeddings, without PCA. Panels (a) and (b) compare each
  estimator before and after calibration using full-sample estimates and 95\%
  percentile intervals from 600 text-hash-group bootstrap refits. Panel (c)
  reports the calibrated f-DP rejection probabilities, where
  $b_i=\min(\hat\pi_i,1-\hat\pi_i)/\max(\hat\pi_i,1-\hat\pi_i)$ and
  $m^{-1}\sum_i b_i$ is the expected retained fraction. Panel (d) uses each
  text's bootstrap-mean calibrated out-of-fold propensity; the campaign did
  not retain the raw score vectors required for a pre-calibration reliability
  curve.}
  \label{fig:wikitext-propensity-summary}
\end{figure}

\subsection{Propensity Score Analysis}
\label{sec:wikitext-propensity-analysis}

\Cref{fig:wikitext-propensity-summary} repeats the propensity-model
sensitivity analysis for real WikiText-103 sequences and GPT-2-generated
suffixes. The reported ROC AUCs lie between approximately \(0.64\) and
\(0.68\), indicating a consistent but moderate distribution shift across
propensity-model classes. As in the CIFAR-10 experiment, calibration primarily
improves probability quality rather than discrimination: it substantially
reduces ECE, particularly for the SVM and tree-based estimators, while the
calibrated reliability curves remain close to the diagonal over the supported
propensity range.

The corresponding \(f\)-DP retention probabilities are stable across the
linear and nonlinear estimators, with approximately \(60\%\)--\(70\%\) of the
audit sample retained. The correction therefore removes a meaningful
feature-only membership signal without eliminating most of the available
evidence, and this conclusion is not driven by a single propensity
specification. Panel~(d) should be interpreted as a post-calibration diagnostic,
since only the bootstrap-mean calibrated out-of-fold scores were retained for
the reliability analysis.

\section{Failure Modes}
\subsection{The Enron email vs TREC spam case}
\label{subsec:enron-trec-propensity}

\begin{figure}[t]
    \centering
    \begin{minipage}[t]{0.48\columnwidth}
        \centering
        \includegraphics[width=\linewidth]{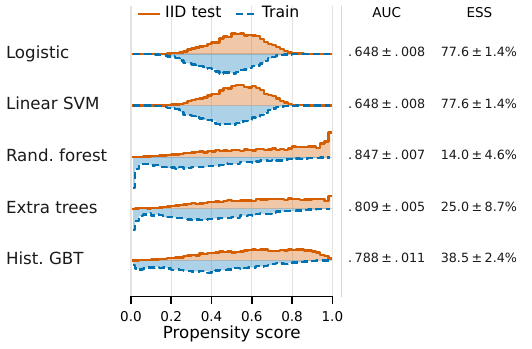}
        \caption{Performance of five propensity-score estimators in
distinguishing the iWildCam training set from the benchmark's nominally
IID test split. Although the two splits contain examples from the same
camera locations, they differ in their label composition. Entries
report five-fold means \(\pm\) std. Nonlinear estimators achieve AUCs
of up to \(0.847\), with effective sample sizes as low as \(14.0\%\),
showing that the nominally IID test set is not exchangeable with the
training set.}
        \label{fig:iwildcamera_train_test_iid}
    \end{minipage}
    \hfill
    \begin{minipage}[t]{0.48\columnwidth}
        \centering
        \includegraphics[width=\linewidth]
            {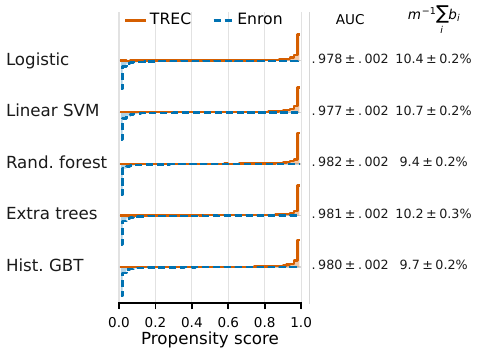}
        \caption{Performance of five propensity-score estimators in distinguishing
        TREC from Enron emails. Logistic denotes logistic regression with isotonic
        calibration. Entries report five-fold means \(\pm\) std.
        AUCs near \(0.98\) 
        indicate strong dataset separability and insufficient overlap.}
        \label{fig:enron-trec-main-separation}
    \end{minipage}
\end{figure}

We evaluated whether Enron emails provide adequate covariate support for the TREC spam corpus using their text embeddings. Because the corpora have different sizes, we sampled 15,223 Enron emails without replacement to match the 15,223 TREC observations. Within each of five outer folds, the embeddings were reduced to 48 principal components using only the corresponding training data; probability calibration was performed in three inner folds. All quantities below are therefore based on out-of-fold predictions, with no bootstrap resampling.

The candidate models produced consistently high discrimination, with ROC AUC between 0.977 and 0.982, and low calibration error, with ECE between 0.002 and 0.011 (Figure~\ref{fig:enron-trec-model-comparison}. The narrow variation across outer folds indicates that the separation is not attributable to a single split. Random forest achieved the largest AUC and smallest Brier score, whereas Histogram GBT retained the largest effective sample size and the smallest fraction of extreme scores. We therefore use Histogram GBT as a descriptive overlap model rather than claiming a universally best classifier.

High predictive performance did not imply adequate overlap. Even for Histogram GBT, the effective sample size was only 297 of 15,223 Enron observations (2.0\%), and 36.8\% of all scores were below 0.01 or above 0.99. Figure~\ref{fig:enron-trec-overlap-calibration} makes the distinction clear: the calibrated probabilities track the observed TREC rate, while the score distributions remain concentrated near opposite ends of the propensity range. Thus, calibration improves the interpretation of the probabilities but cannot recover support where the embedding distributions scarcely overlap.

Taken together, these results show a large  shift between Enron and TREC. Any unrestricted inverse-propensity analysis would place most of its effective weight on a very small subset of Enron emails and should therefore be regarded as unsupported without additional overlap restrictions or a narrower target population.

\begin{figure*}[t]
    \centering
    \begin{minipage}[t]{0.32\textwidth}
        \centering
        \includegraphics[width=\linewidth]{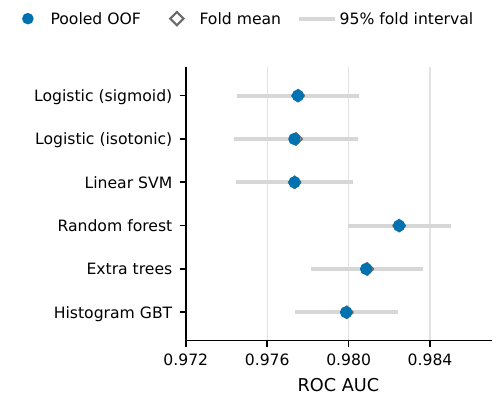}
        \smallskip
        {\small (a) ROC AUC}
    \end{minipage}\hfill
    \begin{minipage}[t]{0.32\textwidth}
        \centering
        \includegraphics[width=\linewidth]{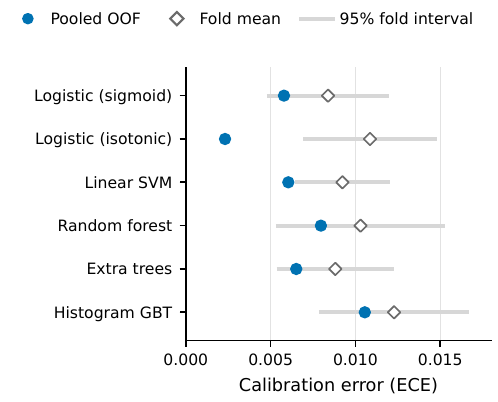}
        \smallskip
        {\small (b) Calibration error}
    \end{minipage}\hfill
    \begin{minipage}[t]{0.32\textwidth}
        \centering
        \includegraphics[width=\linewidth]{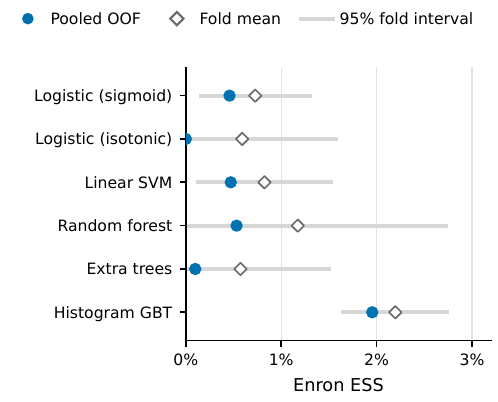}
        \smallskip
        {\small (c) Effective sample size}
    \end{minipage}
    \caption{Cross-fitted comparison of the propensity estimators. Circles denote pooled out-of-fold estimates; diamonds and horizontal intervals show the mean and 95\% $t$ interval across five outer folds. Strong and stable discrimination coexists with a very small effective sample size.}
    \label{fig:enron-trec-model-comparison}
\end{figure*}

\begin{figure*}[t]
    \centering
    \includegraphics[width=0.82\textwidth]{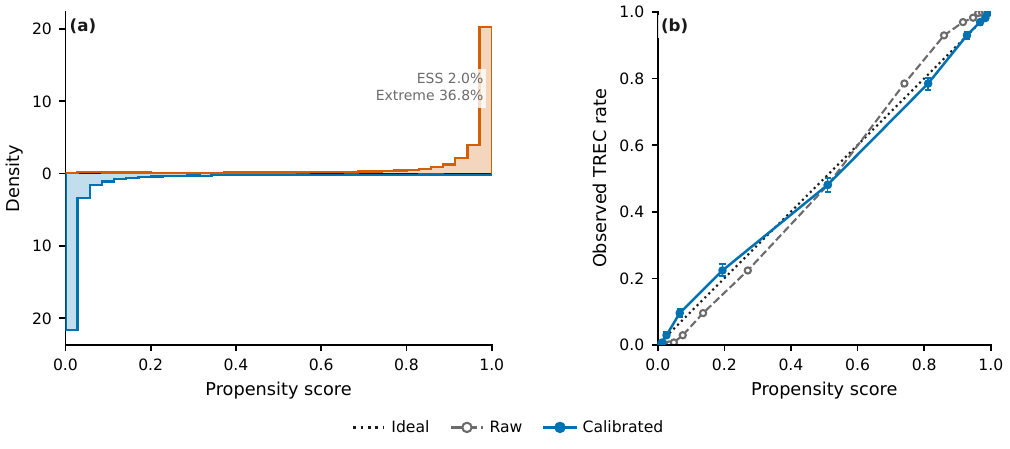}
    \caption{Out-of-fold overlap and calibration for Histogram GBT. The probabilities are well calibrated, but Enron and TREC remain strongly separated, producing low ESS and many extreme scores.}
    \label{fig:enron-trec-overlap-calibration}
\end{figure*}

\begin{figure*}[t]
    \centering
    \centering
    \includegraphics[width=0.5\linewidth]{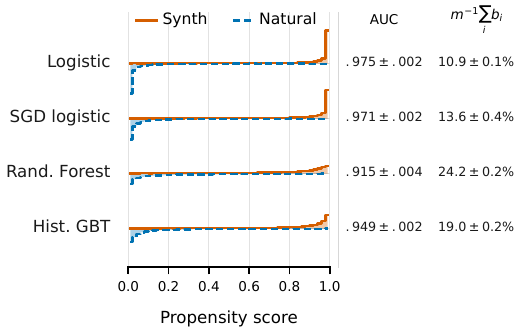}
    \smallskip
    \caption{Cross-fitted comparison of four propensity estimators on natural and PANORAMIA-generated CIFAR-10 images. Mirrored histograms show pooled out-of-fold propensity scores; columns report the mean AUC and retention $m^{-1}\sum_i b_i$, with $\pm$ one standard deviation across three outer folds. All estimators strongly separate synthetic from natural images, resulting in poor overlap and low retention.}
    \label{fig:panoramia-synth-cifar10}
\end{figure*}

\end{document}